\documentclass[preprint,12pt]{elsarticle}

\usepackage{amssymb}
\usepackage{amsmath,amsfonts}
\usepackage{algorithmic}
\usepackage{array}
\usepackage[caption=false,font=normalsize,labelfont=sf,textfont=sf]{subfig}
\usepackage{textcomp}
\usepackage{stfloats}
\usepackage{float}
\usepackage{url}
\usepackage{verbatim}
\usepackage{graphicx}
\def\BibTeX{{\rm B\kern-.05em{\sc i\kern-.025em b}\kern-.08em
    T\kern-.1667em\lower.7ex\hbox{E}\kern-.125emX}}
\usepackage{balance}

\usepackage{booktabs}                        
\usepackage{subcaption}   
\usepackage{color}
\usepackage{amsmath}
\usepackage{amssymb}
\usepackage{graphicx}
 \usepackage{dsfont}
\usepackage{amsfonts}
\usepackage{amsthm} 
\usepackage{stmaryrd}
\usepackage[all]{xy}
\usepackage{multirow}

\usepackage{qcircuit}
 
\def\>{\ensuremath{\rangle}}
\def\<{\ensuremath{\langle}}

\newcommand{\hs}{\mathcal{H}}

\newcommand {\tr} {{\mathit{tr}}}

\newtheorem{thm}{Theorem}[section]

\newtheorem{lem}{Lemma}[section]
\newtheorem{defn}{Definition}[section]
\newtheorem{prop}{Proposition}[section]
\newtheorem{exam}{Example}[section]

\newtheorem{rem}{Remark}[section]

\journal{}

\begin{document}

\begin{frontmatter}

\title{A Practical Quantum Hoare Logic with Classical Variables, I}

\author{Mingsheng Ying}

\affiliation{organization={Centre for Quantum Software and Information, University of Technology Sydney}, 
           addressline={15 Broadway}, 
           city={Ultimo},
           postcode={2007}, 
          state={NSW},
            country={Australia}}

\begin{abstract}
In this paper, we present a Hoare-style logic for reasoning about quantum programs with classical variables. Our approach offers several improvements over previous work:
\begin{enumerate} \item \textbf{\textit{Enhanced expressivity of the programming language}}: Our logic applies to quantum programs with classical variables that incorporate quantum arrays and parameterized quantum gates, which have not been addressed in previous research on quantum Hoare logic, either with or without classical variables. 
\item \textbf{\textit{Intuitive correctness specifications}}: In our logic, preconditions and postconditions for quantum programs with classical variables are specified as a pair consisting of a classical first-order logical formula and a quantum predicate formula (possibly parameterized by classical variables). These specifications offer greater clarity and align more closely with the programmer's intuitive understanding of quantum and classical interactions.

\item \textbf{\textit{Simplified proof system}}: By introducing a novel idea in formulating a proof rule for reasoning about quantum measurements, along with (2), we develop a proof system for quantum programs that requires only minimal modifications to classical Hoare logic. Furthermore, this proof system can be effectively and conveniently combined with classical first-order logic to verify quantum programs with classical variables. 
\end{enumerate}
As a result, the learning curve for quantum program verification techniques is significantly reduced for those already familiar with classical program verification techniques, and existing tools for verifying classical programs can be more easily adapted for quantum program verification. 
\end{abstract}


\end{frontmatter}

\section{Introduction}\label{sec-intro}
Hoare logic has played a foundational role in classical programming methodology and program verification techniques. With the advent of quantum computing, there has been a natural desire to establish Hoare-style logic for quantum programs \cite{Valiron1, Valiron2, Ying24, Zuliani}. 

{\vskip 3pt}

\textbf{Quantum Hoare logic and assertion logic}: A program logic is usually built upon an assertion logic, which is used to specify the pre/postconditions of a program that describe the properties of program variables before and after its execution. So, the first step of defining quantum Hoare logic (QHL for short) is to choose an assertion language for describing the properties of quantum states (i.e. the states of quantum program variables). The natural choice for this purpose should be the standard quantum logic of Birkhoff and von Neumann \cite{BvN36}, in which a proposition about a quantum system is modelled as a closed subspace of (equivalently, projection operator on) the system's Hilbert space. Interestingly, the earliest approaches to QHL did not adopt it: \begin{itemize}\item Chadha, Mateus and Sernadas~\cite{CMS} presented a proof system for reasoning about quantum programs in which the assertion language is exogenous quantum logic \cite{Mat06}, where a superposition of classical semantic models is conceived as a semantic model of quantum logic. 
\item Kakutani \cite{Kaku09} proposed an extension of Hartog's probabilistic Hoare logic \cite{Hartog} for reasoning about quantum programs, in which probabilistic assertion logical formulas are generalised with unitary and measurement operators.  
\end{itemize}

It was first proposed by D'Hondt and Panangaden \cite{DP06} to use a special class of observables as assertions (e.g. pre/postconditions) for quantum programs, called quantum predicates. This approach enjoys an elegant physical interpretation for Hoare triples (correctness formulas) of quantum programs. On the other hand, it aligns with Birkhoff-von Neumann quantum logic. The observables used in \cite{DP06} are mathematically modelled as Hermitian operators between the zero and identity operators. They are indeed called quantum effects in the quantum foundations literature, and are considered as propositions in unsharp quantum logic \cite{unsharp} - an extension of Birkhoff-von Neumann logic (note that projection operators are a special kind of quantum effects). Afterwards,  some useful proof rules for reasoning about quantum programs were introduced by Feng et. al ~\cite{FDJY07} using quantum predicates as pre/postconditions.   

A main limitation of all proof systems for quantum programs mentioned above is that they lack (relative) completeness. The first QHL with (relative) completeness was established in~\cite{Ying11} using the notion of pre/postconditions as observables proposed in  \cite{DP06}. 

Research on QHL and related problems has become active in the last few years. Various variants and extensions of QHL have been proposed, including relational QHL \cite{Unruh19a, Barthe, Yangjia, Yu-cav}, QHL with ghost variables \cite{Unruh19}, QHL restricted to projective quantum predicates \cite{Zhou19}, quantum separation logic \cite{QSL1, QSL2}, local and modular reasoning \cite{Deng1, KarS}, dynamic quantum logic \cite{Tak}, incorrectness logic for quantum programs \cite{Yu-inc}, quantum temporal logic \cite{Yu-temp}, refinement calculus for quantum programs \cite{QbC, Feng-refine}, and verification of recursive quantum programs \cite{Xu}, parallel quantum programs \cite{Ying18} and nondeterministic quantum programs \cite{Feng-nd}. Also, a series of  verification tools for quantum programs based on QHL and other logics have been implemented, including QBricks \cite{Valiron0}, QHLProver \cite{Liujy}, CoqQ \cite{Coqq}, Qafny \cite{LiL}, verification of Shor's algorithm in Coq \cite{Peng, PengYX}, and implementation of dynamic quantum logic in Maude \cite{Do} as well as tools for verification of quantum compilers, e.g. VOQC \cite{Rand, Hie, Hie1}, Giallar \cite{Tao, Tao1, Shi} and VQO \cite{LiLY}. For more detailed discussions about research in this area, the readers can consult surveys \cite{Valiron1, Zuliani} and book \cite{Ying24}. 

{\vskip 3pt}
 
\textbf{Classical variables}: But the QHL in \cite{Ying11} also has a limitation; namely, it is defined only for purely quantum programs without classical variables. Theoretically, such QHL is strong enough because all classical computations can be efficiently simulated by quantum computations. For practical applications, however, a QHL with classical variables is more convenient because most of the known quantum algorithms involve classical computations.  

Indeed, the Hoare-style proof systems in \cite{CMS, Kaku09} and several others proposed recently allows classical variables, but all of them do not have (relative) completeness. A (relatively) complete QHL with classical variables was first presented in \cite{FengY} by generalising the techniques of \cite{Ying11}. The key idea of \cite{FengY} is to define a classical-quantum state (i.e. a state of a quantum program with classical variables; cq-state for short) as a mapping from classical states to (mixed) quantum states. Accordingly, a classical-quantum predicate (i.e. an assertion for a quantum program with classical variables; cq-predicate for short) is defined as a mapping from classical states to observables. This design decision was also adopted in \cite{Deng} for developing a satisfaction-based Hoare logic for quantum programs with classical variables, and further extended in \cite{Feng-D} for reasoning about distributed quantum programs with classical variables.

A benefit of the view of cq-states and cq-predicates in \cite{FengY} is that the QHL for purely quantum programs in \cite{Ying11} can be straightforwardly generalised to a Hoare logic for quantum programs with classical variables \cite{FengY} by a point-wise lifting from quantum predicates to cq-predicates. However, it has a drawback that traditional logic tools (e.g. classical first-order logic) cannot be conveniently used to specify and reason about the classical states in a cq-state or a cq-predicate because they appear in the domain of the cq-state or cq-predicate (as a mapping) and often must be treated point-wise.

{\vskip 3pt}

\textbf{Contributions of this paper}: This paper aims at developing a QHL for quantum programs with classical variables that overcomes the drawback of \cite{FengY} pointed out above, but inherits the advantages of quantum predicates as observables in \cite{DP06, Ying11}. To this end, we first return to the view of cq-state as a pair $(\sigma,\rho)$ of a classical state $\sigma$ and a quantum state $\rho$ adopted in \cite{CMS, Kaku09}. But different from \cite{CMS, Kaku09} as well as \cite{FengY}, we define an assertion for a quantum program with classical variables (a \textit{cq-assertion} for short) as a pair $(\varphi,A)$ of a classical first-order logical formula $\varphi$ and a quantum predicate (i.e. an observable) $A$. It should be pointed out that the classical part $\sigma$ and $\varphi$ and the quantum part $\rho$ and $A$ are not independent, as will be explained shortly. Intuitively, $\sigma$ and $\rho$ are the states of classical and quantum registers, and $\varphi$ and $A$ describes the properties of classical and quantum registers, respectively.    
Obviously, this is the most nature view of cq-states and cq-assertions. \begin{enumerate}
\item  \textbf{\textit{More practical proof system}}: Why has no QHL based on the above view of cq-assertions been given in the previous literature? The main reason seems that Hoare triples with such cq-assertions cannot properly cope with the probability distributions arisen from quantum measurements. Indeed, exogenous quantum logic employed in \cite{CMS}, probabilistic assertions used in \cite{Kaku09} and defining cq-predicates as mappings from classical states to quantum predicates \cite{FengY} are all for handling these probability distributions. In this paper, however, we find a novel (and simple) proof rule for reasoning about quantum measurements (see (Axiom-Meas) in Table \ref{proof-system}). Incorporating it with the proof rules for other program constructs, we are able to build a QHL for quantum programs with classical variables that is much more convenient and easier to use in practical applications than the previous work for the same purpose. 
 
\item \textbf{\textit{Syntax of quantum predicates}}: Quantum predicates in \cite{DP06, Ying11, FengY} and other research along this line are defined as observables. Inherited from physics, an observable is modelled as a Hermitian operator on the Hilbert space of the system under consideration, which is essentially an semantic entity from the logical point of view. For example, a quantum predicate $A$ for $n$ qubits is a $2^n\times 2^n$ matrix. Its dimension exponentially increases with the number $n$ of qubits. This makes verification of large quantum programs based on QHL impractical because a huge amount of matrix calculations will be unavoidable. The author of \cite{Ying22} proposed an expansion of Birkhoff-von Neumann quantum logic as the assertion language of QHL for purely quantum programs, in which quantum predicates can be symbolically and compactly represented; in particular, by introducing their syntax, each quantum predicate $A$ can be constructed from atomic ones using logical connectives in a way similar to that in classical first-order logic (also see Sections 7.5 and 7.6 of \cite{Ying24}).  
In this paper,  the syntax of quantum predicates defined in \cite{Ying22, Ying24} is tailored and further expanded for specifying and reasoning about quantum programs with classical variables. Thus, in a cq-assertion $(\varphi,A)$, $\varphi$ is a classical first-order logical formula, and $A$ can be written as a quantum predicate formula, which is usually more compact and economic than a large matrix representation.    

\item \textbf{\textit{Stronger expressive power of programming language}}: The quantum programming languages on which Hoare-style proof systems with classical variables were established in the previous literature (e.g. \cite{CMS, Kaku09, FengY}) are relatively simple. In particular, the only connection between their classical and quantum parts is a command of the form $x:=M[\overline{q}]$ in which the outcome of a measurement $M$ performed on quantum register $\overline{q}$ is stored in a classical variable $x$. In this paper, quantum arrays and  parameterised quantum gates are introduced into our programming language. Then there will be more connections from the classical part to the quantum part: classical terms (and thus classical variables) are used as the subscripts of quantum array variables, and quantum gates are parameterised by classical terms. Consequently, the framework of our QHL must be meticulously designed to effectively accommodate quantum arrays and parameterised quantum gates. Fortunately, the cq-assertions discussed earlier provide an appropriate logical tool for addressing this issue.
 \end{enumerate}

In summary, we expect that the above contributions can significantly improve the applicability of QHL in the practical verification of quantum programs with classical variables.  

{\vskip 3pt}

\textbf{Organisation of the paper}: The exposition of our QHL with classical variables is divided into two parts, with this paper serving as the first instalment. Here, we establish the framework and introduce the proof rules of QHL, as well as demonstrate the soundness of the logic. However, we defer the discussion of (relative) completeness and its proof to a companion paper \cite{Ying24+}, as these topics require more extensive theoretical treatment and analysis. This division allows readers primarily interested in the practical applications of QHL to focus on this paper without needing to engage with the more theoretical aspects presented in the second part.

This paper is organised as follows. Quantum arrays are defined in Section \ref{sec-arrays}. We introduce our quantum programming language with classical variables and parameterised quantum gates in Section \ref{sec-lang}. The language of cq-assertions and its semantics are defined in Section \ref{sec-assert}. The axioms and proof rules of our QHL for quantum programs with classical variables are presented in Section \ref{sec-logic}. The soundness theorem of the QHL is stated in Section \ref{section-sound}. A brief conclusion is drawn in Section \ref{sec-concl}. Since the main aim of this paper is to introduce the conceptual framework of QHL, all proofs are deferred to the Appendix. 
 
\section{Quantum Arrays and Subscripted Quantum Variables}\label{sec-arrays}

Quantum arrays have already been introduced in several quantum programming languages (e.g. OpenQASM \cite{openqsam} and Q\# \cite{qsharp}). In this section, we recall from \cite{YingZ23, YingZ24} the formal definitions of quantum arrays and subscripted quantum variables, which will be needed in defining our programming language in the next section.    

\subsection{Quantum types} As a basis, let us first define the notion of quantum type. Recall that a basic classical type $T$ denotes an intended set of values. Similarly, a basic quantum type $\hs$ denotes an intended Hilbert space. It will be considered as the state space of a simple quantum variable. Throughout this paper, we adopt the following notions:\begin{itemize}\item \textit{Tensor product types}: if $\hs_1,...,\hs_n$ $(n>1)$ are basic quantum types, then their tensor product $\hs_1\otimes ... \otimes\hs_n$ is a quantum type.
\item \textit{Higher quantum types}: $T_1\times ...\times T_n\rightarrow\hs$, 
where $T_1,...,T_n$ are basic classical types, and $\hs$ is a basic quantum type. Mathematically, this type denotes the following tensor power of Hilbert space $\hs$ (i.e. tensor product of multiple copies of $\hs$):
\begin{equation}\label{h-space}\hs^{\otimes (T_1\times ...\times T_n)}=\bigotimes_{v_1\in T_1,...,v_n\in T_n}\hs_{v_1,...,v_n}\end{equation} where $\hs_{v_1,...,v_n}=\hs$ for all $v_1\in T_1,...,v_n\in T_n$. 
\end{itemize}

Intuitively, if we have $n$ quantum systems $q_1,...,q_n$ and their state spaces are $\hs_1,...,\hs_n$, respectively, then according to the basic postulates of quantum mechanics, the composite system of $q_1,...,q_n$ has $\hs_1\otimes ...\otimes\hs_n$ as its state space. If $\hs$ is the state space of a quantum system $q$, then the Hilbert space (\ref{h-space}) is the state space of the composite system consisting of those quantum systems indexed by $(v_1,...,v_n)\in T_1\times ...\times T_n$, each of which has the same state space as $q$. 
There is a notable basic difference between  the product of classical types and the tensor product of quantum types: there are not only product states of the form $|\psi_1\rangle ...|\psi_n\rangle=|\psi_1\rangle\otimes ...\otimes |\psi_n\rangle$ in $\hs_1\otimes ...\otimes\hs_n$ with $|\psi_i\rangle\in\hs_i$ $(1\leq i\leq n)$ but also entangled states that cannot be written as a tensor product of respective states in $\hs_i$'s. A similar difference between higher classical types and quantum ones is that some states in the space (\ref{h-space}) are entangled between the quantum systems indexed by different tuples $(v_1,...,v_n)\in T_1\times ...\times T_n$. It should be noted that (\ref{h-space}) may be an infinite tensor product of Hilbert spaces \cite{vN39} when one of $T_1,...,T_n$ is infinite (e.g., the type $\mathbf{integer}\rightarrow\hs_2$). However, most applications only need finite tensor products, as shown in the following:

\begin{exam}\label{exam-array} Let $\hs_2$ be the qubit type denoting the $2$-dimensional Hilbert space. If $q$ is a qubit array of type $\mathbf{integer}\rightarrow\hs_2$, where $\mathbf{integer}$ is the (classical) integer type, then for any two integers $k\leq l$, section $q[k:l]$ stands for the restriction of $q$ to the interval $[k:l]=\{{\rm integer}\ i|k\leq i\leq l\}.$ The state space of this section is $\hs_2^{\otimes (l-k+1)}$. In particular, if $k=l$, then $q[k:l]$ actually denotes single qubit $q[k]$. Furthermore,  
the GHZ (Greenberger–Horne–Zeilinger) state $$\frac{\bigotimes_{i=k}^l|0\rangle_i+\bigotimes_{i=k}^l|1\rangle_i}{\sqrt{2}}$$ is an entangled state of the qubits labelled $k$ through $l$.\end{exam}

\begin{exam}\label{QFT-data} Let us consider the data structure used in quantum Fourier transform on $n$ qubits. We introduce a classical bit array $j$ of type $\mathbf{integer}\rightarrow\mathbf{Boolean}$. For any integers $k,l$ with $1\leq k\leq l\leq n$, we use $0.j[k:l]$ to denote the binary representation $$0.j_k ...j_l=\sum_{r=k}^l j[r]\cdot 2^{k-r-1}.$$ Then we can define quantum arrays $|(j,k:l)\rangle$ and $|\mathit{QFT}(j,k:l)\rangle$ as states of qubits $q[k:l]$ by induction on the length $k-l$ of the arrays:  
\begin{align*}&\begin{cases}|(j,k:l)\rangle=|j[l]\rangle\ \mbox{if}\ k=l,\\ |(j,k:r+1)\rangle=|(j,k:r)\rangle\otimes|j[r+1]\rangle\ \mbox{for}\ k\leq r\leq l-1;\end{cases}\\ 
&\begin{cases}|\mathit{QFT}(j,k:l)\rangle=\frac{1}{\sqrt{2}}\left(|0\rangle+e^{2\pi i 0.j[l]}|1\rangle\right)\ \mbox{if}\ k=l,\\ 
|\mathit{QFT}(j,k:r+1)\rangle=|\mathit{QFT}(j,k:r)\rangle\otimes\frac{1}{\sqrt{2}}\left(|0\rangle+e^{2\pi i 0.j[l-r:l]}|1\rangle\right)\ \mbox{for}\ k\leq r\leq l-1. 
\end{cases}
\end{align*} By equation (5.4) in textbook \cite{NC00}, the quantum Fourier transform can be rewritten as: 
$$|(j,1:n)\rangle\mapsto |\mathit{QFT}(j,1: n)\rangle.$$
\end{exam}

\subsection{Quantum variables}\label{sec-qvar}

Only simple quantum variables are allowed in previous QHL (e.g. \cite{Ying11, FengY}). In this paper, however, we will use two sorts of quantum variables:\begin{itemize}\item simple quantum variables, of a basic quantum type, say $\hs$; \item array quantum variables, of a higher quantum type, say $T_1\times ...\times T_n\rightarrow\hs$.  
\end{itemize}
Furthermore, we introduce:
\begin{defn}\label{def-subscripted}Let $q$ be an array quantum variable of the type $T_1\times ...\times T_n\rightarrow\hs$, and for each $1\leq i\leq n$, let $s_i$ be a classical expression of type $T_i$. Then $q[s_1,...,s_n]$ is called a subscripted quantum variable of type $\hs$.
\end{defn}

Intuitively, array variable $q$ denotes a quantum system composed of subsystems indexed by tuples $(v_1,...,v_n)\in T_1\times ...\times T_n$. Thus, whenever expression $s_i$ is evaluated to a value $v_i\in T_i$ for each $i$, then $q[s_1,...,s_n]$ indicates the system of index tuple $(v_1,...,v_n)$. For example, let $q$ be a qubit array of type $\mathbf{integer}\times\mathbf{integer}\rightarrow\hs_2$. Then $q[2x+y,7-3y]$ is a subscripted qubit variable; in particular, if $x=5$ and $y=-1$ in the current classical state, then it stands for the qubit $q[9,10]$; i.e. the qubit in the memory cell with coordinates $(9,7)$. To simplify the presentation, we often treat a simple variable $q$ as a subscripted variable $q[s_1,...,s_n]$ with $n=0$. 

It is usually required that quantum operations are performed on a tuple of distinct quantum variables; for example, a controlled-NOT gate $\mathit{CNOT}[q_1,q_2]$ is meaningful only when $q_1\neq q_2$. To describe this kind of conditions, 
let $\overline{q} = q_1[s_{11},...,s_{1n_1}],...,q_k[s_{k1},...,s_{kn_k}]$ be a sequence of simple or subscripted quantum variables. Then we define the distinctness formula $\mathit{Dist}(\overline{q})$ as the following first-order logical formula:
\begin{equation}\label{distinctness}\mathit{Dist}(\overline{q})=(\forall i,j\leq k)\left[(q_i=q_j)\rightarrow (\exists l\leq n_i)\left(s_{il}\neq s_{jl}\right)\right].\end{equation}
Thus, for a given classical structure as an interpretation of the first-order logic and a classical state $\sigma$ in it, the satisfaction relation: $$\sigma\models\mathit{Dist}(q_1[t_{11},...,s_{1n_1}],...,q_k[s_{k1},...,s_{kn_k}])$$ means that in the classical state $\sigma$, $q_1[s_{11},...,s_{1n_1}],$ $...,q_k[s_{k1},...,s_{kn_k}]$ denotes $k$ distinct quantum systems. This condition will be frequently used in this paper.  

\section{A Quantum Programming Language with Classical Variables}\label{sec-lang}

In this section, we define a language $\mathit{qWhile}^+$ for quantum programming with classical variables. It is a further extension of the purely quantum variant $\mathit{qWhile}$ (see for example, Chapter 5 of \cite{Ying24}) of the classical $\mathit{While}$-language (see for example, Chapter 3 of \cite{Apt09}). In particular, $\mathit{qWhile}^+$ allows to use quantum arrays and parameterised quantum gates. 

\subsection{Parameterised Quantum Gates}\label{sec-gate}

Various quantum gates with certain parameters are implemented in quantum hardware. In addition, parameterised quantum gates and circuits are widely used in variational quantum algorithms \cite{Varia} and quantum machine learning \cite{QML}. Consequently, almost all quantum programming languages support parameterised quantum gates. In our language $\mathit{qWhile}^+$, we assume a set $\mathcal{U}$ of basic parameterised quantum gate symbols. Each gate symbol $U\in\mathcal{U}$ is equipped with: \begin{enumerate}\item[(i)] a classical type of the form $T_1\times ...\times T_m$, called the parameter type, where $m\geq 0$; and \item[(ii)] a quantum type of the form $\hs_1\otimes ...\otimes\hs_n$, where $n\geq 1$ is called the arity of $U$.\end{enumerate} 
In the case of $m=0$, $U$ is a non-parameterised gate symbol.  
If for each $k\leq m$, $t_k$ is a classical expression of type $T_k$, and for each $i\leq n$, $q_i$ is a simple or subscripted quantum variable with type $\hs_i$, then $$U(t_1,...,t_m)[q_1,...,q_n]$$ is called a quantum gate. Intuitively, it means that the quantum gate $U$ with parameters $t_1,...,t_m$ acts on quantum variables $q_1,...,q_n$. The following are two examples of parameterised basic quantum gates:
\begin{exam} 
\begin{enumerate}\item The rotation about $x$-axis is defined as 
\begin{align*}R_x(\theta)=\exp(-iX\theta/2)=\left(\begin{array}{cc}\cos\frac{\theta}{2} & -i\sin\frac{\theta}{2}\\ -i\sin\frac{\theta}{2} & \cos\frac{\theta}{2}\end{array}\right)
\end{align*} where parameter $\theta$ denotes the rotation angle. Let $q$ be a qubit variable. Then $R_x(3\theta+\frac{\pi}{2})[q]$ stands for the rotation of angle $3\theta+\frac{\pi}{2}$ about $x$ applied to qubit $q$.  
\item The $XX$ interaction is a two-qubit interaction defined as 
\begin{align*}R_{xx}(\theta)=\exp(-i\frac{\theta}{2}X\otimes X)=\left(\begin{array}{cccc}\cos\frac{\theta}{2} & & & -i\sin\frac{\theta}{2}\\  &\cos\frac{\theta}{2} & -i\sin\frac{\theta}{2} & \\  & -i\sin\frac{\theta}{2} &\cos\frac{\theta}{2} &\\ -i\sin\frac{\theta}{2} & & &\cos\frac{\theta}{2} \end{array}\right)
\end{align*} where parameter $\theta$ is an index of entangling. Let $q$ be a qubit array. Then $R_{xx}(\theta)[q[m+3],q[2m-7]]$ stands for the $XX$ interaction between qubits $q[m+3]$ and $q[2m-7]$.  
\end{enumerate} 
\end{exam}

\subsection{Quantum Measurements}\label{sec-measure}

In the language $\mathit{qWhile}^+$, we also assume a set $\mathcal{M}$ of basic quantum measurement symbols. Each measurement symbol $M\in\mathcal{M}$ is equipped with: \begin{enumerate}\item[(i)] a quantum type of the form $\hs_1\otimes...\otimes\hs_n$, where $n\geq 1$ is called the arity of $M$; and \item[(ii)] a classical type $T$, called the outcome type of $M$.\end{enumerate} If for each $i\leq n$, $q_i$ is a quantum variable with type $\hs_i$ and $x$ a classical variable with type $T$, then we will use $$x: =M[q_1,...,q_n]$$ to express the basic command that the measurement denoted by symbol $M$ is performed on quantum variables $q_1,...,q_n$, and its outcome is stored in classical variable $x$.  

\begin{rem}Indeed, we can also introduce parameterized measurements. We choose not to do it in this paper because it will make the notations too complicated. At this moment, various quantum algorithms mainly use the measurements in the computational basis, and parameterized measurements are rarely employed in practical applications. But when needed in the future, the parameterization of measurements can be handled in a way similar to the parameterization of quantum gates defined in the above subsection.    
\end{rem}

\subsection{Syntax of $\mathit{qWhile}^+$}\label{sec-qwhile}

Now we are ready to define our quantum programming language $\mathit{qWhile}^+$. Essentially, it is the classical while-language expanded by quantum gates and quantum measurement commands introduced in the previous subsections. 
The alphabet of $\mathit{qWhile}^+$ consists of: \begin{enumerate}\item[(i)] the alphabet of a classical while-language (see for example Chapter 3 of \cite{Apt09} and Section 3.2 of \cite{LSie});
\item[(ii)] a countably infinite set $\mathcal{QV}$ of (simple and array) quantum variables (as described in Subsection \ref{sec-qvar}); \item[(iii)] a set $\mathcal{U}$ of basic (parameterized) quantum gate symbols (as described in Subsections \ref{sec-gate}); and \item[(iv)] a set $\mathcal{M}$ of basic quantum measurement symbols (as described in Subsection \ref{sec-measure}).\end{enumerate}
 The language $\mathit{qWhile}^+$ is built upon the classical while-language given in (i). For simplicity of the presentation, we omit here the standard syntax of classical expressions and programs; the reader can consult textbook \cite{Apt09, LSie}. Then we can introduce:
\begin{defn} The syntax of $\mathit{qWhile}^+$ is given as follows: 
\begin{equation}\label{syntax}\begin{split}P::=\ & \mathbf{skip}\ |\ x:=e\ |\ q:=|0\rangle\\ &|\ U(t_1,...,t_m)[q_1,...,q_n]\ |\ x:=M[q_1,...,q_n]\\ &|\ P_1;P_2\ |\ 
\mathbf{if}\ b\ \mathbf{then}\ P_1\ \mathbf{else}\ P_0\ |\ \mathbf{while}\ b\ \mathbf{do}\ P
\end{split}\end{equation} where $x$ is a classical variable, $e$ and $t_1,...,t_m$ are classical expressions, $q, q_1,...,q_n$ are simple or subscripted quantum variables (see their definitions at the beginning of Subsection \ref{sec-qvar} and in Definition \ref{def-subscripted}), $U\in\mathcal{U}$, $M\in\mathcal{M}$, and $b$ is a classical boolean expression. \end{defn}

A simpler version of the language $\mathit{qWhile}^+$ was already introduced in the previous literature (see for example \cite{YingF, Deng, FengY}), but here we add parameterised gate $U(t_1,...,t_n)$ as well as quantum array into it so that $q,q_1,...,q_n$ in (\ref{syntax}) can be (subscripted) array variables. On the other hand, as mentioned in Section \ref{sec-arrays}, quantum arrays have been already introduced in previous quantum programming languages (e.g. OpenQASM \cite{openqsam} and Q\# \cite{qsharp}). Parameterised quantum gates are also supported by previous quantum programming languages, and they are recently generalised to differential quantum programming \cite{Xiaodi, Wang}. 

Intuitively, $q:=|0\rangle$ means that quantum variable $q$ is initialised in a basis state $|0\rangle$. The quantum gate $U(t_1,...,t_m)[q_1,...,q_n]$ and measurement command $x:=M[q_1,...,q_n]$ were already explained in Subsections \ref{sec-gate} and \ref{sec-measure}, respectively. The meanings of other commands are the same as in a classical while-language. It is worth examining the interface between the classical and quantum parts of $\mathit{qWhile}^+$. In the quantum to classical direction, the outcome of a quantum measurement $M[q_1,...,q_n]$ is stored in a classical variable $x$, which can appear in a classical expression $e$ and thus participate in a classical computation, or in a classical boolean expression $b$ and thus is used in the control of subsequent computation. In the classical to quantum direction, a classical boolean expression $b$ may be used in determining the control flow of quantum computation. In addition, classical variables may appear in the subscripts of some quantum array variables in $U(t_1,...,t_m)[q_1,...,q_n]$ or $M[q_1,...,q_n]$ (see Definition \ref{def-subscripted}) and thus determine the locations that quantum gate $U$ or measurement $M$ performs on.  

\begin{exam}\label{QFT-program} Let us consider an efficient circuit for the quantum Fourier transform on $n$ qubits $q_1,...,q_n$ presented in Figure 5.1 of textbook \cite{NC00}. It can be equivalently written as the program in Table \ref{QFT-table}, 
\begin{table*}[t]
\begin{align*}\mathit{QFT}[q_1,...,q_n]::=\ &H[q_1]; C(R_2)[q_2,q_1];C(R_3)[q_3,q_1];...;C(R_{n-1})[q_{n-1},q_1];C(R_n)[q_n,q_1];\\
&H[q_2];C(R_2)[q_3,q_2];C(R_3)[q_4,q_2];...;C(R_{n-1})[q_n,q_2];\\
&\qquad\qquad\qquad\qquad ..............................\\
&H[q_{n-1}];C(R_2)[q_n,q_{n-1}];\\
&H[q_n]; \\
&\mathit{Reverse}[q_1,...,q_n]
\end{align*}
\caption{A quantum circuit for quantum Fourier transform.}\label{QFT-table}
\end{table*}
where $C(R_l)[q_i,q_j]$ stands for the controlled-$R_l$ with $q_i$ as its control qubit and $q_j$ as its target qubit, $R_l$ denotes the single qubit gate (parameterised by classical integer variable $l$): $$R_l=\left(\begin{array}{cc} 1 & 0\\ 0 & e^{2\pi i/2^l}\end{array}\right)$$ 
 and $\mathit{Reverse}[q_1,...,q_n]$ is the quantum gate that reverses the order of qubits $q_1,...,q_n$. As in Example \ref{exam-array}, we introduce a qubit array $q$ of type $\mathbf{integer}\rightarrow\hs_2$. Then the quantum circuit in Table \ref{QFT-table} can be defined by induction on the number $n-m$ of qubits:\begin{align*}&\mathit{QFT}[q[1:n]]::=\mathit{QFT}^\ast[q[1:n]];\mathit{Reverse}[q[1:n]],\\ 
 &\mathit{QFT}^\ast[q[m:n]]::= \mathbf{if}\ m=n\ \mathbf{then}\ H[q[n]]\\ &\qquad\qquad\qquad\qquad\qquad\quad\ \ \ \mathbf{else}\ H[Q[m]];\mathit{CR}[q[m:n]];\mathit{QFT}^\ast[q[m+1:n]],\\
 &\mathit{CR}[q[m:n]]::=\mathbf{if}\ m=n\ \mathbf{then\ skip}\\ &\qquad\qquad\qquad\qquad\qquad\ \ \ \mathbf{else}\ \mathit{CR}[q[m:n-1]];C(R_{n-m+1})[q_n,q_m].
 \end{align*}
\end{exam}

The readers can find many more quantum program examples using quantum arrays in the previous literature (for example \cite{openqsam, qsharp}). Also, variational quantum algorithms \cite{Varia} and quantum machine learning algorithms \cite{QML} can be conveniently written as quantum programs with parameterised quantum gates as defined above.

\subsection{Semantics of $\mathit{qWhile}^+$}\label{sec-sem}

\subsubsection{Interpretation of programming language} As pointed out at the beginning of Subsection \ref{sec-qwhile}, $\mathit{qWhile}^+$ is obtained by expanding the classical while-language. Thus, the semantics of quantum programming language $\mathit{qWhile}^+$ is defined in a given classical structure as the interpretation of its classical part (i.e. the while-language) together with a quantum structure, as described in the following definition. For simplicity of presentation, we omit the definition of a classical structure, which is standard (for example, see Chapter 3 of \cite{Apt09} and Section 3.2 of \cite{LSie}). 
\begin{defn}A quantum structure consists of: 
\begin{enumerate}\item each quantum (simple or array) variable is assigned a quantum type $\hs_q$. Thus, the Hilbert space of all quantum variables is the tensor product: \begin{equation}\label{def-all}\hs_{\mbox{all}}=\bigotimes_{q\in\mathcal{QV}}\hs_q.\end{equation} \item each symbol $U\in\mathcal{U}$ with parameter type $T_1\times ...\times T_m$ and quantum type $\hs_1\otimes ...\otimes\hs_n$ is interpreted as a family of unitary operators, $$\left\{U(v_1,...,v_m)\right\}_{(v_1,...,v_m)\in T_1\times ...\times T_m},$$ on $\hs_1\otimes ...\otimes\hs_n$. 
Equivalently, parameterised gate symbol $U$ is interpreted as a mapping:
\begin{align*}U:\ &T_1\times ...\times T_m\rightarrow \mathbb{U}\left(\hs_1\otimes...\otimes\hs_n\right).\\ &(v_1,...,v_m)\mapsto U(v_1,...,v_m).\end{align*}
Here, $\mathbb{U}\left(\hs_1\otimes...\otimes\hs_n\right)$ stands for the group of unitary operators on $\hs_1\otimes...\otimes\hs_n$. For simplicity of presentation, we slightly abuse the notation so that the interpretation of $U$ is also written as $U$. Thus, $U$ satisfies the unitarity: $$U(v_1,...,v_m)^\dag U(v_1,...,v_m)=U(v_1,...,v_m)U(v_1,...,v_m)^\dag =I$$ for any $(v_1,...,v_m)\in T_1\times...\times T_m$, where $I$ is the identity operator, and $^\dag$ stands for the adjoint of an operator.  
\item each symbol $M\in\mathcal{M}$ with quantum type $\hs_1\otimes ...\otimes\hs_n$ and classical type $T$ is interpreted as a quantum measurement on $\hs_1\otimes ...\otimes\hs_n$, which is also denoted as $M$ for simplicity; that is, $$M=\{M_m\}_{m\in T},$$  and all $M_m$ are operators on $\hs_1\otimes ...\otimes\hs_n$ such that $\sum_{m\in T} M_m^\dag M_m=I$ (the identity operator). Obviously, the classical type $T$ of $M$ denotes the set of possible measurement outcomes $m$. In this paper, we always assume that the classical (outcome) type $T$ of any measurement $M$ is finite.  
\end{enumerate}\end{defn}

\subsubsection{Operational semantics} Now we assume that a pair of classical structure and quantum structure is given as an interpretation of language $\mathit{qWhile}^+$. We define a \textit{classical-quantum state} (\textit{cq-state} or \textit{state} for short) as a pair $(\sigma,\rho)$, where $\sigma$ is a classical state (i.e. a state of classical variables defined in the classical structure, as a mapping from classical variables to the domain of the classical structure), and $\rho$ is a partial density operator on $\hs_\mathit{all}$, i.e. a positive operator with trace $\tr(\rho)\leq 1$ (following Selinger \cite{Selinger}), that is used to describe the state of quantum variables. The set of all cq-states is denoted as $\Omega$. Then a configuration is defined as a triple $(P,\sigma,\rho)$, where $P$ is a quantum program written in $\mathit{qWhile}^+$ or termination symbol $\downarrow$,  and $(\sigma,\rho)\in\Omega$ is a cq-state. Here, $\rho$ denotes the state of quantum variables in the sense that the quantum variables are in a (possibly mixed) state $\rho/\tr(\rho)$ with probability $\tr(\rho)$, and the probability $\tr(\rho)$ comes from the measurements occurring in the previous computation. Thus, we have:
\begin{defn}\label{def-operational} The operational semantics of $\mathit{qWhile}^+$ is given as the transition relation $\rightarrow$ between configurations defined by the transition rules in Table \ref{trans-rules}. 
\begin{table*}[t]
\begin{equation*}\begin{split}&(\mbox{Ski})\ \ \ (\mathbf{skip},\sigma,\rho)\rightarrow (\downarrow,\sigma,\rho)\qquad\qquad (\mbox{Ass})\ \ (x:=e,\sigma,\rho)\rightarrow (\downarrow,\sigma[x:=\sigma(e)],\rho)\\ 
&(\mbox{Init})\ \ {\rm If}\ B=\{|n\rangle\}\ {\rm is\ an\ orthonormal\ basis\ of}\ \hs_q\ {\rm and}\ |0\rangle\ {\rm is\ a\ basis\ state\ in}\ B,\ {\rm then:}\\  
&\qquad\qquad (q:=|0\rangle, \sigma,\rho)\rightarrow \left(\downarrow,\sigma,\sum_n|0\rangle_{\sigma(q)}\langle n|\rho|n\rangle_{\sigma(q)}\langle 0|\right)\\ 
& 
(\mbox{Uni})\ \ \ 
\frac{\sigma\models\mathit{Dist}(q_1[s_{11},...,s_{1n_1}],...,q_k[s_{k1},...,s_{kn_k}])}{\left(U(t_1,...,t_m)[q_1[s_{11},...,s_{1n_1}],...,q_k[s_{k1},...,s_{kn_k}]],\sigma, \rho\right)\rightarrow \left(
\downarrow, \sigma, U_\sigma\rho U_\sigma^\dag\right)}\\ 
&(\mbox{Meas})\ \ {\rm If\ measurement\ symbol}\ M\ {\rm is\ interpreted\ as}\ M=\{M_m\}\ {\rm (in\ the\ given}\\ &\qquad\ \ \ \ \ \ {\rm quantum\ structure), then\ for\ every\ possible\ outcome}\ m:\\ 
&\qquad\qquad \frac{\sigma\models\mathit{Dist}(q_1[s_{11},...,s_{1n_1}],...,q_k[s_{k1},...,s_{kn_k}])}{(x:=M[q_1[s_{11},...,s_{1n_1}],...,q_k[s_{k1},...,s_{kn_k}]],\sigma, \rho)\rightarrow \left(
\downarrow, \sigma[x:=m], M_m^\sigma\rho (M_m^\sigma)^\dag\right)}\\ 
&(\mbox{Seq})\ \ \ \ \frac{(P_1,\sigma,\rho)\rightarrow (
P_1^{\prime},\sigma^{\prime},\rho^\prime)}{(
P_1;P_2,\sigma,\rho)\rightarrow (P_1^{\prime};P_2,\sigma^\prime,\rho^\prime)}\\
&(\mbox{Cond)}\ \frac{\sigma\models b}{(\mathbf{if}\ b\ \mathbf{then}\ P_1\ \mathbf{else}\ P_0,\sigma,\rho)\rightarrow (P_1,\sigma,\rho)}\qquad \frac{\sigma\models \neg b}{(\mathbf{if}\ b\ \mathbf{then}\ P_1\ \mathbf{else}\ P_0,\sigma,\rho)\rightarrow (P_0,\sigma,\rho)}\\
&(\mbox{Loop})\ \ \frac{\sigma\models b}{(\mathbf{while}\ b\ \mathbf{do}\ P,\sigma,\rho)\rightarrow (P;\mathbf{while}\ b\ \mathbf{do}\ P,\sigma,\rho)}\qquad \frac{\sigma\models\neg b}{(\mathbf{while}\ b\ \mathbf{do}\ P,\sigma,\rho)\rightarrow (\downarrow,\sigma,\rho)}
\end{split}\end{equation*}
\caption{Transition Rules.}\label{trans-rules}
\end{table*}
\end{defn}

The operational semantics defined above is a natural quantum extension of the operational semantics of classical while-language. It is the same as that for quantum programs with classical variables given in \cite{YingF} except that here we need to carefully deal with quantum arrays and parameterised quantum gates, but \cite{YingF} does not. However, it is different from those defined in \cite{Deng, FengY} due to the treatment of cq-states (as explained in Section \ref{sec-intro}).
Some notations used in Figure \ref{trans-rules} are explained as follows:
\begin{itemize}\item In the transition rule (Ass), $\sigma(e)$ denotes the value of classical expression $e$ in the classical state $\sigma$ (in the given classical structure), and $\sigma[x:=\sigma(e)]$ stands for the classical state in which the value of $x$ is $\sigma(e)$ but the values of other classical variables are the same as in $\sigma$. 
\item In the rule (Init), if $q$ is a simple quantum variable, then $\sigma(q)=q$, and if $q$ is a subscripted quantum variable $q^{\prime}[s_1,...,s_n]$, then $\sigma(q)$ denotes the quantum system $q^{\prime}[\sigma(s_1),...,\sigma(s_n)]$. 
\item The distinctness formula $\mathit{Dist}(\overline{q})$ in the premises of the rules (Uni) and (Meas) means that $\overline{q}$ is a string of distinct (simple or subscripted) quantum variables (see equation (\ref{distinctness}) for its precise definition). 
This condition is introduced to guarantee that the quantum gates and measurements under consideration are well-defined. For example, let $q$ be an array of qubits. Then the controlled-NOT gate $\mathit{CNOT}[q[2k+1],q[4k-3]]$ is meaningful in a classical state $\sigma$ if and only if $\sigma(k)\neq 2$.    
\item In the conclusion of the rules (Uni) and (Meas), it should be noted that in the left hand side of transition $\rightarrow$, $U\in\mathcal{U}$ and $M\in\mathcal{M}$ are a gate symbol and a measurement symbol, respectively, but in the right hand side of $\rightarrow$, $U$ and $M$ stands for the interpretation of symbols $U$ and $M$ (in the given quantum structure), respectively. Moreover, $U_\sigma$ means the unitary operator $U(\sigma(t_1),...,\sigma(t_m))$ acting on quantum systems $q_i[\sigma(s_{i1}),...,\sigma(s_{in_i})]$ $(i=1,...,k)$, and $M_m^\sigma$ means the measurement operator $M_m$ acting on quantum systems $q_i[\sigma(s_{i1}),...,\sigma(s_{in_i})]$ $(i=1,...,k)$. 
\end{itemize}

\subsubsection{Denotational semantics} Building upon the above operational semantics, we can introduce the denotational semantics of quantum programs with classical variables:  

\begin{defn}\label{def-denotational}Let $P$ be a quantum program, and let $(\sigma,\rho)\in\Omega$ be a cq-state. Then with input $(\sigma,\rho)$, the output of $P$ is defined to be the multi-set of cq-states in $\Omega$:
\begin{equation}\llbracket P\rrbracket(\sigma,\rho)=\left\{|(\sigma^\prime,\rho^\prime) | (P,\sigma,\rho)\rightarrow^\ast(\downarrow,\sigma^\prime,\rho^\prime)|\right\}
\end{equation} where $\rightarrow^\ast$ stands for the reflexive and transitive closure of transition relation $\rightarrow$, and $\{|\cdot|\}$ denotes a multi-set. \end{defn}

Therefore, the denotational semantics of a quantum program $P$ is defined as a mapping from cq-states $\Omega$ to multi-sets of cq-states in $\Omega$: $$\llbracket P\rrbracket:(\sigma,\rho)\mapsto \llbracket P\rrbracket(\sigma,\rho).$$
Furthermore, for any quantum program $P$, and for any states $(\sigma,\rho),(\sigma^\prime,\rho^\prime)\in\Omega$, we write:
\begin{equation}\label{eq-Right}(\sigma,\rho)\stackrel{P}{\Rightarrow} (\sigma^\prime,\rho^\prime)\end{equation} if and only if $(\sigma^\prime,\rho^\prime)\in\llbracket P\rrbracket(\sigma,\rho).$ The denotational semantics can also be understood in a slightly different way. Let $\Theta$ be a multi-set of cq-states in $\Omega$. Then for each classical state $\sigma$, we define:
\begin{equation}\label{eq-Theta}\Theta(\sigma)=\sum\{|\rho\ |\ (\sigma,\rho)\in\Theta|\}.\end{equation}
We further define the normalisation of $\Theta$ as the set $$N(\Theta)=\{(\sigma,\Theta(\sigma))\ |\ \Theta(\sigma)\neq 0\}.$$ Obviously, for any classical state $\sigma$, there exists at most one quantum state $\rho$ such that $(\sigma,\rho)\in N(\Theta).$ 
Now given a quantum program $P$, for each input state $(\sigma,\rho)\in\Omega$, and for each classical state $\sigma^\prime$, $\llbracket P\rrbracket(\sigma,\rho)(\sigma^\prime)$ is the summation of (partial) quantum states $\rho^\prime$ such that $(\sigma^\prime,\rho^\prime)$ is outputted by $P$ though different execution paths. Thus, the denotational semantics of $P$ can be thought of as a mapping from $\Omega$ to $2^\Omega$:
$$\llbracket P\rrbracket: (\sigma,\rho)\mapsto N(\llbracket P\rrbracket(\sigma,\rho)).$$ 

The following proposition gives an explicit representation of the denotational semantics of quantum programs in terms of their structures. It is often more convenient in applications than Definition \ref{def-denotational}. In particular, it will be used in the proofs of soundness and (relative) completeness of our QHL (quantum Hoare logic).    

\begin{prop}[Structural Representation of Denotational Semantics]\label{prop-structure}
For any input state $(\sigma,\rho)\in\Omega$, we have:
\begin{enumerate}\item $\llbracket\mathbf{skip}\rrbracket(\sigma,\rho)=\{|(\sigma,\rho)|\}.$
\item $\llbracket x:=e\rrbracket(\sigma,\rho)=\{|(\sigma[x:=\sigma(e)],\rho)|\}.$
\item $\llbracket q:=|0\rangle\rrbracket(\sigma,\rho)=\left\{|(\sigma,\sum_n|0\rangle_{\sigma(q)}\langle n|\rho|n\rangle_{\sigma(q)}\langle 0|)\right\}.$
\item $$\llbracket U(t_1,...,t_m)[\overline{q}]\rrbracket(\sigma,\rho)=\begin{cases}\left\{|(\sigma, U_\sigma\rho U_\sigma^\dag)|\right\}\ &{\rm if}\ \sigma\models \mathit{Dist}(\overline{q})\\ \emptyset\ &{\rm otherwise}\end{cases}$$ where $U_\sigma$ denotes quantum gate $U(\sigma(t_1),...,\sigma(t_m))$ acting on $\sigma(\overline{q})$.  
\item $$\llbracket x:=M[\overline{q}]\rrbracket(\sigma,\rho)=\begin{cases}\left\{|\left(\sigma[x:=m],M^\sigma_m\rho (M^\sigma_m)^\dag\right)|\right\}\ &{\rm if}\ \sigma\models \mathit{Dist}(\overline{q})\\ \emptyset\ &{\rm otherwise}\end{cases}$$ where $m$ ranges over all possible outcomes of the measurement $M$, and $M_m^\sigma$ denotes measurement operator $M_m$ acting on $\overline{q}$. 
\item $$\llbracket P_1;P_2\rrbracket(\sigma,\rho)=\llbracket P_2\rrbracket\left(\llbracket P_1\rrbracket(\sigma,\rho)\right)=\bigcup_{(\sigma^\prime,\rho^\prime)\in \llbracket P_1\rrbracket(\sigma,\rho)}\llbracket P_2\rrbracket(\sigma^\prime,\rho^\prime).$$
\item \begin{equation*}\llbracket\mathbf{if}\ b\ \mathbf{then}\ P_1\ \mathbf{else}\ P_0\rrbracket(\sigma,\rho)=\begin{cases}\llbracket P_1\rrbracket(\sigma,\rho)\ {\rm if}\ \sigma\models b,\\
\llbracket P_0\rrbracket(\sigma,\rho)\ {\rm if}\ \sigma\models \neg b. 
\end{cases}\end{equation*}
\item \begin{align*}\llbracket\mathbf{while}\ b\ \mathbf{do}\ P\rrbracket(\sigma,\rho)=\{| (\sigma^\prime,\rho^\prime)| &(\sigma,\rho)= (\sigma_0,\rho_0)\stackrel{P}{\Rightarrow}...\stackrel{P}{\Rightarrow}(\sigma_n,\rho_n)=(\sigma^\prime,\rho^\prime),\\ & n\geq 0, \sigma_i\models b\ (0\leq i<n)\ {\rm and}\ \sigma_n\models\neg b|\}
\end{align*} where the relation $\stackrel{P}{\Rightarrow}$ is defined as (\ref{eq-Right}). 
\end{enumerate}\end{prop}
\begin{proof} Routine by Definitions \ref{def-operational} and \ref{def-denotational}.
\end{proof}

The next proposition indicates that the trace of the state of quantum variables does not increase during the execution of a program. 
\begin{prop}\label{trace-decrease}For any quantum program $P$ and for any input state $(\sigma,\rho)\in\Omega$, we have:
\begin{align*}\sum\left\{|\tr(\rho^\prime)|(\sigma^\prime,\rho^\prime)\in\llbracket P\rrbracket(\sigma,\rho)|\right\}\leq\tr(\rho).
\end{align*}
\end{prop}
\begin{proof} It can be easily proved by induction on the structure of $P$. But it also can be derived directly from Proposition \ref{prop-structure}. 
\end{proof}

\subsubsection{Termination of programs} Let us define: \begin{equation}\label{def-NT}\mathit{NT}(P)(\sigma,\rho)=\tr(\rho)-\sum\left\{|\tr(\rho^\prime)|(\sigma^\prime,\rho^\prime)\in\llbracket P\rrbracket(\sigma,\rho)|\right\}\end{equation} 
Then by Proposition \ref{trace-decrease} we know that $0\leq \mathit{NT}(P)(\sigma,\rho)\leq 1$. Furthermore, it is clear from Definition \ref{def-denotational} that $\mathit{NT}(P)(\sigma,\rho)$ is the probability that program $P$ starting in state $(\sigma,\rho)$ does not terminate. In particular, if $\mathit{NT}(P)(\sigma,\rho)=0$; that is, 
$$(\mbox{Termination condition})\qquad\quad \sum\left\{|\tr(\rho^\prime)|(\sigma^\prime,\rho^\prime)\in\llbracket P\rrbracket(\sigma,\rho)|\right\}=\tr(\rho),$$ then we say that program $P$ starting in state $(\sigma,\rho)$ almost surely terminates. 

\subsubsection{Equivalence of programs} Before concluding this section, let us introduce the notion of equivalence between quantum programs. Two multi-sets $\Theta_1,\Theta_2$ of cq-states are said to be equivalent, written $\Theta_1\approx\Theta_2$, if for every classical state $\sigma$, it holds that $\Theta_1(\sigma)=\Theta_2(\sigma)$ according to the defining equation (\ref{eq-Theta}). It is easy to see that any multiset $\Theta$ of cq-states is equivalent to its normalisation $N(\Theta)$; that is, $\Theta\approx N(\Theta)=\{(\sigma,\Theta(\sigma))|\Theta(\sigma)\neq 0\}.$
With this notation, we have:

\begin{defn}\label{def-equiv} Two quantum programs $P_1,P_2$ are equivalent, written $P_1\approx P_2$, if for any input state $(\sigma,\rho)\in\Omega$, it holds that $\llbracket P_1\rrbracket(\sigma,\rho)\approx\llbracket P_2\rrbracket(\sigma,\rho).$
\end{defn}
 
 It is easy to see that whenever $P_1\approx P_2$, when for any state $(\sigma, \rho)$, $P_1$ starting in $(\sigma,\rho)$ almost surely terminates if and only if so does $P_2$. Moreover, it will be shown in Section \ref{sec-logic} that equivalent quantum programs enjoy the same program logical properties specified as Hoare triples (see Lemma \ref{lem-equivalence}). 
 
\section{An Assertion Language for Quantum Programs with Classical Variables}\label{sec-assert}

In this section, we turn to define an assertion language for specifying preconditions and postconditions of the quantum programs with classical variables described in Section \ref{sec-lang}. In this language, as briefly discussed in Section \ref{sec-intro}, an assertion is a pair consisting of a classical first-order logical formula that describes the property of classical variables and a quantum predicate formula that describes the property of quantum variables. 

Quantum predicates were first introduced in \cite{DP06} as observables represented by Hermitian operators between the zero and identity operators. This purely semantic view of quantum predicates has been adopted in the previous research on quantum Hoare logic. However, it is hard to use in practical verification of quantum programs, in particular, when formalised in a proof assistant (e.g. Coq, Isabelle/HOL or Lean). Our experience in classical logics indicates that syntactic representations of predicates are often much more economic, because they can be constructed from atomic propositions using propositional connectives, and in particular, they allow symbolic reasoning. Therefore, a syntax of quantum predicates was defined in Section 7.6 of \cite{Ying24}. In this section, we extend it to include quantum predicates parameterized by classical variables needed in the next section. 

\subsection{Formal Quantum States}\label{formal-states} Let us first introduce the notion of formal quantum state that will be used to define some atomic quantum predicates in our assertion language.  
 \begin{defn}Formal quantum states are inductively defined as follows:\begin{enumerate}\item If $t$ is a classical expression of type $T$, and $q$ is a simple or subscripted quantum variable of type $\hs_T={\rm span} \{|v\rangle:v\in T\}$ (the Hilbert space with $\{|v\rangle:v\in T\}$ as its orthonormal basis), then $s=|t\rangle_q$ is a formal quantum state, and its signature is $\mathit{Sig}(s)=\{q\}$;
 
 \item If $s_1,s_2$ are formal quantum states, then $s_1\otimes s_2$ (often simply written $s_1s_2$) is a formal quantum state too, and its signature is $\mathit{Sig}(s_1\otimes s_2)=\mathit{Sig}(s_1)\cup\mathit{Sig}(s_2)$;
 
 \item If $s_1, s_2$ are formal quantum states such that $\mathit{Sig}(s_1)=\mathit{Sig}(s_2)$, and $\alpha_1,\alpha_2$ are two classical expressions of type $\mathbb{C}$ (complex numbers), then $\alpha_1s_1+\alpha_2 s_2$ is a formal quantum states, and its signature is $\mathit{Sig}(\alpha_1s_1+\alpha_2s_2)=\mathit{Sig}(s_1)=\mathit{Sig}(s_2)$; and 
 
 \item If $s$ is a formal quantum state and $U(\overline{t})[\overline{q}]$ is a (parameterised) quantum gate defined in Subsection \ref{sec-gate} with $\overline{q}\subseteq\mathit{Sig}(s)$, then $(U(\overline{t})[\overline{q}])s$ is a formal quantum state, and its signature is $\mathit{Sig}((U(\overline{t})[\overline{q}])s)=\mathit{Sig}(s).$
 \end{enumerate}
 \end{defn}
 
 Intuitively, the formula $(U(\overline{t})[\overline{q}])s$ in the clause (4) of the above definition means that quantum gate $U(\overline{t})[\overline{q}]$ is applied to formal quantum state $s$.  
 
 The formal quantum states defined above are not a novel concept. In fact, in the quantum information literature, it is common to represent a quantum state in a form such as \begin{equation}\label{eq-fs-exam}\cos\frac{\theta}{2}|2k+1\rangle_{q[m+1]}|n-5\rangle_{q[3m-4]}+\sin\frac{\theta}{2}|2k-1\rangle_{q[m+1]}|n+3\rangle_{q[3m-4]}\end{equation} where the subscripts $q[m+1]$ and $q[3m-4]$ stand for the $(m+1)$st and $(3m-4)$th subsystems, respectively, of a composite quantum system. Definition \ref{formal-states} simply formalises this commonly used practice as exemplified in (\ref{eq-fs-exam}).
 
Now we assume that a classical structure and a quantum structure are given as an interpretation of our programming language $\mathit{qWhile}^+$ (see Subsection \ref{sec-sem}). Then we can define the semantics of formal quantum states in this interpretation. 
 
 \begin{defn}\label{def-state-sem}Let $s$ be a formal quantum state. For any classical state $\sigma$, if the semantics $\sigma(s)$ of $s$ in classical state $\sigma$ is well-defined, then it is a vector in Hilbert space $$\hs_{\sigma(\mathit{Sig}(s))}=\bigotimes_{q\in\sigma(\mathit{Sig}(s))}\hs_q$$ where $\sigma(\mathit{Sig}(s))=\{\sigma(q^\prime)|q^\prime\in\mathit{Sig}(s)\}$, and $\sigma(q^\prime)$ is defined as in Subsection \ref{sec-sem}. The vector $\sigma(s)$ is inductively defined as follows:   
 \begin{enumerate}\item $\sigma\left(|t\rangle_q\right)=|\sigma(t)\rangle_{\sigma(q)}\in\hs_{\sigma(q)}$; 
 \item $\sigma(s_1\otimes s_2)=\sigma(s_1)\otimes \sigma(s_2)$ provided $\sigma(\mathit{Sig}(s_1))\cap\sigma(\mathit{Sig}(s_2))=\emptyset$;
 \item $\sigma(\alpha_1s_1+\alpha_2s_2)=\sigma(\alpha_1)\sigma(s_1)+\sigma(\alpha_2)\sigma(s_2)$; and
 \item if $\sigma\models\mathit{Dist}(\overline{q})$ (see equation (\ref{distinctness}) for its definition), then: $$\sigma\left((U(\overline{t})[\overline{q}])s\right)=U(\sigma(\overline{t}))_{\sigma(\overline{q})}\sigma(s)$$ where $U(\sigma(\overline{t}))_{\sigma(\overline{q})}$ stands for the unitary operator $U(\sigma(\overline{t}))$ performing on $\sigma(\overline{q})$, $\sigma(\overline{t})=\sigma(t_1),...,\sigma(t_m)$ if $\overline{t}=t_1,...,t_m$, and $\sigma(\overline{q})=\sigma(q_1),...,\sigma(q_n)$ if $\overline{q}=q_1,...,q_n$. 
\end{enumerate}
 \end{defn}
 
It should be noticed that $\sigma(s)$ is not always well-defined because the disjoint condition in clause (2) and the distinction condition in clause (4) may be violated. Whence these conditions are satisfied, it is a vector in the Hilbert space $\hs_{\sigma(\mathit{Sig}(s))}$, but it may not be a quantum state. If it is further the case that the length $\|\sigma(s)\|=1$, then $\sigma(s)$ is a pure quantum state, and we say that $\sigma(s)$ is well-defined. Therefore, the semantics of a formal quantum state $s$ can be considered as a partial function from classical states to (pure) quantum states in $\hs_{\sigma(\mathit{Sig}(s))}$: $$\llbracket s\rrbracket: \sigma\mapsto \sigma(s)\ \mbox{whenever}\ \sigma(s)\ \mbox{is well-defined}.$$ This semantic view indicates that a formal quantum state can be seen as a quantum state parameterised by classical variables, and it coincides with the definition of cq-states in \cite{FengY}.

\subsection{Syntax of Quantum Predicates}\label{sec-pred-syn}
Now let us start to consider general quantum predicates. We assume a set $\mathcal{P}$ of \textit{atomic quantum predicate symbols}. Each quantum predicate symbol $K\in\mathcal{P}$ is equipped with a family of classical variables $x_1,...,x_m$ as its formal parameters, and it is also equipped with a quantum type $\hs_1\otimes ...\otimes\hs_n$, where $n$ is called the arity of $K$. 
We often write $K=K(x_1,...,x_m)$ in order to show the parameters. If for any $i\leq m$, $t_i$ is a classical expression with the same type of $x_i$, and for each $j\leq n$, $q_j$ is a simple or subscripted quantum variables with type $\hs_j$, then \begin{equation}\label{basic-pre}K(\overline{t})[\overline{q}]=K(t_1,...,t_m)[q_1,...,q_n]\end{equation} is called an \textit{atomic quantum predicate}, where $\overline{t}=t_1,...,t_m$ and $\overline{q}=q_1,...,q_n$. 

In particular,  for any quantum type $\hs_1\otimes ...\otimes\hs_n$, we introduce quantum predicate symbol $I_{\hs_1\otimes ...\otimes\hs_n}$. It will be interpreted as the identity operator on Hilbert space $\hs_1\otimes ...\otimes\hs_n$. Indeed, it can be viewed as the quantum version of \textquotedblleft true\textquotedblright\ in the domain $\hs_1\otimes ...\otimes\hs_n$. For simplicity, the subscript $\hs_1\otimes ...\otimes\hs_n$ of $I_{\hs_1\otimes ...\otimes\hs_n}$ is often omitted. For each formal quantum state $s$, we introduce an atomic quantum predicate $[s]$. Intuitively, the quantum predicate $[s]$ means \textquotedblleft to be in the state $s$\textquotedblright. As we will see later, mathematically, it denotes (the projection operator onto) the one-dimensional subspace spanned by the state $s$. Several other basic quantum predicates that are useful in practically specifying correctness of quantum programs were introduced in Section 3 of \cite{Ying19}; for example, unary quantum predicates defined by a subspace and a unitary operator, and symmetric and anti-symmetric operators as quantum counterparts of equality and inequality. We expect that more basic quantum predicates will be conceived in future applications of quantum Hoare logic and other program logics for quantum programs \cite{Unruh19, QSL1, QSL2, Yu-inc, Yu-temp, Yangjia, Unruh19, Barthe}. 

Recall that in classical logics, one uses connectives to construct compound propositions from atomic propositions. Here, we use the following connectives for quantum predicates: \begin{enumerate}\item[(i)] negation $\neg$; \item[(ii)] tensor product $\otimes$; and \item[(iii)] a family $\mathcal{F}$ of \textit{Kraus operator symbols}. Each Kraus operator symbol $F\in\mathcal{F}$ is assigned a rank $k$, equipped with a family of classical variables $x_1,...,x_m$ as its formal parameters, and equipped with a quantum type of the form $\hs_1\otimes...\otimes\hs_n$, where $n$ is called the arity of $F$. 
We often write $F=F(x_1,...,x_m)$ to show the parameters.  

In particular, we require: \begin{itemize}\item for each basic quantum type (i.e. Hilbert space) $\hs$ and an orthonormal basis $B$ of it, a Kraus operator symbol $F_B\in\mathcal{F}$ of rank $d$ is designated, where $d=\dim\hs$; 
\item for each (parameterized) quantum gate symbol $U\in\mathcal{U}$ with (classical) parameter type $T_1\times ...\times T_m$, a Kraus operator symbol $F_U(x_1,...,x_m)$ of rank $1$ is designated, where $x_i$ is a classical variable of type $T_i$ for $1\leq i\leq m$; 
\item for each measurement symbol $M\in\mathcal{M}$, a Kraus operator symbol $F_M(x)\in\mathcal{F}$ of rank $k$ (with a formal parameter $x$) is designated, where $k$ is the number of possible measurement outcomes. 
\end{itemize} 
If $F(x_1,...,x_m)$ is a Kraus operator symbol with type $\hs_1\otimes...\otimes\hs_n$, and for each $i\leq n$, $q_i$ is a simple or subscripted quantum variables with type $\hs_i$, then \begin{equation}\label{eq-Kraus}F(\overline{t})[\overline{q}]=F(t_1,...,t_m)[q_1,...,q_n]\end{equation} is called a \textit{Kraus operator}, where $t_i$ is a classical expression with the same type as $x_i$ for $1\leq i\leq m$, called the actual parameters.\end{enumerate}

Upon the alphabet of quantum predicates described above, we can introduce the syntax of quantum predicates:
\begin{defn}\label{def-predicates} Quantum predicate formulas are inductively defined as follows: 
$$A::=K(\overline{t})[\overline{q}]\ |\ \neg A\ |\ A_1\otimes A_2\ |\ F(\overline{t})[\overline{q}](\{A_i\}).$$
\begin{enumerate}\item An atomic quantum predicate $K(\overline{t})[\overline{q}]$ of the form given in equation (\ref{basic-pre}) is a quantum predicate formula, and its signature is $\mathit{Sig}(K(\overline{t})[\overline{q}])=\left\{\overline{q}\right\};$
\item If $A$ is a quantum predicate formula, then so is $\neg A$, and $\mathit{Sig}(\neg A)=\mathit{Sig}(A)$;
\item If $A_1$ and $A_2$ are quantum predicate formulas, then $A_1\otimes A_2$ is a quantum predicate formula, and $\mathit{Sig}(A_1\otimes A_2)=\mathit{Sig}(A_1)\cup\mathit{Sig}(A_2)$;  
\item If Kraus operator symbol $F$ has rank $k$, $\{A_i\}$ is a family of $k$ quantum predicate formulas with the same signature $\mathit{Sig}(A_i)=X$, and $\overline{q}\subseteq X$, then $$F(\overline{t})[\overline{q}](\{A_i\})$$ is a quantum predicate formula, and $\mathit{Sig}(F(\overline{t})[\overline{q}](\{A_i\})=X$. In particular, if all $A_i$ are the same quantum predicate formula $A$, then $F(\overline{t})[\overline{q}](\{A_i\})$ is abbreviated to $F(\overline{t})[\overline{q}](A).$
\end{enumerate}
\end{defn} 

With the above syntax, one can use connectives to construct various quantum predicate formulas from atomic quantum predicates (e.g. those $[s]$ defined by formal quantum states $s$). 
We will see how quantum predicate formulas, constructed from atomic quantum predicates with Kraus operator symbols $F_B$, $F_U$ and $F_M$, can be used in formulating the axioms and proof rules of our quantum Hoare logic with classical variables in the next section. Furthermore, using the other connectives of quantum predicates introduced above, various useful auxiliary proof rules can be formulated to ease the verification of quantum programs with classical variables (for example, quantum generalisations of the rules for classical programs given in Section 3.8 of \cite{Apt09} and generalisations of the rules for purely quantum programs presented in Section 6.7 of \cite{Ying24} to quantum programs with classical variables). 

\subsection{Semantics of Quantum Predicates}\label{sec-pred-sem}
As in Subsection \ref{sec-sem}, we assume a classical structure and a quantum structure as the interpretation of our language for quantum predicates. To define the semantics of quantum predicates, however, we need extend the quantum structure by adding the following:
\begin{enumerate}\item each quantum predicate symbol $K(x_1,...,x_m)\in\mathcal{P}$ with type $\hs_1\otimes...\otimes\hs_n$ is interpreted as a family of effects, i.e. Hermitian operators between $0$ and $I$ on $\hs_1\otimes...\otimes\hs_n$, where $0$ and $I$ stand for the zero and identity operators, respectively. For simplicity of presentation, the interpretation of $K$ is written as $K$ too. So, we have: $$K=\left\{K(a_1,...,a_m)\right\},$$ where for $i\leq m$, $a_i$ ranges over the values of classical variable $x_i$, and for each fixed tuple $(a_1,...,a_n)$ of parameters, $K(a_1,...,a_m)$ is an operator such that $0\leq K(a_1,...,a_m)\leq I$, where $\leq$ stands for the L\"{o}wner order; that is, $A\leq B$ if and only if $B-A$ is a positive operator. Consequently,  $K(a_1,...,a_m)^\dag =K(a_1,...,a_m)$; that is, $K(a_1,...,a_m)$ is a Hermitian operator and can be seen as an observable in physics. It should be pointed out that this interpretation of $K$ is consistent with the notion of quantum predicate defined in \cite{DP06} in the sense that for a fixed tuple $(a_1,...,a_m)$, $K(a_1,...,a_m)$ is exactly a quantum predicate in the sense of \cite{DP06}. 

In particular, for each formal quantum state $s$, the atomic quantum predicate $[s]$ is interpreted as the projection operator onto the one-dimensional subspace spanned by quantum state $\sigma(s)$, where $\sigma(s)$ is the semantics $s$ in $\sigma$ (see Definition \ref{def-state-sem}). For example, if formal quantum state $$s=\cos\frac{\theta}{2}|x-\frac{1}{2}\rangle_{q[3n-2]}+\sin\frac{\theta}{2}|x+\frac{1}{2}\rangle_{q[3n-2]},$$ and classical state $\sigma$ is given by $\sigma(\theta)=\frac{\pi}{2}, \sigma(x)=\frac{1}{2}, \sigma(n)=3$, then $[s]$ is interpreted as operator $|+\rangle\langle+|$ on the state space of qubit $q[7]$, where $|+\rangle=\frac{1}{\sqrt{2}}(|0\rangle+|1\rangle)$.

\item each Kraus operator symbol $F(x_1,...,x_m)\in\mathcal{F}$ with rank $k$ and type $\hs_1\otimes ...\otimes\hs_n$ is interpreted as a family $$F=\{F_i(a_1,...,a_m)\}$$ of $k$ operators on $\hs_1\otimes ...\otimes\hs_n$, where $i$ ranges over $k$ indices. Furthermore, each $F_i(a_1,...,a_m)$ is parameterised by $a_1,...,a_m$, where for $j\leq m$, $a_j$ ranges over the values of classical variable $x_j$. It is required that for a fixed tuple $(a_1,...,a_m)$ of parameters, $\left\{F_i(a_1,...,a_m)\right\}$ satisfies the normalization condition: \begin{equation}\label{eq-pre-normal}\sum_i F_i(a_1,...,a_n)F_i^\dag(a_1,...,a_n) = I.\end{equation} In particular, \begin{itemize}\item for an orthonormal basis $B=\{|n\rangle\}$, the Kraus operator symbol $F_B$ is interpreted as $F_B=\{|0\rangle\langle n|\},$ where $|0\rangle$ is a fixed basis state in $B$; \item for any (parameterised) quantum gate symbol $U\in\mathcal{U}$, if in the given quantum structure, $U$ is interpreted a family $\left\{U(v_1,...,v_m)\right\}$ of unitary operators (see Subsection \ref{sec-sem}), then the Kraus operator symbol $F_U$ is interpreted as $F_U=\left\{U(v_1,...,v_m)^\dag\right\}$; 
\item for any measurement symbol $M\in\mathcal{M}$, if in the quantum structure, $M$ is interpreted as $M=\left\{M_m\right\}$, then the Kraus operator symbol $F_M(m)$ is interpreted as $F_M(m)=\left\{M_m^\dag\right\}$.\end{itemize}
\end{enumerate}

We should notice the dualities between basis $B$, unitary operator $U$, measurement $M$ and their corresponding Kraus operators $F_B$, $F_U$, $F_M$, respectively.  

Given the above classical structure and (extended) quantum structure as the interpretation of our language for quantum predicates, we can introduce the following: 

\begin{defn}\label{def-sig}For each quantum predicate formula $A$ and for a classical state $\sigma$, the interpretation $\sigma(A)$ of $A$ in $\sigma$ is inductively defined as follows:
\begin{enumerate}\item If $A=K(\overline{t})[\overline{q}]$ is an atomic quantum predicate, then:
\begin{enumerate}\item whenever $\sigma\not \models\mathit{Dist}(\overline{q})$, then $\sigma(A)$ is not well-defined;
\item whenever $\sigma\models\mathit{Dist}(\overline{q})$, then 
\begin{equation}\label{eq-K}\sigma(A)=K(\sigma(\overline{t}))_{\sigma(\overline{q})}\end{equation} 
Note that $K$ in the right hand side of equation (\ref{eq-K}) denotes the interpretation of atomic predicate symbol $K$. Furthermore, $K(\sigma(\overline{t}))_{\sigma(\overline{q})}$ stands for the operator $K(\sigma(\overline{t}))$ acting on the quantum systems $\sigma(\overline{q})$. \end{enumerate} In particular, if $A=[s]$ is an atomic quantum predicate defined by a formal quantum state $s$, and $\sigma(s)=|\psi\rangle$ is well-defined (see Definition \ref{def-state-sem}), then $\sigma(A)=|\psi\rangle\langle\psi|$.  
\item If $A=\neg B$, then: \begin{enumerate}\item $\sigma(A)$ is not well-defined when $\sigma(B)$ is not well-defined; and \item $\sigma(A)= I -\sigma(B)$ when $\sigma(B)$ is well-defined. Here, $I$ stands for the identity operator on Hilbert space $\hs_{\sigma(\mathit{Sig}(B))}.$\end{enumerate}
\item If $A=A_1\otimes A_2$, then: \begin{enumerate}\item $\sigma(A)$ is not well-defined when either $\sigma(A_1)$ or $\sigma(A_2)$ is not well-defined, or $\sigma(\mathit{Sig}(A_1))\cap\sigma(\mathit{Sig}(A_2))\neq \emptyset$; and \item $\sigma(A)=\sigma(A_1)\otimes\sigma(A_2)$ when both $\sigma(A_1)$ and $\sigma(A_2)$ are well-defined, and $\sigma(\mathit{Sig}(A_1))\cap\sigma(\mathit{Sig}(A_2))=\emptyset$.\end{enumerate}
\item If $A=F(\overline{t})[\overline{q}](\{B_i\})$ and Kraus operator symbol $F$ is interpreted as $F=\{F_i(\overline{a})\}$, then:\begin{enumerate}\item whenever $\sigma(B_i)$ is not well-defined for some $i$ or $\sigma\not \models\mathit{Dist}(\overline{q})$, then $\sigma(A)$ is not well-defined; \item whenever $\sigma(B_i)$ is well-defined for every $i$ and $\sigma\models\mathit{Dist}(\overline{q})$, then
\begin{equation}\label{eq-kraus-sem}\sigma(A)=\sum_i F_i(\sigma(\overline{t}))_{\sigma(\overline{q})}\cdot \sigma(B_i)\cdot F_i^\dag(\sigma(\overline{t}))_{\sigma(\overline{q})}\end{equation}
where $F_i(\sigma(\overline{t}))_{\sigma(\overline{q})}$ stands for the operator $F_i(\sigma(\overline{t}))$ acting on quantum systems $\sigma(\overline{q})$. 
\end{enumerate}
\end{enumerate}\end{defn}

The above definition is a parameterised extension of Definition 7.23 in \cite{Ying24}. Now it is a right time to explain why we use the term \textquotedblleft Kraus operator\textquotedblright. Equation (\ref{eq-kraus-sem}) is essentially motivated by Kraus' operator-sum representation of quantum operations (i.e. quantum channels) (see Section 8.2.3 of \cite{NC00}). Kraus operators will be used in defining proof rules (Rule-Accum1) and (Rule-Accum2) of our logic $\mathit{QHL}^+$.  

To illustrate the concepts introduced above, let us see the following simple example:  

\begin{exam}\label{exam-pred} We consider an $n$-dimensional Hilbert space with basis $\{|i\rangle|i=0,1,...,n-1\}$. Let formal quantum state $s_i$ denote the basis state $|i\rangle$. Then:\begin{enumerate}\item Atomic quantum predicate $A_i=[s_i]$ stands for observable (projection operator) $|i\rangle\langle i|$;\item $\neg A_{n-1}$ denotes observable (projection operator) $\sum_{i=0}^{n-2}|i\rangle\langle i|$. 

\item If we use Kraus operator symbol $F(x_0,...,x_{n-1})$ to denote a family $\left\{\sqrt{x_i}\right\}$ of real numbers with $\sum_ix_i\leq 1$; that is, $$F(x_0,...,x_{n-1}) =\left\{F_i(x_0,...,x_{n-1})\right\}$$ with $F_i(x_0,...,x_{n-1})=\sqrt{x_i}$ (viewed as a $0$-dimensional operator) for $0\leq i\leq n-1$, then $F(p_1,...,p_n)(\{A_i\})$ stands for observable (positive operator) $\sum_ip_i|i\rangle\langle i|$.\end{enumerate}\end{exam}

The following lemma shows that every quantum predicate formula indeed represents a quantum predicate parameterised by classical variables.

\begin{lem} For any quantum predicate formula $A$ and for any classical sate $\sigma$, whenever $\sigma(A)$ is well-defined, then it is an effect; that is, an observable (i.e. a Hermitian operator) on $\hs_{\sigma(\mathit{Sig}(A))}$ between the zero operator and the identity operator. 
\end{lem}
\begin{proof}Routine by induction on the structure of $A$. 
\end{proof}

Thus, the semantics of a quantum predicate formula $A$ can be understood as a partial function from classical states to Hermitian operator on $\hs_{\sigma(\mathit{Sig}(A))}$:
$$\llbracket A\rrbracket:\sigma\mapsto\sigma(A)\ \mbox{whenever}\ \sigma(A)\ \mbox{is well-defined.}$$ This understanding is consistent with the notion of cq-predicate defined in \cite{FengY}.

To conclude this subsection, we introduce the following:

\begin{defn}\label{qp-entail}Let $A,B$ be two quantum predicate formulas and $\varphi$ a classical first-order logical formula. Then:\begin{enumerate}\item By the entailment under condition $\varphi$: $$\varphi\models A\leq B$$ we mean that for any classical state $\sigma$, whenever $\sigma\models\varphi$, then either both $\sigma(A)$ and $\sigma(B)$ are not well-defined, or both of them are well-defined and $\sigma(A)\leq\sigma(B)$; that is, $\sigma(B)-\sigma(A)$ is a positive operator. 
\item By equivalence under condition $\varphi$: $$\varphi\models A\equiv B$$ we mean that $\varphi\models A\leq B$ and $\varphi\models B\leq A$; that is,  
whenever $\sigma\models\varphi$, then $\sigma(A)$ is well-defined if and only if so is $\sigma(B)$, and in this case $\sigma(A)=\sigma(B).$
\end{enumerate} 
\end{defn}

For example, we use the same notation as in Example \ref{exam-pred}, and let Kraus operator symbol $F^\prime$ denote a family $\left\{\sqrt{p_i^\prime}\right\}$ of real numbers with $\sum_ip_i^\prime\leq 1$, then it holds that \begin{equation}\label{entail-prob}(\forall i)(p_i\leq p_i^\prime)\models F(p_0,...,p_{n-1})(\{A_i\})\leq F(p_0^\prime,...,p_{n-1}^\prime)(\{A_i\}).\end{equation}

Some basic properties of the entailment are given in the following: 

\begin{prop} \begin{enumerate}
\item Entailment is reflexive, anti-symmetric and transitive: \begin{enumerate}\item $\varphi\models A\leq B$ and $\varphi\models B\leq A$ if and only if $\varphi\models A\equiv B$;
\item If $\varphi\models A\leq B$ and $\varphi\models B\leq C$, then $\varphi\models A\leq C.$
\end{enumerate}
\item Entailment is preserved by quantum predicate connectives: \begin{enumerate}\item If $\varphi\models A\leq B$, then $\varphi\models\neg B\leq \neg A$ and $
\varphi\models A\otimes C\leq B\otimes C$; 
\item If $\varphi\models A_i\leq B_i$ for every $i$, then $\varphi\models F(\overline{t})[\overline{q}](\{A_i\})\leq F(\overline{t})[\overline{q}](\{B_i\}).$
\end{enumerate}\end{enumerate}\end{prop}
\begin{proof} Routine by definition. 
\end{proof}

\subsection{Substitutions for Quantum Predicates}\label{sec-subst}

In this subsection, we introduce the notion of substitution in a quantum predicate formula, which is needed in formulating the axioms and proof rules of our quantum Hoare logic in the next section. It can be defined in a manner familiar to us from classical logic: 

\begin{defn}Let $A$ be a quantum predicate formula, $x$ a classical variable and $e$ a classical expression. Then the substitution $A[e/x]$ of $e$ for $x$ in $A$ is a quantum predicate formula inductively defined as follows:
\begin{enumerate}
\item If $A=K(\overline{t})[\overline{q}]$, then 
\begin{equation}\label{eq-K-s}A[e/x]=K(\overline{t}[e/x])[\overline{q}[e/x]].\end{equation}
\item If $A=\neg B$, then $A[e/x]= \neg (B[e/x])$. 
\item If $A=F(\overline{t})[\overline{q}](\{B_i\})$, then 
\begin{equation}A[e/x]=F(\overline{t}[e/x])[\overline{q}[e/x]](B_i[e/x])\end{equation}
where:\begin{enumerate}\item[(a)] for $\overline{t}=t_1,...,t_m$, we set $\overline{t}[e/x]=t_1[e/x],...,t_m[e/x]$; and \item[(b)] for $\overline{q}=q_1[s_{11},...,s_{1n_1}],...,q_k[s_{k1},...,s_{kn_k}]$, we set: 
$$\overline{q}[e/x]=q_1[s_{11}[e/x],...,s_{1n_1}[e/x]],...,q_k[s_{k1}[e/x],...,s_{kn_k}[e/x]].$$\end{enumerate}
\end{enumerate}\end{defn}

As in the case of classical logic, we can prove the following substitution lemma. It will be used in proving the soundness of our quantum Hoare logic with classical variables to be presented in the next section.  
\begin{lem}[Substitution for Quantum Predicates]\label{lem-sub} For any quantum predicate formula $A$ and classical state $\sigma$, we have:\begin{enumerate}\item $\sigma(A[e/x])$ is well-defined if and only if so is $\sigma[x:=\sigma(e)](A)$; and \item if they are well-defined, then $$\sigma(A[e/x])=\sigma[x:=\sigma(e)](A),$$\end{enumerate}
where $\sigma[x:=\sigma(e)]$ is the update of $\sigma$ that agrees with $\sigma$ except for $x$ where its value is $\sigma(e)$.\end{lem}
\begin{proof} For readability, we deferred the tedious proof to \ref{proof-substitute}
\end{proof}

\subsection{Classical-Quantum Assertions}\label{subsec-assertions}

Now we are ready to define an assertion language for quantum programs with classical variables by combining classical first-order logic and quantum predicate formulas introduced above. 

\begin{defn}A classical-quantum assertion (cq-assertion or simply assertion for short) is a pair $(\varphi, A)$, where:\begin{enumerate}\item $\varphi$ is a classical first-order logical formula; and \item$A$ is a quantum predicate formula (see Definition \ref{def-predicates}).
\end{enumerate}\end{defn}

As pointed out in Section \ref{sec-intro}, these cq-assertions $(\varphi,A)$ will be employed in specifying preconditions and postconditions of quantum programs with classical variables. Intuitively, $\varphi$ describes the properties of classical variables and $A$ describes that of quantum variables in a program. It should be particularly noticed that $\varphi$ and $A$ are not independent of each other. Indeed, quantum predicate $A$ is parameterised by classical variables, as seen in the previous subsections. Thus,    
 classical logical formula $\varphi$ puts certain constraints on quantum predicate $A$.

The entailment and equivalence relations between cq-assertions can be defined as follows: 

\begin{defn}\label{cq-entail}\begin{enumerate}\item Entailment: $(\varphi,A)\models (\psi,B)$ if $\varphi\models\psi$ (entailment in classical first-order logic) and $\varphi\models A\leq B$ (see Definition \ref{qp-entail}(1)).
\item Equivalence: $(\varphi,A)\equiv (\psi,B)$ if $\varphi\equiv \psi$ (equivalence in classical first-order logic) and $\varphi\models A\equiv B$ (see Definition \ref{qp-entail}(2)).
\end{enumerate}
\end{defn}

The following proposition asserts that the entailment between cq-assertions is a reflexive, anti-symmetric and transitive relation:
\begin{prop}\begin{enumerate}\item $(\varphi,A)\models (\psi,B)$ and $(\psi,B)\models (\varphi,A)$ if and only if $(\varphi,A)\equiv (\psi,B)$;
\item If $(\varphi,A)\models (\psi,B)$ and $(\psi,B)\models (\chi,C)$, then $(\varphi,A)\models (\chi,C)$.
\end{enumerate}\end{prop}
\begin{proof}Immediate by definition.\end{proof} 

For applications of our quantum Hoare logic to be presented in the next section, the entailment between cq-assertions and quantum Hoare triples (correctness formulas) will be used together in reasoning about the correctness of quantum programs. On the other hand, it is clear from Definition \ref{qp-entail} that classical first-order logic is helpful in proving the entailment between quantum predicate formulas (e.g. the entailment (\ref{entail-prob})) and thus in establishing the entailment between cq-assertions, as we can see from Definition \ref{cq-entail}. Therefore, we will need to exploit effective techniques and tools for verifying the entailment between quantum predicate formulas based on classical logic in order to implement practical tools for verification of quantum programs with classical variables.

\section{Proof System for Quantum Programs with Classical Variables}\label{sec-logic}

Building on the preparations in the previous sections, we are able to develop a Hoare-style logic in this section for reasoning about the correctness of quantum programs written in the programming language $\mathit{qWhile}^{+}$.

\subsection{Correctness Formulas}

As in the classical Hoare logic, we use Hoare triples to specify correctness of quantum programs with classical variables. The cq-assertions defined in Subsection \ref{subsec-assertions} are employed as preconditions and postconditions to define the correctness formulas of quantum programs. More precisely, we use a quantum Hoare triple of the form: \begin{equation}\label{Hoare-triple}\{\varphi,A\}\ P\ \{\psi,B\}\end{equation} as a correctness formula of a quantum program $P$, where $\varphi,\psi$ are first-order logical formulas, called the classical precondition and postcondition, respectively, and $A,B$ are quantum predicate formulas (see Definition \ref{def-predicates}), called quantum precondition and postcondition, respectively. A quantum Hoare triple of the form (\ref{Hoare-triple}) has two different interpretations: 
\begin{defn}\begin{enumerate}\item We say that Hoare triple $\{\varphi,A\}\ P\ \{\psi,B\}$ is true in the sense of total correctness, written $$\models_\mathit{tot}\{\varphi,A\}\ P\ \{\psi,B\},$$ if for any input $(\sigma,\rho)$, whenever $\sigma\models\varphi$ and $\sigma(A)$ is well-defined, then it holds that 
\begin{equation}\label{eq-total}\tr(\sigma(A)\rho)\leq\sum\left\{|\tr\left(\sigma^\prime(B)\rho^\prime\right) | (\sigma^\prime,\rho^\prime)\in\llbracket P\rrbracket(\sigma,\rho), \sigma^\prime\models\psi,\ \mbox{and}\ \sigma^{\prime}(B)\ \mbox{is\ well-defined}|\right\}.\end{equation}
\item The partial correctness: $$\models_\mathit{par}\{\varphi,A\}\ P\ \{\psi,B\}$$ is defined if inequality (\ref{eq-total}) is weakened by the following: 
\begin{equation}\label{eq-partial}\begin{split}\tr(\sigma(A)\rho)\leq \sum &\left\{|\tr\left(\sigma^\prime(B)\rho^\prime\right) | (\sigma^\prime,\rho^\prime)\in\llbracket P\rrbracket(\sigma,\rho),\ \sigma^\prime\models\psi,\ \mbox{and}\ \sigma^{\prime}(B)\ \mbox{is\ well-defined}|\right\} \\ &+\mathit{NT}(P)(\sigma,\rho)\end{split}\end{equation}
\end{enumerate} where $\mathit{NT}(P)(\sigma,\rho)$ stands for the probability that program $P$ starting in state $(\sigma,\rho)$ does not terminate (see its defining equation (\ref{def-NT})). 
\end{defn}

The above definition is a natural quantum generalisation of the notions of partial and total correctness in classical Hoare logic. First, we note that a premise of inequalities (\ref{eq-total}) and (\ref{eq-partial}) is $\sigma\models\varphi$, which means that the initial classical state $\sigma$ satisfies the precondition $\varphi$ for classical variables. On the other hand, the requirement that the final classical state $\sigma^\prime$ satisfies the postcondition $\psi$ for classical variables (i.e. $\sigma^\prime\models\psi$) is imposed in the right-hand side of 
inequalities (\ref{eq-total}) and (\ref{eq-partial}). 
Furthermore, according to Definition \ref{def-sig}, $\sigma(A)$ denotes the observable specified by quantum predicate $A$ in classical state $\sigma$.  
By the interpretation of quantum mechanics, we know that $\tr(\sigma(A)\rho)$ is the expectation of observable $\sigma(A)$ in quantum state $\rho$, which can be understood as the degree that quantum state $\rho$ satisfies quantum predicate $A$ in classical state $\sigma$. The same interpretation applies to $\tr\left(\sigma^\prime(B)\rho^\prime\right)$. Therefore, inequalities (\ref{eq-total}) and (\ref{eq-partial}) are essentially the quantitative versions of the following two statements, respectively:\begin{itemize}\item 
\textbf{\textit{Total correctness}}: for any input cq-state $(\sigma,\rho)$, whenever $(\sigma,\rho)$ satisfies cq-assertion $(\varphi,A)$, then program $P$ terminates and its output $(\sigma^\prime,\rho^\prime)$ satisfies cq-assertion $(\psi,B)$;
\item \textbf{\textit{Partial correctness}}: for any input cq-state $(\sigma,\rho)$, whenever $(\sigma,\rho)$ satisfies cq-assertion $(\varphi,A)$, then either program $P$ does not terminate, or it terminates and its output $(\sigma^\prime,\rho^\prime)$ satisfies cq-assertion $(\psi,B)$. 
\end{itemize}

\begin{exam}\label{QFT-correct} Let $|(j,k:l)\rangle$ and $|\mathit{QFT}(j,k:l)\rangle$ be defined as in Example \ref{QFT-data}. They can be seen as formal quantum states. Thus, 
\begin{align*}&A(j,k:l)\equiv|(j,k:l)\rangle_{q[k:l]}\langle(k,l)|,\\ &B(j,k:l)\equiv|\mathit{QFT}(j,k:l)\rangle_{q[k:l]}\langle\mathit{QFT}(j,k:l)|\end{align*} are two atomic quantum predicate formulas. All of these quantum states and predicates are parameterised by classical variables $j$, $k$ and $l$. 
Then the correctness of the program $\mathit{QFT}[q[1:n]]$ in Example \ref{QFT-program} can be specified by the following Hoare triple:
\begin{equation}\label{eq-QFT-correct}\{1\leq n,A(j,1:n)\}\ \mathit{QFT}[q[1:n]]\ \{\mathit{true},B(j,1:n)\}\end{equation} for any classical bit array $j$. 
\end{exam}

The following lemma asserts that both partial and total correctness are invariant under the equivalence of quantum programs.
\begin{lem}\label{lem-equivalence} If $P_1\approx P_2$ (see Definition \ref{def-equiv}), then for any precondition $(\varphi,A)$ and postcondition $(\psi,B)$, it holds that 
\begin{align}\label{lem-eq1}&\models_\mathit{par}\{\varphi,A\}\ P_1\ \{\psi,B\}\ \mbox{iff}\ \models_\mathit{par}\{\varphi,A\}\ P_2\ \{\psi,B\},\\
\label{lem-eq2}&\models_\mathit{tot}\{\varphi,A\}\ P_1\ \{\psi,B\}\ \mbox{iff}\ \models_\mathit{tot}\{\varphi,A\}\ P_2\ \{\psi,B\}.
\end{align}
\end{lem}
\begin{proof} By definitions of program equivalence and correctness and the linearity of trace (see \ref{proof-equivalence}).
\end{proof}

\subsection{Proof System}

Now we can present the proof system of our quantum Hoare logic with classical variables, denoted $\mathit{QHL}^+$. The axioms and proof rules for partial correctness are presented in Table \ref{proof-system}. The axioms and proof rules for total correctness are the same except that the rule for while-loop is changed to the rule (Rule-Loop-tot) in Table \ref{proof-system-1}. We write $\mathit{QHL}^+_\mathit{par}$ and $\mathit{QHL}^+_\mathit{tot}$ for the groups of axioms and proof rules for partial and total correctness, respectively.

\begin{table*}[!htbp]
\begin{equation*}\begin{split}&(\mbox{\textit{Axiom-Ski}})\quad\ \{\varphi,A\}\ \mathbf{skip}\ \{\varphi,A\}\\ &(\mbox{\textit{Axiom-Ass}})\quad\ \{\varphi[e/x],A[e/x]\}\ x:=e\ \{\varphi,A\}\\ 
&(\mbox{\textit{Axiom-Init}})\quad \left\{\varphi,F_B[q](A)\right\}\ q:=|0\rangle\ \{\varphi,A\}\\ 
&(\mbox{\textit{Axiom-Uni}})\quad 
\left\{\varphi,F_U(\overline{t})[\overline{q}](A)\right\}\ U(\overline{t})[\overline{q}]\ \{\varphi,A\}\\ 
&(\mbox{\textit{Axiom-Meas}})\quad \frac{y\notin\mathit{free}(\varphi)\cup\mathit{cv}(A)\cup\{x\}}{\left\{\varphi[y/x],F_M(y)[\overline{q}](A[y/x])\right\}\ x:=M[\overline{q}]\ \{\varphi\wedge x=y,A\}}\\ 
&(\mbox{\textit{Rule-Seq}})\qquad\ \ \ \frac{\{\varphi,A\}\ P_1\ \{\psi,B\}\qquad \{\psi,B\}\ P_2\ \{\theta,,C\}}{\{\varphi,A\}\ P_1;P_2\ \{\theta,C\}}\\
&(\mbox{\textit{Rule-Cond}})\qquad \frac{\{\varphi\wedge b,A\}\ P_1\ \{\psi,B\}\qquad \{\varphi\wedge\neg b, A\}\ P_0\ \{\psi,B\}}{\{\varphi, A\}\ \mathbf{if}\ b\ \mathbf{then}\ P_1\ \mathbf{else}\ P_0\ \{\psi,B\}}\\
&(\mbox{\textit{Rule-Loop-par}})\quad \frac{\{\varphi\wedge b,A\}\ P\ \{\varphi,A\}}{\{\varphi,A\}\ \mathbf{while}\ b\ \mathbf{do}\ P\ \{\varphi\wedge\neg b,A\}}\\
&(\mbox{\textit{Rule-Conseq}})\qquad\ \ \ \frac{(\varphi^\prime, A^\prime)\models (\varphi,A)\qquad \{\varphi,A\}\ P\ \{\psi,B\}\qquad (\psi,B)\models (\psi^\prime,B^\prime)}{\{\varphi^\prime,A^\prime\}\ P\ \{\psi^\prime,B^\prime\}}\\
&(\mbox{\textit{Rule-Accum1}})\ \frac{\begin{array}{cc}\{\varphi,A_i\}\ P\ \{\psi_i,B\}\ \mbox{for every}\ i\qquad (\forall i_1,i_2)\left(i_1\neq i_2\rightarrow\neg(\psi_{i_1}\wedge\psi_{i_2})\right)\\ \overline{q}\cap\mathit{qv}(P)=\emptyset\qquad [\mathit{var}(\overline{t})\cup\mathit{cv}(\overline{q})]\cap\mathit{change}(P)=\emptyset\qquad F(\overline{t})\propto F^\prime(\overline{t})\end{array}}{\left\{\varphi,F(\overline{t})[\overline{q}](\{A_i\})\right\}\ P\ \left\{\bigvee_{i=1}^n\psi_i,F^\prime(\overline{t})[\overline{q}](B)\right\}}\\
&(\mbox{\textit{Rule-Accum2}})\quad \frac{\begin{array}{cc}\{\varphi,A_i\}\ P\ \{\psi,B_i\}\ \mbox{for every}\ i\\ \overline{q}\cap\mathit{qv}(P)=\emptyset\qquad [\mathit{var}(\overline{t})\cup\mathit{cv}(\overline{q})]\cap\mathit{change}(P)=\emptyset\end{array}}{\left\{\varphi,F(\overline{t})[\overline{q}](\{A_i\})\right\}\ P\ \left\{\psi,F(\overline{t})[\overline{q}](\{B_i\})\right\}}
\end{split}\end{equation*}
\caption{Proof System $\mathit{QHL}^+_\mathit{par}$. In (Axiom-Init),  $B=\{|n\rangle\}$ is an orthonormal basis and $|0\rangle$ is a basis state in $B$. In (Axiom-Meas), $\mathit{free}(\varphi)$, $\mathit{cv}(A)$ stand for the set of free classical variables in classical first-order logical formula $\varphi$ and the set of classical variables in (the subscripts of) quantum predicate formula $A$, respectively. In (Rule-Accum1) and (Rule-Accum2), $\mathit{qv}(P)$ stands for the set of quantum variables in program $P$,
$\mathit{change}(P)$ for the set of classical variables that can be modified by program $P$, $\mathit{var}(\overline{t})$ for the set of (classical) variables in $\overline{t}$, and   
$\mathit{cv}(\overline{q})$ for the set of classical variables in the (possibly subscripted) quantum variables.}\label{proof-system}
\end{table*}

\begin{table*}[!htbp]
\begin{equation*}(\mbox{\textit{Rule-Loop-tot}})\qquad \frac{\begin{array}{ccc}\{\varphi\wedge b,A\}\ P\ \{\varphi,A\}\\
\{\varphi\wedge b\wedge t=z, I\}\ P\ \{t<z, I\}\\ \varphi\rightarrow t\geq 0
\end{array}}{\{\varphi,A\}\ \mathbf{while}\ b\ \mathbf{do}\ P\ \{\varphi\wedge\neg b,A\}}
\end{equation*}
\caption{Proof System $\mathit{QHL}^+_\mathit{tot}$. In (Rule-Loop-tot), $t$ is an integer expression, $z$ is an integer variable and $z\notin\mathit{free}(\varphi)\cup\mathit{var}(b)\cup\mathit{var}(t)\cup\mathit{cv}(P).$
The symbol $I$ stands for the identity quantum predicate on $\hs_{\mathit{qv}(P)}$.
}\label{proof-system-1}
\end{table*}

Most of the axioms and rules in our quantum Hoare logic $\mathit{QHL}^+_\mathit{par}$ and $\mathit{QHL}^+_\mathit{tot}$ look similar to the corresponding ones in classical Hoare logic.
But the following two ideas behind the design decision of our proof system are worth mentioning: \begin{itemize}
\item The substitution in quantum predicate formulas defined in Subsection \ref{sec-subst} are used in formulating (Axiom-Ass) and (Axiom-Meas), and Kraus operators on quantum predicates introduced in Subsections \ref{sec-pred-syn} and \ref{sec-pred-sem} are employed in formulating (Axiom-Init), (Axiom-Uni) and (Axiom-Meas). 
\item The design idea of (Axiom-Meas) is particularly interesting. Note that the measurement $M$ in the command $x:=M[\overline{q}]$ yields a probability distribution $\mathcal{D}=\{\mbox{Prob}[x=m]\}$ over possible outcomes $m$, 
where $\mbox{Prob}[x=m]$ stands for the probability that the measurement outcome is $m$. It was believed by the authors of \cite{FengY, Deng} that it is unavoidable to introduce a kind of substitution $[\mathcal{D}/x]$ of variable $x$ by probability distribution $\mathcal{D}$ in the precondition in order to properly formulate (Axiom-Meas). This will make this axiom much more complicated and harder to use in practical verification. In our logic $\mathit{QHL}^+$, this difficulty is circumvented by introducing equality $x=y$ in the postcondition, where $y$ is a fresh variable standing for a possible outcome. Thus, we only need to use the ordinary substitution $[y/x]$ of $x$ by a variable $y$ rather than by a probability distribution in the precondition.
Of course, there is a price for this simplification: the probabilities of different measurement outcomes stored in variable $y$ may need be accumulated in later reasoning, as discussed next. 
\item Since (Axiom-Meas) is designed so that different possible outcomes of a measurement can be treated separately, we need a mechanism to accumulated them together. For this purpose, (Rule-Accum1) is introduced in order to merge the postconditions for classical variables by the logical connective of disjunction. The formula $F(\overline{t})\propto F^\prime(\overline{t})$ in the premise of (Rule-Accum1) needs careful explanation. Here, $F(\overline{x})$ is a Kraus operator symbol of rank, say $k$ and $F^\prime(\overline{x})$ a Kraus operator symbol of rank $1$, and they have the same type, say $\hs_1\otimes ...\otimes\hs_n$. Given a classical structure and a quantum structure. Let $\sigma$ be a classical state. Then 
$\sigma\models F(\overline{t})\propto F^\prime(\overline{t})$ 
if for any $i\leq k$, we have: $$F_i^\dag(\sigma(\overline{t}))\cdot\rho\cdot F_i(\sigma(\overline{t}))\leq F^{\prime\dag}(\sigma(\overline{t}))\cdot\rho\cdot F^\prime(\sigma(\overline{t}))$$ for all density operators $\rho$ on $\hs_1\otimes ...\otimes\hs_n$. A special case of (Rule-Accum1) is the following rule:    
$$(\mbox{Rule-Convex1})\quad \frac{\begin{array}{cc}\{\varphi,A_i\}\ P\ \{\psi_i,B\}\ \mbox{for every}\ i\quad 0\leq p_i\ \mbox{for every}\ i\quad \sum_ip_i\leq 1\\ 
(\forall i_1,i_2)\left(i_1\neq i_2\rightarrow\neg(\psi_{i_1}\wedge\psi_{i_2})\right)
\end{array}}{\left\{\varphi,\sum_ip_iA_i\right\}\ P\ \left\{\bigvee_i\psi_i,\max_i p_i\cdot B\right\}}$$
\item As a complement to (Rule-Accum1), we propose (Rule-Accum2), which enables the combination of different preconditions and postconditions for quantum variables using a Kraus operator. A special case of (Rule-Accum2) is the following rule, which handles convex combinations of quantum predicates in the preconditions and postconditions, weighted according to a sub-probability distribution:
$$(\mbox{Rule-Convex2})\quad \frac{\{\varphi,A_i\}\ P\ \{\psi,B_i\}\ \mbox{for every}\ i\quad 0\leq p_i\ \mbox{for every}\ i\quad \sum_ip_i\leq 1}{\left\{\varphi,\sum_ip_iA_i\right\}\ P\ \left\{\psi,\sum_ip_iB_i\right\}}$$
\end{itemize}

It must be pointed out that there are some fundamental differences (e.g. nondeterminism caused by quantum measurements) between the semantics behind $\mathit{QHL}^+$ and classical Hoare logic, although their axioms and proof rules look similar. This should already been noticed through the discussions in the previous sections. It can be seen even more clearly through reading the proofs of the soundness theorem in \ref{proof-sound} and the (relative) completeness theorem in \cite{Ying24+}. 

A key advantage of the logic $\mathit{QHL}^+$ presented above over previous approaches (e.g. \cite{FengY}) is its ability to seamlessly integrate classical first-order logic into the specification and verification of quantum programs with classical variables (see also the discussions at the end of Subsection \ref{subsec-assertions}). The readers will recognise the value of this feature through their practical applications of $\mathit{QHL}^+$.

\begin{exam}\label{QFT-proof} We prove the correctness (\ref{eq-QFT-correct}) of the circuit realisation of quantum Fourier transform in the logic $\mathit{QHL}^+$. Here, we only present an outline of the proof, and some details are given in \ref{app-QFT}. From this example we can see that introducing the data structure of quantum array often enables us to verify quantum programs in a more convenient and elegant way. 

(1) We introduce atomic quantum predicate $$B^\ast(j,k:l)=|\mathit{QFT}^\ast(j,k:l)\rangle_{q[k:l]}\langle\mathit{QFT}^\ast(j,k:l)|$$ where formal quantum state $\mathit{QFT}^\ast(j,k:l)$ is inductively defined as follows:\begin{align*}\begin{cases}
|\mathit{QFT}^\ast(j,k:l)\rangle=\frac{1}{\sqrt{2}}\left(|0\rangle+e^{2\pi i 0.j[l]}|1\rangle\right)\ \mbox{if}\ k=l,\\ \mathit{QFT}^\ast(j,r-1:l)=\frac{1}{\sqrt{2}}\left(|0\rangle+e^{2\pi i 0.j[l-r:l]}|1\rangle\right)\otimes |\mathit{QFT}^\ast(j,r:l)\ \mbox{for}\ k+1\leq r\leq l. \end{cases}
\end{align*} Then using (Axiom-Uni) we can show that 
\begin{equation}\label{proof-reverse}\left\{\mathit{true},B^\ast(j,1:n)\right\}\ \mathit{Reverse}[q[1:n]]\ \{\mathit{true}, B(j,1:n)\}
\end{equation} (The detailed proof of (\ref{proof-reverse}) is deferred to \ref{app-QFT}).  

(2) Let $$|\varphi\rangle=\frac{1}{\sqrt{2}}|0\rangle+e^{2\pi i0.j[m]}|1\rangle).$$ Obviously, we have: 
$$\left\{m<n,|A[m]\rangle_{q[m]}\langle A[m]|\right\}\ H[q[m]]\ \left\{m<n,|\varphi\rangle_{q[m]}\langle\varphi|\right\}.$$ Since $A(j,m:n)=|A[m]\rangle_{q[m]}\langle A[m]|\otimes A(j,m+1:n)$, it holds that 
\begin{equation}\label{inductx-QFT}\{m<n,A(j,m:n)\}\ H[q[m]]\ \left\{m<n, |\varphi\rangle_{q[m]}\langle\varphi|\otimes A(j,m+1:n)\right\}.
\end{equation}

(3) Let $$|\psi\rangle=\frac{1}{\sqrt{2}}|0\rangle+e^{2\pi i0.j[m:n]}|1\rangle).$$ We can prove:\begin{equation}\label{inductxx-QFT}\begin{split}&\left\{m<n,|\varphi\rangle_{q[m]}\langle\varphi|\otimes A(j,m+1:n)\right\}\ \mathit{CR}[q[m:n]]\\ &\qquad\qquad\qquad\qquad\qquad\left\{m<n,|\psi\rangle_{q[m]}\langle\psi|\otimes A(j,m+1:n)\right\}\end{split}\end{equation} 
by induction on the length $n-m$ of quantum array $q[m:n]$ (The detailed proof of (\ref{inductxx-QFT}) can be found in \ref{app-QFT}). Using (Rule-Seq), a combination of (\ref{inductx-QFT}) and (\ref{inductxx-QFT}) yields
\begin{equation}\label{inductxxx-QFT}\{m<n,A(j,m:n)\}\ H[q[m]];\mathit{CR}[q[m:n]]\ \left\{m<n,|\psi\rangle_{q[m]}\langle\psi|\otimes A(j,m+1:n)\right\}.\end{equation} 

(4) Now we prove:\begin{equation}\label{induct-QFT}\{m\leq n,A(j,m:n)\}\ \mathit{QFT}^\ast[q[m:n]]\ \left\{\mathit{true},B^\ast(j,m:n)\right\}\end{equation}
by induction on the length $n-m$ of quantum array $q[m:n]$. First, it is obvious by definitions of $A(j,m:n)$ and $B^\ast(j,m:n)$ that 
\begin{equation}\label{induct0-QFT}\{m=n,A(j,m:n)\}\ H[q[n]]\ \left\{m=n,B^\ast(j,m:n)\right\}\models \left\{\mathit{true},B^\ast(j,m:n)\right\}.\end{equation} For the case of $m<n$, by the induction hypothesis we obtain: 
\begin{equation}\label{induct1-QFT}\{m+1\leq n,A(j,m+1:n)\}\ \mathit{QFT}^\ast[q[m+1:n]]\ \left\{\mathit{true},B^\ast(j,m+1:n)\right\}.\end{equation}
 By the inductive definition of $B^\ast(j,m:n)$ we have: 
$$B^\ast(j,m:n)=|\psi\rangle_{q[m]}\langle\psi|\otimes B^\ast(j,m+1:n).$$
Consequently, it holds that \begin{equation}\label{induct2-QFT}\left\{m<n,|\psi\rangle_{q[m]}\langle\psi|\otimes A(j,m+1:n)\right\}\ \mathit{QFT}^\ast[q[m+1:n]]\ \left\{\mathit{true},B^\ast(j,m:n)\right\}.\end{equation}
Then we can use (Rule-Seq) to combine (\ref{inductxxx-QFT}) and (\ref{induct2-QFT}) and thus derive: 
\begin{equation}\label{induct3-QFT}\{m<n,A(j,m:n)\}\ H[q[m]];\mathit{CR}[q[m:n]];\mathit{QFT}^\ast[q[m+1:n]]\ \{\mathit{true},B^\ast(j,m:n)\}.\end{equation}
Therefore, employing (Rule-Cond), we obtain:
\begin{equation}\label{final-QFT}
\{m\leq n,A(j,m:n)\}\ \mathit{QFT}^\ast[q[m:n]]\ \{\mathit{true},B^\ast(j,m:n)\}.\end{equation} from (\ref{induct0-QFT}), (\ref{induct3-QFT}) and the inductive definition of $\mathit{QFT}^\ast[q[m:n]]$. 

(5) Finally, using (Rule-Seq), we obtain (\ref{eq-QFT-correct}) from (\ref{proof-reverse}) and (\ref{final-QFT}). Thus, the correctness of the circuit implementation (Table \ref{QFT-table}) of quantum Fourier transform is verified. 
\end{exam}

\section{Soundness Theorem}\label{section-sound}

In this section, we present the soundness theorem of our quantum Hoare logic $\mathit{QHL}^+$.

\begin{thm}[Soundness]\label{thm-sound} For any quantum program $P$, classical first-order logical formulas $\varphi,\psi$, and quantum predicate formulas $A,B$, we have:
\begin{align*}&\vdash_{\mathit{QHL}^+_\mathit{par}} \{\varphi,A\}\ P\ \{\psi,B\}\ {\rm implies}\ \models_\mathit{par}\{\varphi,A\}\ P\ \{\psi,B\};\\
&\vdash_{\mathit{QHL}^+_\mathit{tot}}\{\varphi,A\}\ P\ \{\psi,B\}\ {\rm implies}\ \models_\mathit{tot}\{\varphi,A\}\ P\ \{\psi,B\}.
\end{align*}
\end{thm}
\begin{proof} The basic idea of the proof is the same as usual: we prove that every axiom in $\mathit{QHL}^+$ is valid, and the validity is preserved by each rule in $\mathit{QHL}^+$. But some calculations in the proof are quite involved. So, for readability, we defer the lengthy details of the proof into \ref{proof-sound}.
 \end{proof}
 
 The above soundness theorem warrants that if the correctness of a quantum program written in the programming language $\mathit{qWhile}^+$ can be proved in logic $\mathit{QHL}_\mathit{par}^+$ (respectively, $\mathit{QHL}_\mathit{tot}^+$), then it is indeed partially (respectively, totally) correct. Conversely, we can establish a (relative) completeness of our logic, which affirms that with a certain assumption about the expressive power of the assertion language, partial (respectively, total) correctness of quantum programs written in $\mathit{qWhile}^+$ can always be proved in $\mathit{QHL}_\mathit{par}^+$ (respectively, $\mathit{QHL}_\mathit{tot}^+$). But the presentation and proof of the (relative) completeness theorem require much more technical preparations and theoretical treatments. So, we postpone them to a companion paper \cite{Ying24+}. 

\section{Conclusion}\label{sec-concl}

In this paper, we develop a Hoare-style logic for quantum programs with classical variables, termed $\mathit{QHL}^+$. It incorporates the following key features: \begin{enumerate}\item[(i)] support for quantum arrays and parameterised quantum gates, enabling more versatile quantum programming; and \item[(ii)] a syntax for quantum predicates, allowing the construction of complex quantum predicate formulas from atomic predicates using logical connectives. These quantum predicate formulas can also be parameterised by classical variables. Consequently, preconditions and postconditions can be expressed as pairs consisting of a classical first-order logical formula (describing properties of classical variables) and a quantum predicate formula (describing properties of quantum variables).\end{enumerate}
The first feature enhances the convenience of quantum programming, while the second facilitates the derivation of a simple proof system for quantum programs through minimal modifications to classical Hoare logic. As a result, $\mathit{QHL}^+$ allows for more intuitive specification of quantum program correctness and enables more effective verification compared to previous approaches. In particular, classical first-order logic can be seamlessly integrated with $\mathit{QHL}^+$ for verifying quantum programs with  classical variables.

To conclude this paper, we would like to point out several topics for further research: \begin{enumerate}\item Formalisation of the logic $\mathit{QHL}^+$ for quantum programs with classical variables in proof assistants such as Coq, Isabelle/HOL, or Lean. 

\item Applications to the  verification of quantum algorithms, QEC (quantum error correction) codes and quantum communication and cryptographic protocols, leveraging the tools developed in point (1).

\item Extension of $\mathit{QHL}^+$ to parallel and distributed quantum programs \cite{Feng-D, Ying18} as well as its applications to relational reasoning \cite{Unruh19a, Barthe, Yangjia, Yu-cav}, local reasoning \cite{QSL1, QSL2}, and reasoning about incorrectness \cite{Yu-inc} of quantum programs.  

\item Broader applications of the quantum programming language $\mathit{qWhile}^+$ and the assertion language defined in Section \ref{sec-assert}, beyond verification. For instance, these could be utilised to develop more advanced techniques of abstract interpretation \cite{Yu-ab, abs-1, abs-2}, refinement \cite{QbC, Feng-refine} and symbolic execution \cite{Fang24, QSE1, QSE2, QSE3} for quantum programs. 
\end{enumerate}

\smallskip\

\textbf{Acknowledgement}: The author is very grateful to the referees for their invaluable comments and suggestions which help to improve the presentations significantly. The author also would like to thank the members of Tsinghua's quantum programming group for their valuable feedback, particularly for identifying several errors in the original version of this paper.

After this paper was published in \textit{Information and Computation} 309 (2026), Yuanjie Ren, Jinzheng Li, and Yidi Qi used their AI-agentic autoformalization framework to formalize the logic $\mathit{QHL}^+$ in Lean. During this process, a gap was found in the proof of the soundness theorem (Theorem \ref{thm-sound}) for the case of total correctness. The author is particularly grateful to them for bringing this gap to attention. In this revision, the gap has been fixed.

\appendix

\section{Appendix: Deferred Proofs}\label{app-proofs}

\subsection{Proof of Lemma \ref{lem-sub}}\label{proof-substitute}

We prove the substitution lemma by induction on the structure of quantum predicate formula $A$. 

{\vskip 3pt}

\textbf{Case 1}: $A=K(\overline{t})[\overline{q}]$. Then $A[e/x]=K(\overline{t}[e/x])[\overline{q}[e/x]]$. We assume that $$\overline{q}=q_1[s_{11},...,s_{1n_1}],...,q_k[s_{k1},...,s_{kn_k}].$$ Then 
\begin{align*}&\sigma(A[e/x])\ \mbox{is\ well-defined}\Leftrightarrow\sigma\models\mathit{Dist}(\overline{q}[e/x])\\ 
&\Leftrightarrow \sigma\models (\forall i,j\leq k)\left(q_i=q_j\rightarrow (\exists l\leq n_i)\left(s_{il}[e/x]\neq s_{jl}[e/x]\right)\right)\qquad ({\rm by\ Eq.(\ref{distinctness})})\\
&\qquad \ \ \ \ \ \ =((\forall i,j\leq k)\left(q_i=q_j\rightarrow (\exists l\leq n_i)\left(s_{il}\neq s_{jl}\right)\right))[e/x]\\
&\Leftrightarrow \sigma[x:=\sigma(e)]\models (\forall i,j\leq k)\left(q_i=q_j\rightarrow (\exists l\leq n_i)\left(s_{il}\neq s_{jl}\right)\right)\\ 
&\qquad\qquad\qquad\qquad\qquad\qquad\ \ \ \ \ \ (\mbox{by\ the\ substitution\ lemma\ in\ first-order\ logic})\\
&\Leftrightarrow \sigma[x:=\sigma(e)]\models\mathit{Dist}(\overline{q})\qquad\ \  ({\rm by\ Eq.(\ref{distinctness})})\\
&\Leftrightarrow \sigma[x:=\sigma(e)](A)\ \mbox{is\ well-defined}.
\end{align*}

Now assume that $\sigma(A[e/x])$ is well-defined. It follows from the substitution lemma for classical expressions that \begin{equation}\label{app-sub}\sigma[\overline{t}[e/x]]=\sigma[x:=\sigma(e)](\overline{t}).\end{equation} On the other hand, by the same lemma, it is easy to show that \begin{equation}\label{app-sub1}\sigma(\overline{q}[e/x])=\sigma[x:=\sigma(e)](\overline{q}).\end{equation} Therefore, by Definition \ref{def-sig}, we obtain:
\begin{align*}\sigma(A[e/x])&=K(\sigma(\overline{t}[e/x]))_{\sigma(\overline{q}[e/x])}\\
&=K(\sigma[x:=\sigma(e)](\overline{t}))_{\sigma[x:=\sigma(e)](\overline{q})}\\ 
&=\sigma[x:=\sigma(e)](A). 
\end{align*}

{\vskip 3pt}

\textbf{Case 2}: $A=\neg B$. Then $A[e/x]=\neg(B[e/x])$ and \begin{align*}\sigma(A[e/x])&=I-\sigma(B[e/x])\qquad\qquad\qquad\qquad ({\rm by\ Definition\ \ref{def-sig}})\\
&=I-\sigma[x:=\sigma(e)](B)\qquad ({\rm by\ the\ induction\ hypothesis\ on}\ B)\\ &=\sigma[x:=\sigma(e)](A).\qquad\qquad\qquad\qquad ({\rm by\ Definition\ \ref{def-sig}})
\end{align*}

{\vskip 3pt}

\textbf{Case 3}: $A=F(\overline{t})[\overline{q}](\{B_i\})$. Then $$A[e/x]=F(\overline{t}[e/x])[\overline{q}[e/x]](B_i[e/x])$$ and 
\begin{align*}&\sigma(A[e/x])=\sum_i F_i(\sigma(\overline{t}[e/x]))_{\sigma(\overline{q}[e/x])}\cdot \sigma(B_i[e/x])\cdot F_i^\dag(\sigma(\overline{t}[e/x]))_{\sigma(\overline{q}[e/x])}\\ &\qquad\qquad\qquad\qquad\qquad\qquad\qquad\qquad\qquad\qquad\qquad\qquad\qquad ({\rm by\ Definition\ \ref{def-sig}})\\
&=\sum_i F_i(\sigma[x:=\sigma(e)](\overline{t}))_{\sigma[x:=\sigma(e)](\overline{q})}\cdot \sigma[x:=\sigma(e)](B_i)\cdot F_i^\dag(\sigma[x:=\sigma(e)](\overline{t}))_{\sigma[x:=\sigma(e)](\overline{q})}\\ 
&\qquad\qquad\qquad\qquad\qquad\qquad ({\rm by\ Eqs. (\ref{app-sub})\ and\ (\ref{app-sub1})\ and\ the\ induction\ hypothesis\ on}\ B_i)\\ 
&=\sigma[x:=\sigma(e)](A). \qquad\qquad\qquad\qquad\qquad\qquad\qquad ({\rm by\ Definition\ \ref{def-sig}})
\end{align*}

\subsection{Proof of Lemma \ref{lem-equivalence}}\label{proof-equivalence}

Suppose that $P_1\approx P_2$. We only prove (\ref{lem-eq1}); (\ref{lem-eq2}) can be proved in a similar way. Assume cq-state $(\sigma,\rho)$ satisfies that $\sigma\models\varphi$ and $\sigma(A)$ is well-defined. Then for each $\sigma^\prime$, we set:
$$\Theta_i(\sigma^\prime)=\left\{|\rho^\prime\ |\ (\sigma^\prime,\rho^\prime)\in\llbracket P_i\rrbracket(\sigma,\rho)|\right\}$$ for $i=1,2$. By the definition of $P_1\approx P_2$, we have $\Theta_1(\sigma^\prime)=\Theta_2(\sigma^\prime)$. Then:
\begin{align*}\mathit{NT}(P_1)(\sigma,\rho)&=\tr(\rho)-\sum\left\{|\tr(\rho^\prime)\ |\ (\sigma^\prime,\rho^\prime)\in\llbracket P_1\rrbracket(\sigma,\rho)|\right\}\\ 
&=\tr(\rho)-\sum\left\{|\sum\left\{|\tr(\rho^\prime)\ |\ \rho^\prime\in \Theta_1(\sigma^\prime)|\right\}|\sigma^\prime\ \mbox{is a classical state}|\right\}\\ 
&=\tr(\rho)-\sum\left\{|\sum\left\{|\tr(\rho^\prime)\ |\ \rho^\prime\in \Theta_2(\sigma^\prime)|\right\}|\sigma^\prime\ \mbox{is a classical state}|\right\}\\ 
&=\tr(\rho)-\sum\left\{|\tr(\rho^\prime)\ |\ (\sigma^\prime,\rho^\prime)\in\llbracket P_2\rrbracket(\sigma,\rho)|\right\}\\ &=\mathit{NT}(P_2)(\sigma,\rho). 
\end{align*} 
Furthermore, let \textquotedblleft :WD\textquotedblright\ stand for \textquotedblleft is well-defined\textquotedblright. If $\models_\mathit{par}\{\varphi,A\}\ P_1\ \{\psi,B\}$, then:   
\begin{align*}\tr(\sigma(A)\rho)&\leq\sum\left\{|\tr(\sigma^\prime(B)\rho^\prime)\ |\ (\sigma^\prime,\rho^\prime)\in\llbracket P_1\rrbracket(\sigma,\rho), \sigma^\prime\models\psi\ \mbox{and}\ \sigma^\prime(B):WD|\right\}\\
&=\sum\left\{|\sum\left\{|\tr(\sigma^\prime(B)\rho^\prime)\ |\ \rho^\prime\in\Theta_1(\sigma^\prime)|\right\}\ |\ \sigma^\prime\models\psi\ \mbox{and}\ \sigma^\prime(B):WD|\right\}\\
&=\sum\left\{|\sum\left\{|\tr(\sigma^\prime(B)\rho^\prime)\ |\ \rho^\prime\in\Theta_2(\sigma^\prime)|\right\}\ |\ \sigma^\prime\models\psi\ \mbox{and}\ \sigma^\prime(B):WD|\right\}\\
&=\sum\left\{|\tr(\sigma^\prime(B)\rho^\prime)\ |\ (\sigma^\prime,\rho^\prime)\in\llbracket P_2\rrbracket(\sigma,\rho), \sigma^\prime\models\psi\ \mbox{and}\ \sigma^\prime(B):WD|\right\}
\end{align*} and $\models_\mathit{par}\{\varphi,A\}\ P_2\ \{\psi,B\}$. 

Conversely, we can prove that $\models_\mathit{par}\{\varphi,A\}\ P_2\ \{\psi,B\}$ implies $\models_\mathit{par}\{\varphi,A\}\ P_1\ \{\psi,B\}$. 

\subsection{Proof of Theorem \ref{thm-sound}}\label{proof-sound}

\subsubsection{The Case of Partial Correctness} We first prove the soundness of proof system $\mathit{QHL}^+_\mathit{par}$. It suffices to sow that every axiom in $\mathit{QHL}^+_\mathit{par}$ is valid in the sense of partial correctness, and the validity in the sense of partial correctness is preserved by each rule in $\mathit{QHL}^+_\mathit{par}$.

{\vskip 3pt}

(\textbf{Axiom-Ass}): We prove: $$\models_\mathit{par}\{\varphi[e/x],A[e/x]\}\ x:=e\ \{\varphi,A\}.$$ For any input state $(\sigma,\rho)\in\Omega$, by the transition rule (Ass), we have:
$$\llbracket x:=e\rrbracket(\sigma,\rho)=\{|(\sigma[x:=\sigma(e)],\rho)|\}.$$ 
If $\sigma\models\varphi[e/x]$ and $\sigma(A[e/x])$ is well-defined, then by the substitution lemma for classical first-order logic (Lemma 2.4 in \cite{Apt09}), it holds that $\sigma[x:=\sigma(e)]\models\varphi$. On the other hand, by Lemma \ref{lem-sub} (the substitution lemma for quantum predicates), we know that $\sigma[x:=\sigma(e)](A)$ is well-defined and $$\sigma[x:=\sigma(e)](A)=\sigma(A[e/x]).$$ Consequently, we have: 
\begin{align*}&\tr(\sigma(A[e/x])\rho)=\tr(\sigma[x:=\sigma(e)](A)\rho)\\
&=\sum\{|\tr(\sigma^\prime(A)\rho^\prime|(\sigma^\prime,\rho^\prime)\in\llbracket x:=e\rrbracket(\sigma,\rho), \sigma^\prime\models\varphi,\ \mbox{and}\ \sigma^\prime(A)\ \mbox{is\ well-defined}|\}. 
\end{align*}

{\vskip 3pt}

(\textbf{Axiom-Init}): We prove: $$\models_\mathit{par}\left\{\varphi, F_B[q](A)\right\}\ q:=|0\rangle\ \{\varphi,A\}.$$ For any input state $(\sigma,\rho)\in\Omega$, if $\sigma\models\varphi$ and $\sigma\left(F_B[q](A)\right)$ is well-defined, then according to the interpretation of modifier $F_B$ and clause (3) in Definition \ref{def-sig}, we know that $\sigma(A)$ is well-defined, and $$\sigma\left(F_B[q](A)\right)=\sum_n|n\rangle_{\sigma(q)}\langle 0|\sigma(A)|0\rangle_{\sigma(q)}\langle n|.$$
Therefore, we have: \begin{align}\label{proof-init-1}\tr\left(\sigma\left(F_B[q](A)\right)\rho\right)&=\sum_n\tr\left(|n\rangle_{\sigma(q)}\langle 0|\sigma(A)|0\rangle_{\sigma(q)}\langle n|\rho\right)\\
\label{proof-init-2}&=\sum_n\tr\left(\sigma(A)|0\rangle_{\sigma(q)}\langle n|\rho|n\rangle_{\sigma(q)}\langle 0|\right)\\
\label{proof-init-3}&=\tr\left(\sigma(A)\sum_n|0\rangle_{\sigma(q)}\langle n|\rho|n\rangle_{\sigma(q)}\langle 0|\right)\\ 
\label{proof-init-4}&=\sum\left\{|\tr(\sigma^\prime(A)\rho^\prime)|(\sigma^\prime,\rho^\prime)\in\llbracket q:=|0\rangle\rrbracket(\sigma,\rho)\ \mbox{and}\ \sigma^\prime(A)\ \mbox{is\ well-define}|\right\}
\end{align}
where (\ref{proof-init-1}) and (\ref{proof-init-3}) comes from the linearity of trace, (\ref{proof-init-2}) follows from the equality $\tr(AB)=\tr(BA)$, and (\ref{proof-init-4}) is from Proposition \ref{prop-structure}(2). 

{\vskip 3pt}

(\textbf{Axiom-Uni}): We prove: $$\models_\mathit{par}\left\{\varphi,F_U(\overline{t})[\overline{q}]\right\}\ U(\overline{t})[\overline{q}]\ \{\varphi,A\}.$$ 
For any input state $(\sigma,\rho)\in\Omega$, if $\sigma\models\varphi$ and $\sigma\left(F_U(\overline{t})[\overline{q}](A)\right)$ is well-defined, then by Definition \ref{def-sig}(3), we know that $\sigma\models\mathit{Dist}(\overline{q})$, $\sigma(A)$ is well-defined, and $$\sigma\left(F_U(\overline{t})[\overline{q}](A)\right)=U_\sigma^\dag \sigma(A)U_\sigma,$$ where $U_\sigma$ denotes the unitary operator $U(\sigma(\overline{t}))$ acting on quantum systems $\sigma(\overline{q})$. On the other hand, it follows from $\sigma\models\mathit{Dist}(\overline{q})$ and transition rule (Uni) that $$\llbracket U(\overline{t})[\overline{q}]\rrbracket(\sigma,\rho)=\left\{|(\sigma, U_\sigma\rho U_\sigma^\dag)|\right\}.$$ Therefore, we obtain:
\begin{align*}&\tr\left[\sigma(F_U(\overline{t})[\overline{q}](A))\rho\right]=\tr\left[(U_\sigma^\dag\sigma(A) U_\sigma)\rho\right]\\
&=\tr\left[\sigma(A)(U_\sigma\rho U_\sigma^\dag)\right]\\ 
&=\sum\{|\tr(\sigma^\prime(A)\rho^\prime|(\sigma^\prime,\rho^\prime)\in\llbracket U(\overline{t})[\overline{q}]\rrbracket(\sigma,\rho), \sigma^\prime\models\varphi,\ \mbox{and}\ \sigma^\prime(A)\ \mbox{is\ well-defined}|\}. 
\end{align*}

{\vskip 3pt}

(\textbf{Axiom-Meas}): For any classical variable $y\in\mathit{free}(\varphi)\cup\mathit{free}(A)\cup\{x\}$, we prove: $$\models_\mathit{par}\left\{\varphi[y/x],F_M(y)[\overline{q}](A[y/x])\right\}\ x:=M[\overline{q}]\ \{\varphi\wedge x=y,A\}.$$ For any input state $(\sigma,\rho)\in\Omega$, if $\sigma\models\varphi[y/x]$ and $\sigma(F_M(y)[\overline{q}](A[y/x]))$ is well-defined, then $\sigma\models\mathit{Dist}(\overline{q})$ and $\sigma(A[y/x])$ is well-defined. 
By transition rule (Meas), we have: 
\begin{equation}\llbracket x:=M[\overline{q}]\rrbracket(\sigma,\rho)=\left\{(\sigma[x:=m],M^\sigma_m\rho (M^\sigma_m)^\dag) | m\ {\rm ranges\ over\ all possible\ outcomes}\ m\right\}\end{equation} where $M_m^\sigma$ means that operator $M_m$ acting on quantum systems $\overline{q}$.

Now  Let us set $m=\sigma(y)\in T$ (the outcome type of measurement $M$), $\sigma_0=\sigma[x:=m]$ and $\rho_0=M_m^\sigma\rho(M_m^\sigma)^\dag$. Then it follows from the assumption $\sigma\models\varphi[y/x]$ and the substitution lemma for classical first-order logic (Lemma 2.4 in \cite{Apt09}) that 
\begin{equation}\label{eq-meas-sub}\sigma_0=\sigma[x:=\sigma(y)]\models\varphi.
\end{equation} On the other hand, we have: $$\sigma_0(x)=\sigma[x:=m](x)=m=\sigma(y)=\sigma_0(y).$$ Thus, $\sigma_0\models x=y$. Combining this with (\ref{eq-meas-sub}), we obtain $\sigma_0\models\varphi\wedge x=y$.     

Furthermore, we have: \begin{align*}
\sigma(F_M(y)[\overline{q}](A[y/x]))=&\left(M_{\sigma(y)}^\sigma\right)^\dag\sigma(A[y/x])M_{\sigma(y)}^\sigma\\ =&(M_m^\sigma)^\dag\sigma(A[y/x])M_m^\sigma\\ =&(M_m^\sigma)^\dag\sigma_0(A)M_m^\sigma
\end{align*} because it follows from the substitution lemma for quantum predicate formulas (Lemma \ref{lem-sub}) that $$\sigma(A[y/x])=\sigma[x:=\sigma(y)](A)=\sigma[x:=m](A)=\sigma_0(A).$$ In addition, from the fact that $\sigma(A[y/x])$ is well-defined and Lemma \ref{lem-sub}, we can assert that $\sigma_0(A)=\sigma[x:=\sigma(y)](A)$ is well-defined. Therefore, we have: \begin{align*}
&\tr\left[\sigma(F_M(y)[\overline{q}](A[y/x]))\rho\right]=\tr\left[\left((M_m^\sigma)^\dag \sigma_0(A) M_m^\sigma\right)\rho\right]\\ 
&=\tr\left[\sigma_0(A)\left(M_m^\sigma\rho(M_m^\sigma)^\dag\right)\right]\\ 
&=\tr(\sigma_0(A)\rho_0)\\
&=\sum\left\{|\tr(\sigma^\prime(A)\rho^\prime)|(\sigma^\prime,\rho^\prime)\in\llbracket x:=M[\overline{q}]\rrbracket(\sigma,\rho), \sigma^\prime\models\varphi\wedge x=y,\ \mbox{and}\ \sigma^\prime(A)\ \mbox{is\ well-defined} |\right\}. 
\end{align*}

{\vskip 3pt} 

(\textbf{Rule-Seq}): We prove that $\models_\mathit{par}\{\varphi,A\}\ P_1\ \{\psi,B\}$ and $\models_\mathit{par}\{\psi,B\}\ P_2\ \{\theta,C\}$ imply $$\models_\mathit{par}\{\varphi,A\}\ P_1;P_2\ \{\theta,C\}.$$ First, for any state $(\sigma,\rho)\in\Omega$, it is easy to show from transition rule (Seq) that \begin{equation}\label{eq-seq}\llbracket P_1;P_2\rrbracket(\sigma,\rho)=\llbracket P_2\rrbracket\left(\llbracket P_1\rrbracket(\sigma,\rho)\right)=\bigcup_{(\sigma^\prime,\rho^\prime)\in \llbracket P_1\rrbracket(\sigma,\rho)}\llbracket P_2\rrbracket(\sigma^\prime,\rho^\prime).
\end{equation} Now we set:\begin{align*}
&\mathcal{L}_1(\sigma,\rho)=\left\{|(\sigma^\prime,\rho^\prime)\in \llbracket P_1\rrbracket(\sigma,\rho) | \sigma^\prime\models\psi\ \mbox{and}\ \sigma^\prime(B)\ \mbox{is\ well-defined}\right|\},\\
&\mathcal{L}_2(\sigma^\prime,\rho^\prime)=\left\{|(\sigma^{\prime\prime},\rho^{\prime\prime})\in \llbracket P_2\rrbracket(\sigma^\prime,\rho^\prime) | \sigma^{\prime\prime}\models\theta\ \mbox{and}\ \sigma^{\prime\prime}(C)\ \mbox{is\ well-defined}|\right\}.
\end{align*} Using (\ref{eq-seq}), we can see that 
\begin{equation}\label{seq-mid}
\bigcup_{(\sigma^\prime,\rho^\prime)\in\mathcal{L}_1(\sigma,\rho)}\mathcal{L}_2(\sigma^\prime,\rho^\prime)\subseteq\left\{|(\sigma^{\prime\prime},\rho^{\prime\prime})\in\llbracket P_1;P_2\rrbracket(\sigma,\rho)|\sigma^{\prime\prime}\models\theta\ \mbox{and}\ \sigma^{\prime\prime}(C)\ \mbox{is\ well-define}|\right\}.
\end{equation}

By the assumptions $\models_\mathit{par}\{\varphi,A\}\ P_1\ \{\psi,B\}$ and $\models_\mathit{par}\{\psi,B\}\ P_2\ \{\theta,C\}$, we obtain:

{\vskip 3pt}

(\textbf{\textit{Claim 1}}) If $\sigma\models\varphi$ and $\sigma(A)$ is well-defined, then 
\begin{equation}\label{seq-1}\tr(\sigma(A)\rho)\leq\sum\left\{|\tr(\sigma^\prime(B)\rho^\prime)|(\sigma^\prime,\rho^\prime)\in\mathcal{L}_1(\sigma,\rho)|\right\}+\mathit{NT}(P_1)(\sigma,\rho).\end{equation}

{\vskip 3pt}

(\textbf{\textit{Claim 2}}) If $\sigma^\prime\models\psi$ and $\sigma(B)$ is well-defined, then 
\begin{equation}\label{seq-2}\tr(\sigma^\prime(B)\rho^\prime)\leq\sum\left\{|\tr(\sigma^{\prime\prime}(C)\rho^{\prime\prime})|(\sigma^{\prime\prime},\rho^{\prime\prime})\in\mathcal{L}_2(\sigma^\prime,\rho^\prime)|\right\}+\mathit{NT}(P_2)(\sigma^\prime,\rho^\prime).\end{equation}

Then by combining (\ref{seq-1}) and (\ref{seq-2}), we have: if $\sigma\models\varphi$ and $\sigma(A)$ is well-defined, then
\begin{equation}\label{seq-3}\begin{split}\tr(\sigma(A)\rho)&\leq\sum\left\{|\mu(\sigma^\prime,\rho^\prime)+\mathit{NT}(P_2)(\sigma^\prime,\rho^\prime)|(\sigma^\prime,\rho^\prime)\in\mathcal{L}_1(\sigma,\rho)|\right\}+\mathit{NT}(P_1)(\sigma,\rho)
\\ &= \sum\left\{|\mu(\sigma^\prime,\rho^\prime)|(\sigma^\prime,\rho^\prime)\in\mathcal{L}_1(\sigma,\rho)|\right\}\\ &\qquad\qquad + \left(\sum\left\{|\mathit{NT}(P_2)(\sigma^\prime,\rho^\prime)|(\sigma^\prime,\rho^\prime)\in\mathcal{L}_1(\sigma,\rho)|\right\} +\mathit{NT}(P_1)(\sigma,\rho)\right)\end{split}\end{equation} 
where $$\mu(\sigma^\prime,\rho^\prime)=\sum\left\{|\tr(\sigma^{\prime\prime}(C)\rho^{\prime\prime}|(\sigma^{\prime\prime},\rho^{\prime\prime})\in\mathcal{L}_2(\sigma^\prime,\rho^\prime) |\right\}.$$
It follows from (\ref{seq-mid}) that 
\begin{equation}\label{seq-4}\begin{split}&\sum\left\{|\mu(\sigma^\prime,\rho^\prime)|(\sigma^\prime,\rho^\prime)\in\mathcal{L}_1(\sigma,\rho)|\right\}\\ &=\sum\left\{|\tr(\sigma^{\prime\prime}(C)\rho^{\prime\prime}|(\sigma^{\prime\prime},\rho^{\prime\prime})\in\bigcup_{(\sigma^\prime,\rho^\prime)\in\mathcal{L}_1(\sigma,\rho)}\mathcal{L}_2(\sigma^\prime,\rho^\prime) |\right\}\\ &\leq\sum\left\{|\tr(\sigma^{\prime\prime}(C)\rho^{\prime\prime}| (\sigma^{\prime\prime},\rho^{\prime\prime})\in\llbracket P_1;P_2\rrbracket(\sigma,\rho), \sigma^{\prime\prime}\models\theta\ \mbox{and}\ \sigma^{\prime\prime}\ \mbox{is\ well-define}|\right\}.
\end{split}\end{equation}
On the other hand, by the definition of nontermination probability $\mathit{NT}(P)$ and (\ref{eq-seq}) we have: 
\begin{equation}\label{seq-5}\begin{split}&\sum\left\{|\mathit{NT}(P_2)(\sigma^\prime,\rho^\prime)|(\sigma^\prime,\rho^\prime)\in\mathcal{L}_1(\sigma,\rho)|\right\} +\mathit{NT}(P_1)(\sigma,\rho)\\ 
&\leq \sum\left\{|\mathit{NT}(P_2)(\sigma^\prime,\rho^\prime)|(\sigma^\prime,\rho^\prime)\in\llbracket P_1\rrbracket(\sigma,\rho)|\right\} +\mathit{NT}(P_1)(\sigma,\rho)\\ 
&= \sum\left\{|\tr(\rho^\prime)-\sum\left\{|\tr(\rho{^{\prime\prime}})|(\sigma^{\prime\prime},\rho^{\prime\prime})\in\llbracket P_2\rrbracket(\sigma^\prime,\rho^\prime)|\right\}|(\sigma^\prime,\rho^\prime)\in\llbracket P_1\rrbracket(\sigma,\rho)|\right\}\\ &\qquad\qquad\qquad\qquad\qquad\qquad\qquad\qquad\qquad
+\mathit{NT}(P_1)(\sigma,\rho)\\ 
&= \sum\left\{|\tr(\rho^\prime)|(\sigma^\prime,\rho^\prime)\in\llbracket P_1\rrbracket(\sigma,\rho)|\right\}\\ &\qquad -\sum\left\{|\tr(\rho{^{\prime\prime}})|(\sigma^{\prime\prime},\rho^{\prime\prime})\in\bigcup_{(\sigma^\prime,\rho^\prime)\in\llbracket P_1\rrbracket(\sigma,\rho)}\llbracket P_2\rrbracket(\sigma^\prime,\rho^\prime)|\right\}
+\mathit{NT}(P_1)(\sigma,\rho)\\ 
&= \sum\left\{|\tr(\rho^\prime)|(\sigma^\prime,\rho^\prime)\in\llbracket P_1\rrbracket(\sigma,\rho)|\right\}\\ &\qquad\qquad\qquad\qquad -\sum\left\{|\tr(\rho{^{\prime\prime}})|(\sigma^{\prime\prime},\rho^{\prime\prime})\in \llbracket P_1;P_2\rrbracket(\sigma,\rho)|\right\}
+\mathit{NT}(P_1)(\sigma,\rho)\\ 
&= \left(\sum\left\{|\tr(\rho^\prime)|(\sigma^\prime,\rho^\prime)\in\llbracket P_1\rrbracket(\sigma,\rho)|\right\}+\mathit{NT}(P_1)(\sigma,\rho)\right)\\ &\qquad\qquad\qquad\qquad\qquad\qquad -\sum\left\{|\tr(\rho{^{\prime\prime}})|(\sigma^{\prime\prime},\rho^{\prime\prime})\in \llbracket P_1;P_2\rrbracket(\sigma,\rho)|\right\}
\\ &=\tr(\rho)-\sum\left\{|\tr(\rho{^{\prime\prime}})|(\sigma^{\prime\prime},\rho^{\prime\prime})\in \llbracket P_1;P_2\rrbracket(\sigma,\rho)|\right\}\\
&=\mathit{NT}(P_1;P_2)(\sigma,\rho).
\end{split}\end{equation} 
Finally, plugging (\ref{seq-4}) and (\ref{seq-5}) into (\ref{seq-3}), we obtain: 
\begin{align*}\tr(\sigma(A)\rho)&\leq\sum\left\{|\tr(\sigma^{\prime\prime}(C)\rho^{\prime\prime}| (\sigma^{\prime\prime},\rho^{\prime\prime})\in\llbracket P_1;P_2\rrbracket(\sigma,\rho), \sigma^{\prime\prime}\models\theta\ \mbox{and}\ \sigma^{\prime\prime}\ \mbox{is\ well-defined}|\right\}\\ &\qquad\qquad\qquad\qquad\qquad\qquad\qquad +\mathit{NT}(P_1;P_2)(\sigma,\rho).
\end{align*} 
Therefore, it holds that $\models_\mathit{par}\{\varphi,A\}\ P_1;P_2\ \{\theta,C\}$. 

{\vskip 3pt}

(\textbf{Rule-Cond}): We prove that $\models_\mathit{par}\{\varphi\wedge b,A\}\ P_1\ \{\psi,B\}$ and $\models_\mathit{par}\{\varphi\wedge\neg b,A\}\ P_0\ \{\psi,B\}$ imply $$\models_\mathit{par}\{\varphi,A\}\ \mathbf{if}\ b\ \mathbf{then}\ P_1\ \mathbf{else}\ P_0\ \{\psi,B\}.$$ For simplicity of the presentation, let us write \textquotedblleft$\mathbf{if}...$\textquotedblright\ for \textquotedblleft$\mathbf{if}\ b\ \mathbf{then}\ P_1\ \mathbf{else}\ P_0$\textquotedblright. Then for any input state $\sigma,\rho)\in\Omega$, by transition rule (Cond), we obtain:
\begin{equation}\label{eq-cond}\llbracket\mathbf{if}...\rrbracket(\sigma,\rho)=\begin{cases}\llbracket P_1\rrbracket(\sigma,\rho)\ {\rm if}\ \sigma\models b,\\
\llbracket P_0\rrbracket(\sigma,\rho)\ {\rm if}\ \sigma\models \neg b. 
\end{cases}\end{equation}

Now assume that $\sigma\models\varphi$ and $\sigma(A)$ is well-defined. We consider the following two cases:

{\vskip 3pt}

\textbf{\textit{Case 1}}. $\sigma\models b$. Then it holds that $\sigma\models\varphi\wedge b$. By the assumption $\models_\mathit{par}\{\varphi\wedge b,A\}\ P_1\ \{\psi,B\}$, we have: \begin{align*}
&\tr(\sigma(A)\rho)\leq\sum\left\{|\tr(\sigma^\prime(A)\rho^\prime)|(\sigma^\prime,\rho^\prime)\in\llbracket P_1\rrbracket(\sigma,\rho), \sigma^\prime\models\psi\ \mbox{and}\ \sigma^\prime(B)\ \mbox{is\ well-defined} 
|\right\}\\ &\qquad\qquad\qquad\qquad\qquad\qquad\qquad\qquad+\mathit{NP}(P_1)(\sigma,\rho)\\
&\leq\sum\left\{|\tr(\sigma^\prime(A)\rho^\prime)|(\sigma^\prime,\rho^\prime)\in\llbracket\mathbf{if}...\rrbracket(\sigma,\rho), \sigma^\prime\models\psi\ \mbox{and}\ \sigma^\prime(B)\ \mbox{is\ well-defined}
|\right\}\\ &\qquad\qquad\qquad\qquad\qquad\qquad\qquad\qquad+\mathit{NT}(\mathbf{if}...)(\sigma,\rho). 
\end{align*}

{\vskip 3pt}

\textbf{\textit{Case 2}}. $\sigma\models\neg b$. Similarly, we can prove \begin{align*}
&\tr(\sigma(A)\rho)\leq\sum\left\{|\tr(\sigma^\prime(A)\rho^\prime)|(\sigma^\prime,\rho^\prime)\in\llbracket\mathbf{if}...\rrbracket(\sigma,\rho), \sigma^\prime\models\psi\ \mbox{and}\ \sigma^\prime(B)\ \mbox{is\ well-defined}
|\right\}\\ &\qquad\qquad\qquad\qquad\qquad\qquad\qquad\qquad+\mathit{NT}(\mathbf{if}...)(\sigma,\rho).\end{align*}

Finally, by combining the above two cases, we assert that $$\models_\mathit{par}\{\varphi,A\}\ \mathbf{if}\ b\ \mathbf{then}\ P_1\ \mathbf{else}\ P_0\ \{\psi,B\}.$$

{\vskip 3pt}

(\textbf{Rule-Loop-par}): Let us first introduce several notations. For any $n\geq 0$, we set:
\begin{align*}&\mathcal{P}_n(\sigma,\rho)=\\ & \ \left\{|(\sigma^\prime,\rho^\prime)|(\sigma,\rho)=(\sigma_0,\rho_0)\stackrel{P}{\Rightarrow}...\stackrel{P}{\Rightarrow}(\sigma_k,\rho_k)=(\sigma^\prime, \rho^\prime), k\leq n, \sigma_i\models b\ (0\leq i<k)\ \mbox{and}\ \sigma_k\models\neg b|\right\},\\ 
&\mathcal{Q}_n(\sigma,\rho)=\\ &\ \ \ \ \ \ \left\{|(\sigma^\prime,\rho^\prime)|(\sigma,\rho)=(\sigma_0,\rho_0)\stackrel{P}{\Rightarrow}...\stackrel{P}{\Rightarrow}(\sigma_n,\rho_n)=(\sigma^\prime,\rho^\prime)\ \mbox{and}\ \sigma_i\models b\ (0\leq i\leq n)|\right\}.
\end{align*} Also, to simplify the presentation, we write \textquotedblleft $\mathbf{while}$\textquotedblright\ for \textquotedblleft $\mathbf{while}\ b\ \mathbf{d}\ P$\textquotedblright\  and \textquotedblleft :WD\textquotedblright\ for \textquotedblleft is well-defined\textquotedblright. Then $\left\{\mathcal{P}_n\right\}$ is an increasing sequence, and Proposition \ref{prop-structure}(8) can be rewritten as 
\begin{equation}\label{P-infty}\llbracket\mathbf{while}\rrbracket(\sigma,\rho)=\bigcup_{n=0}^\infty\mathcal{P}_n(\sigma,\rho).\end{equation}

Now we assume:\begin{equation}\label{assume-loop}\models_\mathit{par}\{\varphi\wedge b,A\}\ P\ \{\varphi,A\}\end{equation} and prove:
\begin{equation}\label{concl-loop}\models_\mathit{par}\{\varphi,A\}\ \mathbf{while}\ \{\varphi\wedge\neg b,A\},\end{equation} that is, for any input state $(\sigma,\rho)\in\Omega$, if $\sigma\models\varphi$ and $\sigma(A): WD$, then
\begin{equation}\label{ineq-loop}\begin{split}
\tr(\sigma(A)\rho)\leq\sum&\left\{|\tr(\sigma^\prime(A)\rho^\prime)|(\sigma^\prime,\rho^\prime)\in\llbracket\mathbf{while}\rrbracket(\sigma,\rho), \sigma^\prime\models\varphi\wedge\neg b\ \mbox{and}\ \sigma^\prime(A):WD|\right\}\\ 
&\qquad\qquad\qquad\qquad\qquad\qquad\qquad\qquad+\mathit{NT}(\mathbf{while})(\sigma, \rho)
\end{split}\end{equation} where $$\mathit{NT}(\mathbf{while})(\sigma,\rho)=\tr(\rho)-\sum\left\{|\tr(\rho^\prime)|(\sigma^\prime,\rho^\prime)\in\llbracket\mathbf{while}\rrbracket(\sigma,\rho)|\right\}.$$In fact, whenever $\sigma\models\neg b$, then $\sigma\models\varphi\wedge\neg b$, $(\sigma,\rho)\in\mathcal{P}_0(\sigma,\rho)$, and   
\begin{equation*}\tr(\sigma(A)\rho)\leq\sum\left\{|\tr(\sigma^\prime(A)\rho^\prime)|(\sigma^\prime,\rho^\prime)\in\mathcal{P}_0(\sigma,\rho), \sigma^\prime\models\varphi\wedge\neg b\ \mbox{and}\ \sigma^\prime(A):WD|\right\}.\end{equation*}
Thus, (\ref{ineq-loop}) is correct. Otherwise, that is, $\sigma\models b$, then we prove: 

{\vskip 3pt}

\textbf{\textit{Claim}}: For any integer $n\geq 1$, it holds that 
\begin{equation}\label{ineq-loop-0}\tr(\sigma(A)\rho)\leq\Lambda_n+\Delta_n+\tr(\rho)-\Gamma_n-\Omega_n\end{equation}
where: 
\begin{equation*}\begin{split}
\Lambda_n&=\sum\left\{|\tr(\sigma^\prime(A)\rho^\prime)|(\sigma^\prime,\rho^\prime)\in\mathcal{P}_n(\sigma,\rho), \sigma^\prime\models\varphi\ \mbox{and}\ \sigma^\prime(A):WD|\right\},\\
\Delta_n&=\sum\left\{|\tr(\sigma^\prime(A)\rho^\prime)|(\sigma^\prime,\rho^\prime)\in\mathcal{Q}_n(\sigma,\rho), \sigma^\prime\models\varphi\ \mbox{and}\ \sigma^\prime(A):WD|\right\},\\
\Gamma_n&=\sum\left\{|\tr(\rho^\prime)|(\sigma^\prime,\rho^\prime)\in\mathcal{P}_n(\sigma,\rho)|\right\},\\
\Omega_n&=\sum\left\{|\tr(\rho^\prime)|(\sigma^\prime,\rho^\prime)\in\mathcal{Q}_n(\sigma,\rho)|\right\}.
\end{split}\end{equation*} 

We proceed by induction on $n$. For the case of $n=1$, note that $\sigma\models\varphi\wedge b$. Then by the assumption (\ref{assume-loop}) we have: 
\begin{equation}\label{ineq-loop-1}\begin{split}&tr(\sigma(A)\rho)\leq\sum\left\{|\tr(\sigma^\prime(A)\rho^\prime)|(\sigma^\prime,\rho^\prime)\in\llbracket P\rrbracket(\sigma,\rho), \sigma^\prime\models\varphi\ \mbox{and}\ \sigma^\prime(A):WD|\right\}\\ 
&\qquad\qquad\qquad\qquad\qquad\qquad\qquad+\mathit{NT}(P)(\sigma,\rho)
\\ &=\sum\left\{|\tr(\sigma^\prime(A)\rho^\prime)|(\sigma^\prime,\rho^\prime)\in\llbracket P\rrbracket(\sigma,\rho), \sigma^\prime\models\varphi\ \mbox{and}\ \sigma^\prime(A):WD|\right\}\\ &\qquad\qquad\qquad\qquad
+\tr(\rho)-\sum\left\{|\tr(\rho^\prime)|(\sigma^\prime,\rho^\prime)\in\llbracket P\rrbracket(\sigma,\rho)|\right\}\\ 
&=\sum\left\{|\tr(\sigma^\prime(A)\rho^\prime)|(\sigma^\prime,\rho^\prime)\in\mathcal{P}_1(\sigma,\rho), \sigma^\prime\models\varphi\ \mbox{and}\ \sigma^\prime(A):WD|\right\}\\
&\qquad+\sum\left\{|\tr(\sigma^\prime(A)\rho^\prime)|(\sigma^\prime,\rho^\prime)\in\mathcal{Q}_1(\sigma,\rho), \sigma^\prime\models\varphi\ \mbox{and}\ \sigma^\prime(A):WD|\right\}
\\ &\ \ \ \ \ \ \ \ +\tr(\rho)-\sum\left\{|\tr(\rho^\prime)|(\sigma^\prime,\rho^\prime)\in\mathcal{P}_1(\sigma,\rho)|\right\}-\sum\left\{|\tr(\rho^\prime)|(\sigma^\prime,\rho^\prime)\in\mathcal{Q}_1(\sigma,\rho)|\right\}\\ &=\Lambda_1+\Delta_1+\tr(\rho)-\Gamma_1-\Omega_1
\end{split}\end{equation} because $\llbracket P\rrbracket(\sigma,\rho)=\mathcal{P}_1(\sigma,\rho)\cup\mathcal{Q}_1(\sigma,\rho)$ and $\mathcal{P}_1(\sigma,\rho)\cap\mathcal{Q}_1(\sigma,\rho)=\emptyset$.
So, (\ref{ineq-loop-0}) is true for $n=1$. 

Now assume (\ref{ineq-loop-0})  is true for $n$, and we prove it is also true for $n+1$. We first note that for any $(\sigma^\prime,\rho^\prime)\in\mathcal{Q}_n(\sigma,\rho)$ with $\sigma^\prime\models\varphi$ and $\sigma^\prime(A):WD$, it holds that $\sigma^\prime\models\varphi\wedge b$. Then by by the assumption (\ref{assume-loop}) we obtain: 
\begin{equation}\label{ineq-loop-2}\begin{split}&\tr(\sigma^\prime(A)\rho^\prime)\leq\sum\left\{|\tr(\sigma^{\prime\prime}(A)\rho^{\prime\prime})|(\sigma^{\prime\prime},\rho^{\prime\prime})\in\llbracket P\rrbracket(\sigma^\prime,\rho^\prime), \sigma^{\prime\prime}\models\varphi\ \mbox{and}\ \sigma{^{\prime\prime}}(A):WD|\right\}\\ &\qquad\qquad\qquad\qquad\qquad\qquad\qquad +\mathit{NT}(P)(\sigma^\prime,\rho^\prime)\\
&=\sum\left\{|\tr(\sigma^{\prime\prime}(A)\rho^{\prime\prime})|(\sigma^{\prime\prime},\rho^{\prime\prime})\in\llbracket P\rrbracket(\sigma^\prime,\rho^\prime), \sigma^{\prime\prime}\models\varphi\ \mbox{and}\ \sigma{^{\prime\prime}}(A):WD|\right\}\\ 
&\qquad\qquad\qquad\qquad\qquad\qquad\qquad +\mathit{NT}(P)(\sigma^\prime,\rho^\prime)\\
&=\sum\left\{|\tr(\sigma^{\prime\prime}(A)\rho^{\prime\prime})|(\sigma^{\prime\prime},\rho^{\prime\prime})\in\mathcal{P}_1(\sigma^\prime,\rho^\prime), \sigma^{\prime\prime}\models\varphi\ \mbox{and}\ \sigma^{\prime\prime}(A):WD|\right\}\\
&\qquad+\sum\left\{|\tr(\sigma^{\prime\prime}(A)\rho^{\prime\prime})|(\sigma^{\prime\prime},\rho^{\prime\prime})\in\mathcal{Q}_1(\sigma^\prime,\rho^\prime), \sigma^{\prime\prime}\models\varphi\ \mbox{and}\ \sigma^{\prime\prime}(A):WD|\right\}\\
&\qquad\qquad\qquad\qquad\qquad\qquad\qquad +\mathit{NT}(P)(\sigma^\prime,\rho^\prime)\end{split}.\end{equation}
Taking the summation over $(\sigma^\prime,\rho^\prime)\in\mathcal{Q}_n(\sigma,\rho)$ with $\sigma^\prime\models\varphi$ and $\sigma^\prime(A):WD$, we obtain:
\begin{equation}\label{ineq-loop-21}\begin{split}\Delta_n&\leq\lambda_n+\delta_n+\sum\left\{|\mathit{NT}(P)(\sigma^\prime,\rho^\prime)|(\sigma^\prime,\rho^\prime)\in\mathcal{Q}_n(\sigma,\rho), \sigma^\prime\models\varphi\ \mbox{and}\ \sigma^\prime(A):WD|\right\}\\ &\leq\lambda_n+\delta_n+\sum\left\{|\mathit{NT}(P)(\sigma^\prime,\rho^\prime)|(\sigma^\prime,\rho^\prime)\in\mathcal{Q}_n(\sigma,\rho)|\right\}
\end{split}\end{equation} 
where:
\begin{align*}&\lambda_n=\\ &\sum\left\{|\sum\left\{|\tr(\sigma^{\prime\prime}(A)\rho^{\prime\prime})|(\sigma^{\prime\prime},\rho^{\prime\prime})\in\mathcal{P}_1(\sigma^\prime,\rho^\prime), \sigma^{\prime\prime}\models\varphi\ \mbox{and}\ \sigma^{\prime\prime}(A):WD|\right\}|(\sigma^\prime,\rho^\prime)\in\mathcal{Q}_n(\sigma,\rho)|\right\},\\ 
&\delta_n=\\ &\sum\left\{|\sum\left\{|\tr(\sigma^{\prime\prime}(A)\rho^{\prime\prime})|(\sigma^{\prime\prime},\rho^{\prime\prime})\in\mathcal{Q}_1(\sigma^\prime,\rho^\prime), \sigma^{\prime\prime}\models\varphi\ \mbox{and}\ \sigma^{\prime\prime}(A):WD|\right\}|(\sigma^\prime,\rho^\prime)\in\mathcal{Q}_n(\sigma,\rho)|\right\}.\end{align*}
Plugging (\ref{ineq-loop-21}) into (\ref{ineq-loop-0}) yields: 
\begin{equation}\label{ineq-loop-3}\begin{split}&\tr(\sigma(A)\rho)\leq \Lambda_n+\left (\lambda_n+\delta_n+\sum\left\{|\mathit{NT}(P)(\sigma^\prime,\rho^\prime)|(\sigma^\prime,\rho^\prime)\in\mathcal{Q}_n(\sigma,\rho)|\right\}
\right)+\tr(\rho)-\Gamma_n-\Omega_n\\ 
&=\left (\Lambda_n+\lambda_n\right)+\delta_n+\tr(\rho)+\left(\sum\left\{|\mathit{NT}(P)(\sigma^\prime,\rho^\prime)|(\sigma^\prime,\rho^\prime)\in\mathcal{Q}_n(\sigma,\rho)|\right\}
-\Gamma_n-\Omega_n\right).  
\end{split}\end{equation}
It is easy to see that \begin{align}\label{loop-par1}&\Lambda_n+\lambda_n\leq \sum\left\{|\tr(\sigma^\prime(A)\rho^\prime)|(\sigma^\prime,\rho^\prime)\in\mathcal{P}_{n+1}(\sigma,\rho), \sigma^\prime\models\varphi\ \mbox{and}\ \sigma^\prime(A):WD|\right\}=\Lambda_{n+1} \end{align}
and \begin{align}\label{loop-par2}\delta_n=
\sum\left\{|\tr(\sigma^{\prime\prime}(A)\rho^{\prime\prime})|(\sigma^{\prime\prime},\rho^{\prime\prime})\in\mathcal{Q}_{n+1}(\sigma,\rho), \sigma^{\prime\prime}\models\varphi\ \mbox{and}\ \sigma^{\prime\prime}(A):WD|\right\}=\Delta_{n+1}.
\end{align} 
By the definition of $\mathit{NT}(P)$, we have:
\begin{align*}&\sum\left\{|\mathit{NT}(P)(\sigma^\prime,\rho^\prime)|(\sigma^\prime,\rho^\prime)\in\mathcal{Q}_n(\sigma,\rho)|\right\}\\ 
&=\sum\left\{|\tr(\rho^\prime)-\sum\left\{|\tr(\rho^{\prime\prime})|(\sigma^{\prime\prime},\rho^{\prime\prime})\in\llbracket P\rrbracket(\sigma^\prime,\rho^\prime)|\right\}|(\sigma^\prime,\rho^\prime)\in\mathcal{Q}_n(\sigma,\rho)|\right\}\\
&=\sum\left\{|\tr(\rho^\prime)|(\sigma^\prime,\rho^\prime)\in\mathcal{Q}_n(\sigma,\rho)|\right\}\\ 
&\qquad\qquad\qquad -\sum\left\{|\tr(\rho^{\prime\prime})|(\sigma^{\prime\prime},\rho^{\prime\prime})\in\llbracket P\rrbracket(\sigma^\prime,\rho^\prime)\ \mbox{and}\ (\sigma^\prime,\rho^\prime)\in\mathcal{Q}_n(\sigma,\rho)|\right\}\\
 &=\Omega_n-\sum\left\{|\tr(\rho^{\prime\prime})|(\sigma^{\prime\prime},\rho^{\prime\prime})\in\llbracket P\rrbracket(\sigma^\prime,\rho^\prime)\ \mbox{and}\ (\sigma^\prime,\rho^\prime)\in\mathcal{Q}_n(\sigma,\rho)|\right\}
\end{align*} Moreover, we have:
\begin{equation*}\begin{split}&\sum\left\{|\tr(\rho^{\prime\prime})|(\sigma^{\prime\prime},\rho^{\prime\prime})\in\llbracket P\rrbracket(\sigma^\prime,\rho^\prime)\ \mbox{and}\ (\sigma^\prime,\rho^\prime)\in\mathcal{Q}_n(\sigma,\rho)|\right\}
 \\ &=\sum\left\{|\tr(\rho^{\prime\prime})|(\sigma^{\prime\prime},\rho^{\prime\prime})\in\mathcal{P}_1(\sigma^\prime,\rho^\prime)\ \mbox{and}\ (\sigma^\prime,\rho^\prime)\in\mathcal{Q}_n(\sigma,\rho)|\right\}\\ &\qquad\qquad\qquad +\sum\left\{|\tr(\rho^{\prime\prime})|(\sigma^{\prime\prime},\rho^{\prime\prime})\in\mathcal{Q}_1(\sigma^\prime,\rho^\prime)\ \mbox{and}\ (\sigma^\prime,\rho^\prime)\in\mathcal{Q}_n(\sigma,\rho)|\right\}
  \\ &=\sum\left\{|\tr(\rho^{\prime\prime})|(\sigma^{\prime\prime},\rho^{\prime\prime})\in\mathcal{P}_1(\sigma^\prime,\rho^\prime)\ \mbox{and}\ (\sigma^\prime,\rho^\prime)\in\mathcal{Q}_n(\sigma,\rho)|\right\}+\Omega_{n+1}.
\end{split}\end{equation*}
Then it follows that 
\begin{equation}\label{loop-pars}\begin{split}&\sum\left\{|\mathit{NT}(P)(\sigma^\prime,\rho^\prime)|(\sigma^\prime,\rho^\prime)\in\mathcal{Q}_n(\sigma,\rho)|\right\}-\Gamma_n-\Omega_n\\ &=-\left(\sum\left\{|\tr(\rho^{\prime\prime})|(\sigma^{\prime\prime},\rho^{\prime\prime})\in\mathcal{P}_1(\sigma^\prime,\rho^\prime)\ \mbox{and}\ (\sigma^\prime,\rho^\prime)\in\mathcal{Q}_n(\sigma,\rho)|\right\}+\Gamma_n\right)-\Omega_{n+1}\\ &=-\Gamma_{n+1}-\Omega_{n+1}. 
\end{split}\end{equation} Combining (\ref{ineq-loop-3}), (\ref{loop-par1}), (\ref{loop-par2}) and (\ref{loop-pars}), we obtain:
$$\tr(\sigma(A)\rho)\leq\Lambda_{n+1}+\Delta_{n+1}+\tr(\rho)-\Gamma_{n+1}-\Omega_{n+1}.$$
Therefore, we see from (\ref{ineq-loop-3}) that (\ref{ineq-loop-0}) is true for $n+1$. Thus, the proof of (\ref{ineq-loop-0}) is completed. 

 Now we prove the following two facts:

{\vskip 3pt}

\textbf{\textit{Fact 1}}: By definition, we know that $\sigma^\prime\models\neg b$ whenever $(\sigma^\prime,\rho^\prime)\in\mathcal{P}_n(\sigma,\rho)$. Then it follows from (\ref{P-infty}) that \begin{equation}\label{ineq-loop-14}\begin{split}
&\lim_{n\rightarrow\infty}\Lambda_n =\sum\left\{|\tr(\sigma^\prime(A)\rho^\prime)|(\sigma^\prime,\rho^\prime)\in\llbracket\mathbf{while}\rrbracket(\sigma,\rho), \sigma^\prime\models\varphi\wedge\neg b\ \mbox{and}\ \sigma^\prime(A):WD|\right\}.
\end{split}\end{equation}

{\vskip 3pt}

\textbf{\textit{Fact 2}}: Note that $\sigma^\prime(A)\leq I$ (the identity operator). Then it follows that $\tr\left(\sigma^\prime(A)\rho\right)\leq\tr(\rho)$. Consequently, we have:
\begin{align*}\Delta_n\leq\sum\left\{|\tr(\sigma^\prime(A)\rho^\prime)|(\sigma^\prime,\rho^\prime)\in\mathcal{Q}_n(\sigma,\rho)|\right\}
\leq \sum\left\{|\tr(\rho^\prime)|(\sigma^\prime,\rho^\prime)\in\mathcal{Q}_n(\sigma,\rho)|\right\}
\leq \Omega_n
\end{align*} and $$\Delta_n-\Gamma_n-\Omega_n\leq -\Gamma_n.$$ Then it holds that 
\begin{equation}\label{lim-loop}\lim_{n\rightarrow\infty}(\Delta_n-\Gamma_n-\Omega_n)\leq -\lim_{n\rightarrow\infty}\Gamma_n=-\sum\left\{|\tr(\rho^\prime)|(\sigma^\prime,\rho^\prime)\in\llbracket\mathbf{while}\rrbracket(\sigma,\rho)|\right\}.\end{equation}

Finally, we take the limit of $n\rightarrow\infty$ in (\ref{ineq-loop-0}) and invoke (\ref{ineq-loop-14}) and (\ref{lim-loop}). Then the conclusion (\ref{ineq-loop}) is achieved, and we complete the proof of the soundness of (Rule-Loop-par). 

{\vskip 3pt}

(\textbf{Rule-Conseq}): We assume that $(\varphi^\prime,A^\prime)\models(\varphi,A)$, $(\psi,B)\models(\psi^\prime,B^\prime)$ and \begin{equation}\label{assume-con}\models_\mathit{par}\{\varphi,A\}\ P\ \{\psi,B\}\end{equation} and prove 
\begin{equation}\label{assume-con-1}\models_\mathit{par}\{\varphi^\prime,A^\prime\}\ P\ \{\psi^\prime,B^\prime\}.\end{equation}
For any input state $(\sigma,\rho)\in\Omega$, if $\sigma\models\varphi^\prime$ and $\sigma(A^\prime)$ is well-defined, then by the assumption $(\varphi^\prime,A^\prime)\models(\varphi,A)$ and Definitions \ref{qp-entail} and \ref{cq-entail}, we have $\sigma\models\varphi$, $\sigma(A)$ is well-defined, and $\sigma(A^\prime)\leq\sigma(A)$. Consequently, it follows from assumptions (\ref{assume-con}) and $(\psi,B)\models(\psi^\prime,B^\prime)$ that 
\begin{align*}&\tr(\sigma(A^\prime)\rho)\leq\tr(\sigma(A)\rho)\\ &\leq\sum\left\{|\tr(\sigma^{\prime\prime}(B)\rho^{\prime\prime})|(\sigma^{\prime\prime},\rho^{\prime\prime})\in\llbracket P\rrbracket(\sigma,\rho),\sigma^{\prime\prime}\models\psi\ \mbox{and}\ \sigma^{\prime\prime}(B)\ \mbox{is\ well-defined}|\right\}\\ &\leq\sum\left\{|\tr(\sigma^{\prime\prime}(B^\prime)\rho^{\prime\prime})|(\sigma^{\prime\prime},\rho^{\prime\prime})\in\llbracket P\rrbracket(\sigma,\rho),\sigma^{\prime\prime}\models\psi^\prime\ \mbox{and}\ \sigma^{\prime\prime}(B^\prime)\ \mbox{is\ well-defined}|\right\}.
\end{align*}Thus, (\ref{assume-con-1}) is proved. 

{\vskip 3pt}

(\textbf{Rule-Accum1}): We assume that $\overline{q}\cap\mathit{qv}(P)=\emptyset$ and \begin{align}\label{accum1}\models_\mathit{par}\{\varphi,A_i\}\ P\ \{\psi_i,B\}\ \mbox{for every}\ i,\\
\label{accum2}\models(\forall i_1,i_2)(i_1\neq i_2\rightarrow\neg (\psi_{i_1}\wedge\psi_{i_2})). 
\end{align} 

Let us write \textquotedblleft:WD\textquotedblright\ for \textquotedblleft is well-defined\textquotedblright. 
Then for any input state $(\sigma,\rho)\in\Omega$, whenever $\sigma\models\varphi$ and $\sigma(A_i)$ is well-defined, then: 
\begin{equation}\label{Kraus-proof1}\begin{split}\tr(\sigma(A_i)\rho)\leq\sum &\left\{|\tr(\sigma^\prime(B)\rho^\prime)|(\sigma^\prime,\rho^\prime)\in\llbracket P\rrbracket(\sigma,\rho),\sigma^\prime\models\psi_i\ \mbox{and}\ \sigma^\prime(B):WD\right\}\\ &\qquad + \mathit{NT}(P)(\sigma,\rho).
\end{split}\end{equation}
We want to prove \begin{equation}\label{conv2}\models_\mathit{par}\left\{\varphi,F(\overline{t})[\overline{q}](\{A_i\})\right\}\ P\ \left\{\bigvee_i\psi_i,F(\overline{t})[\overline{q}](B)\right\}.\end{equation}
For any input state $(\sigma,\rho)\in\Omega$, if $\sigma\models \varphi$ and $\sigma\left(F(\overline{t})[\overline{q}](\{A_i\})\right)$ is well-defined, then for all $i$, $\sigma(A_i)$ is well-defined, and $\sigma\models\mathit{Dist}(\overline{q})$. Assume that the Kraus operator symbol $F$ is interpreted as $F=\left\{F_i(\overline{a})\right\}$. Then we have:
\begin{equation}\label{Kraus-proof2}\begin{split}&\tr\left[\sigma\left(F(\overline{t})[\overline{q}](\{A_i\})\right)\rho\right]=\tr\left[\left(\sum_iF_i(\sigma(\overline{t}))_{\sigma(\overline{q})}\cdot\sigma(A_i)\cdot F^\dag_i(\sigma(\overline{t}))_{\sigma(\overline{q})}\right)\rho\right]\\
&=\sum_i\tr\left[\left(F_i(\sigma(\overline{t}))_{\sigma(\overline{q})}\cdot\sigma(A_i)\cdot F^\dag_i(\sigma(\overline{t})_{\sigma(\overline{q})}\right)\rho\right]\\ 
&=\sum_i\tr\left[\sigma(A_i)\left(F^\dag_i(\sigma(\overline{t}))_{\sigma(\overline{q})}\cdot\rho\cdot F_i(\sigma(\overline{t}))_{\sigma(\overline{q})}\right)\right]\\ 
&\leq\sum_i\Lambda_i\end{split}\end{equation}
where: 
\begin{align*}&\Lambda_i=\\ &\sum \left\{|\tr(\sigma^\prime(B)\rho^\prime)|(\sigma^\prime,\rho^\prime)\in\llbracket P\rrbracket\left(\sigma, F^\dag_i(\sigma(\overline{t}))_{\sigma(\overline{q})}\cdot\rho\cdot F_i(\sigma(\overline{t}))_{\sigma(\overline{q})}\right), \sigma^\prime\models\psi_i\ \mbox{and}\ \sigma^\prime(B):WD|\right\}\\ & +\mathit{NT}(P)\left(\sigma, F^\dag_i(\sigma(\overline{t}))_{\sigma(\overline{q})}\cdot\rho\cdot F_i(\sigma(\overline{t}))_{\sigma(\overline{q})}\right).\end{align*}
Note that the last inequality of (\ref{Kraus-proof2}) comes from the assumption (\ref{Kraus-proof1}). 

Now with the assumption that $\overline{q}\cap\mathit{qv}(P)=\emptyset$ and $[\mathit{var}(\overline{t})\cup\mathit{cv}(\overline{q})]\cap\mathit{change}(P)=\emptyset$, we can prove:
\begin{equation}\label{Kraus-proof3}\begin{split}\llbracket P\rrbracket &\left(\sigma, F^\dag_i(\sigma(\overline{t}))_{\sigma(\overline{q})}\cdot\rho\cdot F_i(\sigma(\overline{t}))_{\sigma(\overline{q})}\right)\\ &\qquad =\left\{|\left(\sigma^\prime, F^\dag_i(\sigma(\overline{t}))_{\sigma(\overline{q})}\cdot\rho^\prime\cdot F_i(\sigma(\overline{t}))_{\sigma(\overline{q})}\right)|(\sigma^\prime,\rho^\prime)\in \llbracket P\rrbracket(\sigma,\rho)|\right\}. 
\end{split}\end{equation} by induction on the structure of $P$. 

(1) By the definition of $\mathit{NT}(\cdot)$ we obtain:
\begin{equation}\label{Kraus-proof4}\begin{split}
\sum_i\mathit{NT}(P)\left(\sigma, F^\dag_i(\sigma(\overline{t}))_{\sigma(\overline{q})}\cdot\rho\cdot F_i(\sigma(\overline{t}))_{\sigma(\overline{q})}\right)=\sum_i\Delta_i\end{split}\end{equation}
where 
\begin{align*}\Delta_i=\tr&\left(F^\dag_i(\sigma(\overline{t}))_{\sigma(\overline{q})}\cdot\rho\cdot F_i(\sigma(\overline{t}))_{\sigma(\overline{q})}\right)\\ &-\sum\left\{|\tr(\rho^{\prime\prime})|(\sigma^\prime,\rho^{\prime\prime})\in\llbracket P\rrbracket\left(\sigma,F^\dag_i(\sigma(\overline{t}))_{\sigma(\overline{q})}\cdot\rho\cdot F_i(\sigma(\overline{t}))_{\sigma(\overline{q})}\right)|\right\}.\end{align*}
Using (\ref{Kraus-proof3}) we have: \begin{equation}\label{Kraus-proof5}\begin{split}
\Delta_i =&\tr\left(F^\dag_i(\sigma(\overline{t}))_{\sigma(\overline{q})}\cdot\rho\cdot F_i(\sigma(\overline{t}))_{\sigma(\overline{q})}\right)\\ &-\sum\left\{|\tr\left(F^\dag_i(\sigma(\overline{t}))_{\sigma(\overline{q})}\cdot\rho^\prime\cdot F_i(\sigma(\overline{t}))_{\sigma(\overline{q})}\right)|(\sigma^\prime,\rho^\prime)\in\llbracket P\rrbracket\left(\sigma,\rho\right)|\right\}\\ 
=&\tr\left[F^\dag_i(\sigma(\overline{t}))_{\sigma(\overline{q})}\cdot\left(\rho-\sum\left\{|\rho^\prime|(\sigma^\prime,\rho^\prime)\in\llbracket P\rrbracket\left(\sigma,\rho\right)|\right\}\right)\cdot F_i(\sigma(\overline{t}))_{\sigma(\overline{q})}\right]\\ 
=&\tr\left[F_i(\sigma(\overline{t}))_{\sigma(\overline{q})}\cdot F^\dag_i(\sigma(\overline{t})_{\sigma(\overline{q})}\cdot\left(\rho-\sum\left\{|\rho^\prime|(\sigma^\prime,\rho^\prime)\in\llbracket P\rrbracket\left(\sigma,\rho\right)|\right\}\right)\right].
\end{split}\end{equation} It follows from condition (\ref{eq-pre-normal}) that $$\sum_iF_i(\sigma(\overline{t}))_{\sigma(\overline{q})}\cdot F^\dag_i(\sigma(\overline{t})_{\sigma(\overline{q})}= I.$$ 
Then plugging (\ref{Kraus-proof5}) into (\ref{Kraus-proof4}) yields: 
 \begin{equation}\label{Kraus-proof6}\begin{split}&\sum_i\mathit{NT}(P)\left(\sigma, F^\dag_i(\sigma(\overline{t}))_{\sigma(\overline{q})}\cdot\rho\cdot F_i(\sigma(\overline{t}))_{\sigma(\overline{q})}\right)\\ &=\sum_i\tr\left[F_i(\sigma(\overline{t}))_{\sigma(\overline{q})}\cdot F^\dag_i(\sigma(\overline{t})_{\sigma(\overline{q})}\cdot\left(\rho-\sum\left\{|\rho^\prime|(\sigma^\prime,\rho^\prime)\in\llbracket P\rrbracket\left(\sigma,\rho\right)|\right\}\right)\right]\\
 &=\tr\left[\left(\sum_iF_i(\sigma(\overline{t}))_{\sigma(\overline{q})}\cdot F^\dag_i(\sigma(\overline{t})_{\sigma(\overline{q})}\right)\cdot\left(\rho-\sum\left\{|\rho^\prime|(\sigma^\prime,\rho^\prime)\in\llbracket P\rrbracket\left(\sigma,\rho\right)|\right\}\right)\right]\\ 
 &=\tr\left(\rho-\sum\left\{|\rho^\prime|(\sigma^\prime,\rho^\prime)\in\llbracket P\rrbracket\left(\sigma,\rho\right)|\right\}\right)\\ 
&=\mathit{NT}(P)(\sigma,\rho).
\end{split}\end{equation} 

(2) Using (\ref{Kraus-proof3}) we obtain: \begin{equation}\label{Kraus-proof7+}\begin{split}&\mu\stackrel{\triangle}{=} \sum_i\sum\\ &\left\{|\tr(\sigma^\prime(B)\rho^{\prime\prime})|(\sigma^\prime,\rho^{\prime\prime})\in\llbracket P\rrbracket\left(\sigma, F^\dag_i(\sigma(\overline{t}))_{\sigma(\overline{q})}\cdot\rho\cdot F_i(\sigma(\overline{t}))_{\sigma(\overline{q})}\right), \sigma^\prime\models\psi_i\ \mbox{and}\ \sigma^\prime(B_i):WD|\right\}\\ &=\sum_i\sum\\ &\left\{|\tr\left[\sigma^\prime(B)\left(F^\dag_i(\sigma(\overline{t}))_{\sigma(\overline{q})}\cdot\rho^\prime\cdot F_i(\sigma(\overline{t}))_{\sigma(\overline{q})}\right)\right]|(\sigma^\prime,\rho^\prime)\in\llbracket P\rrbracket(\sigma, \rho), \sigma^\prime\models\psi_i\ \mbox{and}\ \sigma^\prime(B):WD|\right\}.\end{split}\end{equation} Let $\psi=\bigvee_i\psi_i$. Then using (\ref{Kraus-proof7+}), the assumption (\ref{accum2}), the assumptions: $$\sigma\models T(\overline{t})\propto F^\prime(\overline{t}),$$ $$\sigma\models\left(\forall i_1,i_2\right)\left(i_1\neq i_2\rightarrow\neg(\psi_{i_1}\wedge\psi_{i_2})\right),$$  and the commutativity $\tr(AB)=\tr(BA)$ for trace, we have:  
\begin{equation}\label{Kraus-proof7}\begin{split}
&\mu\leq \sum_i\sum\\ &\left\{|\tr\left[\sigma^\prime(B)\left(F^{\prime\dag}(\sigma(\overline{t}))_{\sigma(\overline{q})}\cdot\rho^\prime\cdot F^\prime(\sigma(\overline{t}))_{\sigma(\overline{q})}\right)\right]|(\sigma^\prime,\rho^\prime)\in\llbracket P\rrbracket(\sigma, \rho), \sigma^\prime\models\psi_i\ \mbox{and}\ \sigma^\prime(B):WD|\right\}\\ 
&\leq \sum\\ &\left\{|\tr\left[\sigma^\prime(B)\left(F^{\prime\dag}(\sigma(\overline{t}))_{\sigma(\overline{q})}\cdot\rho^\prime\cdot F^\prime(\sigma(\overline{t}))_{\sigma(\overline{q})}\right)\right]|(\sigma^\prime,\rho^\prime)\in\llbracket P\rrbracket(\sigma, \rho), \sigma^\prime\models\psi\ \mbox{and}\ \sigma^\prime(B):WD|\right\}\\ 
&=\sum\\ &\left\{|\tr\left[\left(F^\prime(\sigma(\overline{t}))_{\sigma(\overline{q})}\cdot\sigma^\prime(B)\cdot F^{^\prime\dag}(\sigma(\overline{t}))_{\sigma(\overline{q})}\right)\rho^\prime\right]|(\sigma^\prime,\rho^\prime)\in\llbracket P\rrbracket(\sigma, \rho), \sigma^\prime\models\psi\ \mbox{and}\ \sigma^\prime(B):WD|\right\}\\
 &=\sum \left\{|\tr\left[\sigma\left(F^\prime(\overline{t})[\overline{q}](B)\right)\rho^\prime\right]|(\sigma^\prime,\rho^\prime)\in\llbracket P\rrbracket(\sigma, \rho), \sigma^\prime\models\psi\ \mbox{and}\ \sigma^\prime(B):WD|\right\}. 
\end{split}\end{equation}

 Finally, by plugging (\ref{Kraus-proof6}) and (\ref{Kraus-proof7}) into (\ref{Kraus-proof2}), it is derived that 
 \begin{align*}&\tr\left[\sigma\left(F(\overline{t})[\overline{q}](\{A_i\})\right)\rho\right]\\ &\leq 
 \sum \left\{|\tr\left[\sigma\left(F^\prime(\overline{t})[\overline{q}](B)\right)\rho^\prime\right]|(\sigma^\prime,\rho^\prime)\in\llbracket P\rrbracket(\sigma, \rho), \sigma^\prime\models\psi\ \mbox{and}\ \sigma^\prime(B):WD|\right\}
+\mathit{NT}(P)(\sigma,\rho).
 \end{align*}
 Thus, (\ref{conv2}) is proved. 
 
{\vskip 3pt}

(\textbf{Rule-Accum2}): Let us use the same notations as in (\textbf{Rule-Accum2}). We assume that $\overline{q}\cap\mathit{qv}(P)=\emptyset$ and for every $i$, \begin{equation}\label{conv10}\models_\mathit{par}\{\varphi,A_i\}\ P\ \{\psi,B_i\};\end{equation} that is, for any input state $(\sigma,\rho)\in\Omega$, whenever $\sigma\models\varphi$ and $\sigma(A_i)$ is well-defined, then: 
\begin{equation}\label{Kraus-proof10}\begin{split}\tr(\sigma(A_i)\rho)\leq\sum &\left\{|\tr(\sigma^\prime(B_i)\rho^\prime)|(\sigma^\prime,\rho^\prime)\in\llbracket P\rrbracket(\sigma,\rho),\sigma^\prime\models\psi\ \mbox{and}\ \sigma^\prime(B_i):WD\right\}\\ &\qquad + \mathit{NT}(P)(\sigma,\rho).
\end{split}\end{equation}
We want to prove \begin{equation}\label{conv20}\models_\mathit{par}\left\{\varphi,F(\overline{t})[\overline{q}](\{A_i\})\right\}\ P\ \left\{\psi,F(\overline{t})[\overline{q}](\{B_i\})\right\}.\end{equation}
For any input state $(\sigma,\rho)\in\Omega$, if $\sigma(\varphi)$ and $\sigma\left(F(\overline{t})[\overline{q}](\{A_i\})\right)$ is well-defined, then for all $i$, $\sigma(A_i)$ is well-defined, and $\sigma\models\mathit{Dist}(\overline{q})$. Assume that the Kraus operator symbol $F$ is interpreted as $F=\left\{F_i(\overline{a})\right\}$. Then by the assumption (\ref{Kraus-proof10}) we can derive: 
\begin{equation}\label{Kraus-proof20}\begin{split}&\tr\left[\sigma\left(F(\overline{t})[\overline{q}](\{A_i\})\right)\rho\right]=\tr\left[\left(\sum_iF_i(\sigma(\overline{t}))_{\sigma(\overline{q})}\cdot\sigma(A_i)\cdot F^\dag_i(\sigma(\overline{t}))_{\sigma(\overline{q})}\right)\rho\right]\\
&\leq\sum_i\Gamma_i\end{split}\end{equation} in a way similar to (\ref{Kraus-proof2}), 
where: 
\begin{align*}&\Gamma_i=\sum\\ &\left\{|\tr(\sigma^\prime(B_i)\rho^\prime)|(\sigma^\prime,\rho^\prime)\in\llbracket P\rrbracket\left(\sigma, F^\dag_i(\sigma(\overline{t}))_{\sigma(\overline{q})}\cdot\rho\cdot F_i(\sigma(\overline{t}))_{\sigma(\overline{q})}\right), \sigma^\prime\models\psi\ \mbox{and}\ \sigma^\prime(B_i):WD|\right\}\\ & +\mathit{NT}(P)\left(\sigma, F^\dag_i(\sigma(\overline{t}))_{\sigma(\overline{q})}\cdot\rho\cdot F_i(\sigma(\overline{t}))_{\sigma(\overline{q})}\right).\end{align*}

Now using (\ref{Kraus-proof3}) we obtain: \begin{equation}\label{Kraus-proof70}\begin{split}&\sum_i\sum\\ &\left\{|\tr(\sigma^\prime(B_i)\rho^{\prime\prime})|(\sigma^\prime,\rho^{\prime\prime})\in\llbracket P\rrbracket\left(\sigma, F^\dag_i(\sigma(\overline{t}))_{\sigma(\overline{q})}\cdot\rho\cdot F_i(\sigma(\overline{t}))_{\sigma(\overline{q})}\right), \sigma^\prime\models\psi\ \mbox{and}\ \sigma^\prime(B_i):WD|\right\}\\  
&=\sum_i\sum\\ &\left\{|\tr\left[\sigma^\prime(B_i)\left(F^\dag_i(\sigma(\overline{t}))_{\sigma(\overline{q})}\cdot\rho^\prime\cdot F_i(\sigma(\overline{t}))_{\sigma(\overline{q})}\right)\right]|(\sigma^\prime,\rho^\prime)\in\llbracket P\rrbracket(\sigma, \rho), \sigma^\prime\models\psi\ \mbox{and}\ \sigma^\prime(B_i):WD|\right\}\\ 
&=\sum\\ &\left\{|\sum_i\tr\left[\sigma^\prime(B_i)\left(F^\dag_i(\sigma(\overline{t}))_{\sigma(\overline{q})}\cdot\rho^\prime\cdot F_i(\sigma(\overline{t}))_{\sigma(\overline{q})}\right)\right]|(\sigma^\prime,\rho^\prime)\in\llbracket P\rrbracket(\sigma, \rho), \sigma^\prime\models\psi\ \mbox{and}\ \sigma^\prime(B_i):WD|\right\}\\ 
&=\sum\\ &\left\{|\sum_i\tr\left[\left(F_i(\sigma(\overline{t}))_{\sigma(\overline{q})}\cdot\sigma^\prime(B_i)F^\dag_i(\sigma(\overline{t})_{\sigma(\overline{q})}\right)\rho^\prime\right]|(\sigma^\prime,\rho^\prime)\in\llbracket P\rrbracket(\sigma, \rho), \sigma^\prime\models\psi\ \mbox{and}\ \sigma^\prime(B_i):WD|\right\}\\
&=\sum\\ &\left\{|\tr\left[\sum_i\left(F_i(\sigma(\overline{t}))_{\sigma(\overline{q})}\cdot\sigma^\prime(B_i)F^\dag_i(\sigma(\overline{t})_{\sigma(\overline{q})}\right)\rho^\prime\right]|(\sigma^\prime,\rho^\prime)\in\llbracket P\rrbracket(\sigma, \rho), \sigma^\prime\models\psi\ \mbox{and}\ \sigma^\prime(B_i):WD|\right\}\\ 
&=\sum \left\{|\tr\left[\sigma\left(F(\overline{t})[\overline{q}](\{B_i\})\right)\rho^\prime\right]|(\sigma^\prime,\rho^\prime)\in\llbracket P\rrbracket(\sigma, \rho), \sigma^\prime\models\psi\ \mbox{and}\ \sigma^\prime(B_i):WD|\right\}. 
\end{split}\end{equation}

 Finally, by plugging (\ref{Kraus-proof6}) and (\ref{Kraus-proof70}) into (\ref{Kraus-proof20}), it is derived that 
 \begin{align*}&\tr\left[\sigma\left(F(\overline{t})[\overline{q}](\{A_i\})\right)\rho\right] \leq 
 \sum\\ &\left\{|\tr\left[\sigma\left(F(\overline{t})[\overline{q}](\{B_i\})\right)\rho^\prime\right]|(\sigma^\prime,\rho^\prime)\in\llbracket P\rrbracket(\sigma, \rho), \sigma^\prime\models\psi\ \mbox{and}\ \sigma^\prime(B_i):WD|\right\}. 
+\mathit{NT}(P)(\sigma,\rho).
 \end{align*}
 Thus, (\ref{conv20}) is proved. 

\subsubsection{The Case of Total Correctness} Now we turn to prove the soundness of proof system $\mathit{QHL}^+_\mathit{tot}$. The validity of all axioms and rules except (Rule-Loop-tot) in the sense of total correctness can be proved in a way similar to (but easier than) the case of partial correctness. So, we only prove the validity of (Rule-Loop-tot) here. We use the same notation as in the validity proof of (Rule-Loop-par). Assume \begin{align}
\label{loop-s1}&\models_\mathit{tot}\{\varphi\wedge b,A\}\ P\ \{\varphi,A\},\\ \label{loop-s2} &\models_\mathit{tot}\{\varphi\wedge b\wedge t=z,I\}\ P\ \{t<z,I\},\\ \label{loop-s3} &\models\varphi\rightarrow t\geq 0
\end{align} 
where $z$ is an integer variable with $z\notin\mathit{free}(\varphi)\cup\mathit{var}(b)\cup\mathit{var}(t)\cup\mathit{cv}(P)$, and $I$ stands for the identity quantum predicate.  
We want to prove: 
\begin{equation}\label{loop-s4}\models_\mathit{tot}\{\varphi,A\}\ \mathbf{while}\ \{\varphi\wedge \neg b,A\};\end{equation}
that is, for any input state $(\sigma,\rho)\in\Omega$, if $\sigma\models\varphi$ and $\sigma(A)$ is well-defined, then 
\begin{equation}\label{ineq-loop-00}\tr(\sigma(A)\rho)\leq\sum\left\{|\tr(\sigma^\prime(A)\rho^\prime))|(\sigma^\prime,\rho^\prime)\in\llbracket\mathbf{while}\rrbracket(\sigma,\rho), \sigma^\prime\models\varphi\wedge\neg b\ \mbox{and}\ \sigma^\prime(A):WD|\right\}.\end{equation}

First, by the assumption (\ref{loop-s1}), we are able to prove the following claim in a way similar to the proof of claim (\ref{ineq-loop-0}):

\textbf{\textit{Claim 1}}: For any integer $n\geq 1$, 
\begin{equation}\label{ineq-loop-10}\begin{split}\tr(\sigma(A)\rho)&\leq\sum\left\{|\tr(\sigma^\prime(A)\rho^\prime)|(\sigma^\prime,\rho^\prime)\in\mathcal{P}_n(\sigma,\rho), \sigma^\prime\models\varphi\ \mbox{and}\ \sigma^\prime(A):WD|\right\}+\Delta^\ast_n
\end{split}\end{equation} where: \begin{align}\label{def-Delta}\Delta^\ast_n &=\sum\left\{|\tr(\sigma^\prime(A)\rho^\prime))\mid (\sigma^\prime,\rho^\prime)\in \mathcal{Q}^\ast_n(\sigma,\rho)\ \mbox{and}\ \sigma^\prime(A):WD|\right\},\\
\label{def-Q}\mathcal{Q}^\ast_n(\sigma,\rho)&=\left\{|(\sigma^\prime,\rho^\prime)\mid (\sigma,\rho)=(\sigma_0,\rho_0)\stackrel{P}{\Rightarrow}...\stackrel{P}{\Rightarrow}(\sigma_n,\rho_n)=(\sigma^\prime,\rho^\prime)\ \mbox{and}\ \sigma_i\models \varphi\wedge b\ (0\leq i\leq n)|\right\}.
\end{align}
Let us take the limit of $n\rightarrow\infty$ in the right-hand side of (\ref{ineq-loop-10}). Then it follows from (\ref{ineq-loop-14}) that \begin{equation}\label{eq-4z}\begin{split}
\tr(\sigma(A)\rho)\leq &\sum\left\{|\tr(\sigma^\prime(A)\rho^\prime)|(\sigma^\prime,\rho^\prime)\in\llbracket\mathbf{while}\rrbracket(\sigma,\rho), \sigma^\prime\models\varphi\wedge\neg b\ \mbox{and}\ \sigma^\prime(A):WD|\right\}\\ &\qquad\qquad+\lim_{n\rightarrow\infty}\Delta^\ast_n
\end{split}\end{equation} Thus, it suffices to show that $\lim_{n\rightarrow\infty}\Delta^\ast_n=0$. We do this by refutation. Suppose that $\lim_{n\rightarrow\infty}\Delta^\ast_n>0$. We derive a contradiction in four steps:

\textbf{\textit{Step 1}}: If $\tr(\sigma(A)\rho)=0$, then (\ref{ineq-loop-00}) is obviously true. So, we now assume $\tr(\sigma(A)\rho)>0$ and use the assumption (\ref{loop-s2}) to prove the following:

\textbf{\textit{Claim 2}}: If $(\sigma,\rho)\stackrel{P}{\Rightarrow}(\sigma^\prime,\rho^\prime)$ and $\rho^\prime\neq 0$, then $\sigma^\prime(t)<\sigma(t)$.  

To prove this claim, by the assumptions, we have $\sigma\models\varphi\wedge b$ and $\sigma(I):WD$. We can introduce a fresh integer variable $z$ such that \begin{equation}\label{z-cond}z\notin\mathit{free}(\varphi)\cup\mathit{var}(b)\cup\mathit{var}(t)\cup\mathit{cv}(P),\end{equation} and set $\sigma(z)=\sigma(t)$. Thus, it holds that $\sigma\models\varphi\wedge b\wedge t=z$. By the assumption (\ref{loop-s2}), we obtain:   
\begin{equation}\label{def-R}\begin{split}0<\tr(\rho)&=\tr(\sigma(I)\rho)\\ &\leq\sum\left\{|\tr(\sigma^\prime(I)\rho^\prime)|(\sigma,\rho)\stackrel{P}{\Rightarrow}(\sigma^\prime,\rho^\prime)\ \mbox{and}\ \sigma^\prime\models t<z|\right\}\\ 
&=\sum\left\{|\tr(\rho^\prime)|(\sigma,\rho)\stackrel{P}{\Rightarrow}(\sigma^\prime,\rho^\prime)\ \mbox{and}\ \sigma^\prime\models t<z|\right\}. 
\end{split}\end{equation} 
For any $(\sigma^\prime,\rho^\prime)$ in the right-hand side of (\ref{def-R}), by the condition (\ref{z-cond}) we know that the program $P$ does not change the value of $z$, and thus $\sigma^\prime(z)=\sigma(z)$. On the other hand, $\sigma^\prime\models t<z$. Then $\sigma^\prime(t)<\sigma^\prime(z)=\sigma(z)=\sigma(t)$. Therefore, we obtain:
\begin{equation}\label{zz-cond}\tr(\rho)\leq  \sum\left\{|\tr(\rho^\prime)|(\sigma,\rho)\stackrel{P}{\Rightarrow}(\sigma^\prime,\rho^\prime)\ \mbox{and}\ \sigma^\prime(t)<\sigma(t)|\right\}
\end{equation} from (\ref{def-R}). 

It follows from Proposition \ref{trace-decrease} that 
\begin{equation}\label{zzz-cond}\sum\left\{|\tr(\rho^\prime)|(\sigma,\rho)\stackrel{P}{\Rightarrow}(\sigma^\prime,\rho^\prime)|\right\}\leq\tr(\rho).\end{equation}
By combining (\ref{zz-cond}) and (\ref{zzz-cond}) we have: 
\begin{equation}\label{zzzz-cond}\sum\left\{|\tr(\rho^\prime)|(\sigma,\rho)\stackrel{P}{\Rightarrow}(\sigma^\prime,\rho^\prime)|\right\}\leq\sum\left\{|\tr(\rho^\prime)|(\sigma,\rho)\stackrel{P}{\Rightarrow}(\sigma^\prime,\rho^\prime)\ \mbox{and}\ \sigma^\prime(t)<\sigma(t)|\right\}.\end{equation} 
If \textbf{Claim 2} is not true; that is, there exists $(\sigma^\prime_0,\rho^\prime_0)$ such that $(\sigma,\rho)\stackrel{P}{\Rightarrow}(\sigma^\prime_0,\rho^\prime_0)$, $\tr(\rho^\prime_0)>0$ and $\sigma^\prime_0(t)\geq\sigma(t)$, 
then it holds that \begin{align*}\mbox{The right-side of}\ (\ref{zzzz-cond})&\leq \sum\left\{|\tr(\rho^\prime)|(\sigma,\rho)\stackrel{P}{\Rightarrow}(\sigma^\prime,\rho^\prime)|\right\}-\tr(\rho^\prime)\\ 
&< \mbox{the left-side of}\ (\ref{zzzz-cond}),\end{align*} and inequality (\ref{zzzz-cond}) is violated, a contradiction. Thus, \textbf{Claim 2} is proved. 

\textbf{\textit{Step 2}}: We construct a tree using \textbf{Claim 2}. The root of the tree is labelled by $(\sigma_0,\rho_0)=(\sigma,\rho)$. 
Let 
\begin{align*}\mathcal{R}_1=\left\{|(\sigma_1,\rho_1)|(\sigma_1,\rho_1)\in\mathcal{Q}_1^\ast(\sigma,\rho)\ \mbox{and}\ \rho_1\neq 0|\right\}.
\end{align*} 

From the assumption $\lim_{n\rightarrow\infty}\Delta^\ast_n>0$, we have:

\textbf{\textit{Claim 3}}: There exists an integer $n_0\geq 0$ such that $\Delta^\ast_n>0$ for all $n\geq n_0$. 

Then from \textbf{Claim 3} and the defining equations (\ref{def-Delta}) and (\ref{def-Q}) of $\Delta^\ast_n$ and $\mathcal{Q}^\ast_n(\sigma,\rho)$, it is easy to see that $\mathcal{R}_1\neq\emptyset$. We expand the tree by letting each $(\sigma_1,\rho_1)\in\mathcal{R}_1$ with $\sigma_1\models\varphi\wedge b$ as an immediate child of $(\sigma_0,\rho_0)$. It follows from \textbf{Claim 2} that $\sigma_1(t)<\sigma_0(t)$. 

Next, for each newly added node $(\sigma_1,\rho_1)$, we repeat the above process and generate its immediate childs $(\sigma_2,\rho_2)$. In this way, we expand the tree step by step, and the tree can be constructed as desired. We need to remember that in this tree:\begin{enumerate}
\item[(i)] for any node $(\sigma_i,\rho_i)$, it holds that $\sigma_n\models\varphi\wedge b$; and \item[(ii)] if $(\sigma_{i+1},\rho_{i+1})$ is a child of $(\sigma_i,\rho_i)$, then $\sigma_{i+1}(t)<\sigma_i(t)$.
\end{enumerate}

\textbf{\textit{Step 3}}: For any $n\geq n_0$, using \textbf{Claim 3} and the defining equations (\ref{def-Delta}) and (\ref{def-Q}) of $\Delta^\ast_n$ and $\mathcal{Q}^\ast_n(\sigma,\rho)$, we can assert that there exists $(\sigma^\prime,\rho^\prime)\in\mathcal{Q}_n(\sigma,\rho)$ such that $\sigma^\prime\models\varphi$, $\sigma^\prime(A):WD$ and $\tr(\sigma^\prime(A)\rho^\prime)>0$. Consequently, there exists an execution path: \begin{equation}\label{eq-path}(\sigma,\rho)=(\sigma_0,\rho_0)\stackrel{P}{\Rightarrow} ...\stackrel{P}{\Rightarrow}(\sigma_n,\rho_n)=(\sigma^\prime,\rho^\prime)\end{equation} such that $\sigma_i\models\varphi\wedge b\ (0\leq i\leq n).$ 
Since $\tr(\sigma^\prime(A)\rho^\prime)>0$, it must hold that $\rho_i\neq 0$ for $i=0,1,...,n$. By \textbf{Claim 2}, we have $\sigma_{i+1}(t)<\sigma_i(t)$ for $i=0,1,...,n-1$. 
Therefore, execution path (\ref{eq-path}) is a path in the tree that we constructed in Step 2. 

As a conclusion of this step, since for any $n\geq n_0$, the tree contains a path of length $n+1$, it must be an infinite tree.

\textbf{\textit{Step 4}}: Note that we assume that the classical (outcome) type $T$ of any measurement $M$ is finite. Then from the transition rules given in Definition \ref{def-operational}, we know that the tree constructed in Step 2 must be finitely branching. On the other hand, we proved it is an infinite tree in Step 3. Thus, by K\"{o}nig's lemma, the tree has an infinite path     
$$(\sigma,\rho)=(\sigma_0,\rho_0)\stackrel{P}{\Rightarrow} ...\stackrel{P}{\Rightarrow}(\sigma_n,\rho_n)\stackrel{P}{\Rightarrow}(\sigma_{n+1},\rho_{n+1})\stackrel{P}{\Rightarrow} ....$$
Consequently, by the fact (ii) at the end of Step 2, we obtain an infinite decreasing sequence $$\sigma(t)=\sigma_0(t)>...>\sigma_n(t)>\sigma_{n+1}(t)>...$$ of integers. On the other hand, by the fact (i) at the end of Step 2, we have $\sigma_n\models\varphi$ for all $n\geq 0$.  
Invoking the assumption (\ref{loop-s3}), i.e. $\models\varphi\rightarrow t\geq 0$, we then obtain $\sigma_n(t)\geq 0$ for all $n\geq 0$. This is a contradiction. 

Therefore, we have proved that $\lim_{n\rightarrow\infty}\Delta_n^\ast=0$. Substituting this into (\ref{eq-4z}), we obtain (\ref{ineq-loop-00}) and the proof is completed.

\subsection{Verification Example: Quantum Fourier Transform}\label{app-QFT}

In this section, we provide the proof details omitted in Example \ref{QFT-proof}.

(1) Proof of correctness formula (\ref{proof-reverse}): First, we notice that \begin{align}\label{x-reverse}\begin{cases}\mathit{Reverse}[q[1:1]]=I\ (\mbox{the identity operator on})\ \hs_2,
\\ \mathit{Reverse}[q[1:n]](|\Phi\rangle\otimes|\psi\rangle)=|\psi\rangle\otimes\mathit{Reverse}[q[1:n]]|\Phi\rangle
\end{cases}\end{align} for any $\psi\rangle\in\hs_2$ and $|\Phi\rangle\in\hs_2^{\otimes(n-1)}$, and $$\mathit{Reverse}^\dag[q[1:n]]=\mathit{Reverse}[q[1:n]].$$ Then we can prove 
\begin{equation}\label{xx-reverse}\mathit{Reverse}^\dag[q[1:n]]|\mathit{QFT}(j,k:l)\rangle=|\mathit{QFT}^\ast(j,n-l+1:n-k+1)\rangle\end{equation} by induction on $l-k$. Indeed, the basis case of $k=l$ is immediate from the first equality in (\ref{x-reverse}). The induction step is then derived as follows. By the induction hypothesis and the second equality in (\ref{x-reverse}) we obtain:\begin{align*}
&\mathit{Reverse}^\dag[q[1:n]]|\mathit{QFT}(j,k:l)\rangle=\mathit{Reverse}[[1:n]]|\mathit{QFT}(j,k:l)\rangle\\ 
&=\mathit{Reverse}[[q[1:n]]\left(|\mathit{QFT}(j,k:l-1)\rangle\otimes\frac{1}{\sqrt{2}}\left(|0\rangle+e^{2\pi 0.j[l-k-1:l]}|1\rangle\right)\right)\\ &=\frac{1}{\sqrt{2}}\left(|0\rangle+e^{2\pi 0.j[l-k-1:l]}|1\rangle\right)\otimes \mathit{Reverse}[[q[1:n]]|\mathit{QFT}(j,k:l-1)\rangle
\\ &=\frac{1}{\sqrt{2}}\left(|0\rangle+e^{2\pi 0.j[l-k-1:l]}|1\rangle\right)\otimes |\mathit{QFT}^\ast(j,n-l+2:n-k+1)\rangle\\
&=|\mathit{QFT}^\ast(j,n-l+1:n-k+1)\rangle.
\end{align*} Here, the fourth equality comes from the induction hypothesis for $n-1$. Therefore, it holds that 
$$B^\ast(j,1:n)=\mathit{Reverse}^\dag[q[1:n]] B(j,1:n)\mathit{Reverse}[q[1:n]]$$ and thus (\ref{proof-reverse}) is established by (Axiom-Uni). 

(2) Proof of correctness formula (\ref{inductxx-QFT}): Let $A(j,m:n)=1$ (constant) when $m>n$. We want to prove:
\begin{equation}\label{CR-proof}\left\{m\leq n,|\varphi\rangle_{q[m]}\langle\varphi|\otimes A(j,m+1:n)\right\}\ \mathit{CR}[q[m:n]]\ \left\{m\leq n,|\psi\rangle_{q[m]}\langle\psi|\otimes A(j,m+1:n)\right\}\end{equation} by induction on the length $n-m$ of the arrays. 
If $m=n$, then $|\varphi\rangle=|\psi\rangle$ and it is obvious that (\ref{CR-proof}) holds. By the induction hypothesis, we have: 
\begin{equation*}\begin{split}&\left\{m\leq n-1,|\varphi\rangle_{q[m]}\langle\varphi|\otimes A(j,m+1:n-1)\right\}\ \mathit{CR}[q[m:n-1]]\\ &\qquad\qquad\qquad\qquad\qquad\qquad\left\{m\leq n-1,|\psi^\prime\rangle_{q[m]}\langle\psi^\prime|\otimes A(j,m+1:n-1)\right\}\end{split}\end{equation*}
where $$|\psi^\prime=\frac{1}{\sqrt{2}}(|0\rangle+e^{2\pi i 0.j[m:n-1]}|1\rangle).$$ Note that $q[n]\notin q[m:n-1]$ and 
\begin{equation}\label{CR-proof2}A(j,m+1:n)=A(j,m+1:n-1)\otimes |j[n]\rangle_{q[n]}\langle j[n]|.\end{equation} Then it holds that 
\begin{equation}\label{CR-proof1}\begin{split}\left\{m\leq n,|\varphi\rangle_{q[m]}\langle\varphi|\otimes A(j,m+1:n)\right\}\ &\mathit{CR}[q[m:n-1]]\\ &\left\{m\leq n,|\psi^\prime\rangle_{q[m]}\langle\psi^\prime|\otimes A(j,m+1:n)\right\}.\end{split}\end{equation}

On the other hand, by (Axiom-Uni) we have: 
\begin{align*}\left\{m\leq n,|\psi^\prime\rangle_{q[m]}\langle\psi^\prime|\otimes |j[n]\rangle_{q[n]}\langle j[n]|\right\}\ &C(R_{n-m+1})[q_n,q_m]\\ &\left\{m\leq n,|\psi\rangle_{q[m]}\langle\psi| \otimes |j[n]\rangle_{q[n]}\langle j[n]|\right\}.
\end{align*} Note that $q[m]$ and $q[n]$ do not occur in $A(j,m+1:n-1)$. Then by (\ref{CR-proof2}) we obtain:
\begin{equation}\label{CR-proof3}\begin{split}\left\{m\leq n,|\psi^\prime\rangle_{q[m]}\langle\psi^\prime|\otimes A(j,m+1:n)\right\}\ &C(R_{n-m+1})[q_n,q_m]\\ &\left\{m\leq n,|\psi\rangle_{q[m]}\langle\psi| \otimes A(j,m+1:n)\right\}.\end{split}\end{equation}
Now we can use (Rule-Seq) to derive:
\begin{equation*}\begin{split}\left\{m\leq n,|\varphi\rangle_{q[m]}\langle\varphi|\otimes A(j,m+1:n)\right\}\ &\mathit{CR}[q[m:n-1]]; C(R_{n-m+1})[q_n,q_m]\\ &\left\{m\leq n,|\psi\rangle_{q[m]}\langle\psi| \otimes A(j,m+1:n)\right\}.\end{split}\end{equation*}
from (\ref{CR-proof1}) and (\ref{CR-proof3}). Thus, by the inductive definition of $\mathit{CR}[q[m:n]]$, we know that (\ref{CR-proof}) holds. 

\end{document}